%% file: main.tex
\documentclass[10pt,conference]{IEEEtran}

\input{packages}

\usepackage{ushort}

\input{macros.tex}

\theoremstyle{plain}
\newtheorem{lemma}{Lemma}[section]
\newtheorem{theorem}[lemma]{Theorem}
\newtheorem{corollary}[lemma]{Corollary}
\theoremstyle{definition}
\newtheorem{definition}{Definition}[section]
\theoremstyle{remark}
\newtheorem{example}{Example}[section]

\begin{document}

\allowdisplaybreaks
\pagestyle{plain}

\title{Conditioning in Probabilistic Programming}

\author{
\IEEEauthorblockN{Nils Jansen\\ and Benjamin Lucien Kaminski\\  and Joost-Pieter Katoen\\ and Federico Olmedo}
\IEEEauthorblockA{RWTH Aachen University\\ Aachen, Germany}

\and

\IEEEauthorblockN{Friedrich Gretz\\ and Annabelle McIver}
\IEEEauthorblockA{Macquarie University\\ Sydney, Australia}
}


\maketitle
\begin{abstract}
We investigate the semantic intricacies of conditioning, a main feature in probabilistic programming.
We provide a weakest (liberal) pre--condition (\textsf{w}(\textsf{l})\textsf{p}) semantics for the elementary probabilistic 
programming language \pGCL extended with conditioning.
We prove that quantitative weakest (liberal) pre--conditions coincide with conditional (liberal) 
expected rewards in Markov chains and show that semantically conditioning is a truly conservative 
extension. 
We present two program transformations which entirely eliminate conditioning from any
program and prove their correctness using the \textsf{w}(\textsf{l})\textsf{p}--semantics.
Finally, we show how the \textsf{w}(\textsf{l})\textsf{p}--semantics can be used to determine conditional probabilities in a 
parametric anonymity protocol and show that an inductive \textsf{w}(\textsf{l})\textsf{p}--semantics for conditioning in 
non--deterministic probabilistic programs cannot exist.
\end{abstract}


\section{Introduction}
\input{introduction}

\section{Preliminaries}\label{sec:prelim}
\input{prelim}

%
\section{Conditional \pGCL}\label{sec:conditional_pgcl}
\input{cpGCL}
\section{Operational Semantics for \cpGCL}\label{sec:ope}
\input{operational}
\section{Denotational Semantics for \cfpGCL}\label{sec:denotational}
\input{denotational}

\section{Applications}\label{sec:app}
\input{applications}

%
\section{Denotational Semantics for Full \cpGCL}\label{sec:nd}
\input{nondet}
%
\section{Conclusion and Future Work}
\input{conclusion}

\section*{Acknowledgment}
This work was supported by the Excellence Initiative of the German federal and state government. Moreover, we would like to thank Pedro d'Argenio and Tahiry Rabehaja for the valuable discussions preceding this paper.

\bibliographystyle{IEEEtran}
\bibliography{IEEEabrv,bibliography}
\clearpage
\appendix
\section{Appendix to ``Conditioning in Probabilistic Programming"} \label{sec:appendix}

\input{appendix}
\end{document}

%% file: packages.tex
\usepackage{mathtools,amssymb,amsfonts,amsthm}    
\usepackage{nicefrac}
\usepackage{stmaryrd}                      
\usepackage[]{todonotes}                     
\usepackage{enumitem}                      
\usepackage{xspace}                        
\usepackage{tikz}
\usetikzlibrary{arrows,decorations.pathmorphing,positioning,fit,trees,shapes,shadows,automata,calc}
\usepackage{pgfplots}
\usepackage{leftidx}
\usepackage{framed}
\usepackage{needspace}

\usepackage{listings}          

%% file: macros.tex
\newcommand{\ie}{i.e.\@\xspace}
\newcommand{\wrt}{w.r.t.\@\xspace}
\newcommand{\eg}{e.g.\@\xspace}
\newcommand{\cf}{cf.\@\xspace}

\newcommand{\Reals}{\ensuremath{\mathbb{R}}\xspace}    
\newcommand{\Integers}{\ensuremath{\mathbb{Z}}\xspace}    
\newcommand{\Naturals}{\ensuremath{\mathbb{N}}}
\newcommand{\Ireal}{[0,\, 1]\subseteq\mathbb{R}}  
\newcommand{\Ex}{\ensuremath{\mathbb{E}}\xspace}        
\newcommand{\BEx}{\ensuremath{\mathbb{E}_{\scriptscriptstyle \leq 1}}\xspace}        
\newcommand{\EV}{\ensuremath{\mathbf{E}}}

\newcommand{\true}{\ensuremath{\mathsf{true}}\xspace}
\newcommand{\false}{\ensuremath{\mathsf{false}}\xspace}
\newcommand{\ExpZero}{\CteFun{0}\xspace}
\newcommand{\ExpOne}{\CteFun{1}\xspace}

\newcommand{\eqdef}{\triangleq}                              
\DeclareMathOperator*{\lfp}{\boldsymbol \mu}            
\DeclareMathOperator*{\gfp}{\boldsymbol \nu}            
\newcommand{\suchthat}{\mydot}

\newcommand{\mydot}{\ensuremath                         
  \raisebox{0pt}{${\scriptscriptstyle \bullet}~$}}
\newcommand{\CteFun}[1]{\ensuremath{\boldsymbol{#1}}\xspace}    

\newcommand{\Finally}{\lozenge \, }

\renewcommand{\wp}{{\sf wp}\xspace}
\newcommand{\wlp}{{\sf wlp}\xspace}
\newcommand{\wllp}{{\textsf{w}(\textsf{l})\textsf{p}}\xspace}
\newcommand{\cwp}{{\sf cwp}\xspace}
\newcommand{\cwlp}{{\sf cwlp}\xspace}
\newcommand{\cwllp}{{\textsf{cw}(\textsf{l})\textsf{p}}\xspace}
\newcommand{\qcwp}{\ensuremath{\mathsf{\ushortw{cw}p}}\xspace}
\newcommand{\qcwlp}{\ensuremath{\mathsf{\ushortw{cwl}p}}\xspace}
\newcommand{\sem}[1]{\llbracket #1 \rrbracket}
\newcommand{\ToExp}[1]{\ensuremath{ \chi_{#1} }\xspace}    

\newcommand{\GCL}{\textnormal{\textsf{GCL}}\xspace}
\newcommand{\pGCL}{\textnormal{\textsf{pGCL}}\xspace}
\newcommand{\fpGCL}{\textnormal{$\textsf{pGCL}^{\boxtimes}$}\xspace}
\newcommand{\cpGCL}{\textnormal{\textsf{cpGCL}}\xspace}
\newcommand{\cfpGCL}{\textnormal{$\textsf{cpGCL}^{\boxtimes}$}\xspace}
\newcommand{\Cmd}{\ensuremath{\mathcal{P}}\xspace}        
\newcommand{\Var}{\ensuremath{\mathcal{V}}\xspace}        

\newcommand{\State}{\ensuremath{\mathbb{S}}\xspace}      
\newcommand{\Loop}{\textsf{loop}\xspace} 
\newcommand{\Body}{\textsf{body}\xspace}
\newcommand{\Init}{\textsf{init}\xspace}
\newcommand{\counter}{\textit{cntr}\xspace}
\newcommand{\intercepted}{\textit{int}\xspace}
\newcommand{\delivered}{\textit{del}\xspace}
\newcommand{\continue}{\textit{rerun}\xspace}
\newcommand{\Skip}{{\tt skip}\xspace}
\newcommand{\Abort}{{\tt abort}\xspace}
\newcommand{\Ass}[2]{#1 \coloneqq #2}
\newcommand{\Observe}{\textnormal{\texttt{observe}}\xspace}
\newcommand{\PChoice}[3]{\{#1\} \, \left[#2 \right] \, \{#3\}}
\newcommand{\NDChoice}[2]{\{#1\} \, \Box \, \{#2\}}
\newcommand{\If}{{\tt if}\xspace}

\newcommand{\While}{{\tt while}\xspace}

\newcommand{\Ite}{{\tt ite}\xspace}
\newcommand{\Cond}[3]{{\Ite} \, (#1) \allowbreak\, \{#2\} \allowbreak\, \{#3\}}
\newcommand{\WhileDo}[2]{{\While} \,\allowbreak (#1) \,\allowbreak \{#2\}}
\newcommand{\WhileDok}[3]{{\While}^{{}< #1} \,\allowbreak (#2) \,\allowbreak \{#3\}}

\newcommand{\Tr}{\ensuremath{\mathcal{T}}}    
\newcommand{\To}{\rightarrow}                          
\newcommand{\subst}[2]{[#1 / #2]}
\newcommand{\ourunderscore}{\rule{1.5ex}{0.4pt}}

\newcommand{\Par}{\ensuremath{{V}}}

\newcommand{\poly}[1]{\ensuremath{\Integers_{#1}}}
\newcommand{\Distr}{\mathit{Distr}}
\newcommand{\rdtmc}{\ensuremath{\mathcal{R}}}
\newcommand{\Rdtmc}[3]{\ensuremath{\rdtmc_{#1}^{#2}\llbracket #3\rrbracket}}
\newcommand{\pdtmc}{\ensuremath{\mathcal{P}}}
\newcommand{\sinit}{\ensuremath{s_{I}}}
\newcommand{\AP}{\ensuremath{\mathit{AP}}}
\renewcommand{\Pr}{\ensuremath{\textnormal{Pr}}}
\newcommand{\Prr}[1]{\ensuremath{\Pr^{#1}}}
\newcommand{\sched}{\ensuremath{\mathfrak{S}}}
\newcommand{\Sched}{\ensuremath{\mathit{Sched}}}

\newcommand{\rmdp}{\ensuremath{\mathfrak{R}}}
\newcommand{\Rmdp}[3]{\ensuremath{\rmdp_{#1}^{#2}\llbracket #3\rrbracket}}
\newcommand{\Act}{\ensuremath{\mathit{Act}}}
\newcommand{\act}{\ensuremath{\alpha}}
\newcommand{\pmdp}{\ensuremath{\mathcal{P}}}

\newcommand{\pathset}{\mathsf{Paths}}
\newcommand{\paths}[2]{\pathset(#1,\allowbreak\,#2)}
\newcommand{\pathsfin}{\pathset_{\mathit{fin}}}
\newcommand{\pathsinf}{\pathset_{\mathit{inf}}}
\newcommand{\path}{\pi}
\newcommand{\fpath}{\hat\path}
\newcommand{\last}[1]{\mathit{last}(#1)}
\newcommand{\bad}{\ensuremath{\lightning}}
\newcommand{\undesired}{\ensuremath{\langle \bad \rangle}\xspace}
\newcommand{\sinklabel}{\mathpzc{sink}\xspace}
\newcommand{\sink}{\ensuremath{\langle \sinklabel \rangle}\xspace}
\newcommand{\term}{\ensuremath{{\downarrow}}}
\newcommand{\exit}{\ensuremath{\checkmark}}
\newcommand{\ExpRew}[2]{\ensuremath{\textnormal{\textsf{ExpRew}}^{#1}\left(#2\right)}}
\newcommand{\LExpRew}[2]{\ensuremath{\textnormal{\textsf{LExpRew}}^{#1}\left(#2\right)}}
\newcommand{\CExpRew}[3]{\ensuremath{\textnormal{\textsf{CExpRew}}^{#1}\left(#2\,|\,#3\right)}}
\newcommand{\CLExpRew}[3]{\ensuremath{\textnormal{\textsf{CLExpRew}}^{#1}\left(#2\,|\,#3\right)}}

\DeclareMathAlphabet{\mathpzc}{OT1}{pzc}{m}{it}
\def\presuper#1#2%
  {\mathop{}%
   \mathopen{\vphantom{#2}}^{#1}%
   \kern-\scriptspace%
   #2}

%% file: introduction.tex
Probabilistic programming is \emph{en vogue}~\cite{Goodman:2014,DBLP:conf/icse/GordonHNR14}. 
It is mainstream in machine learning for describing distribution functions; Bay{\-}esi{\-}an inference is pivotal in their analysis.
It is used in security for describing both cryptographic constructions such as randomized encryption and experiments defining security notions~\cite{DBLP:journals/toplas/BartheKOB13}.
Probabilistic programs, being an extension of familiar notions, render these various fields accessible to programming communities.
A rich palette of probabilistic programming languages exists including {\tt Church}~\cite{DBLP:conf/uai/GoodmanMRBT08} 
as well as modern approaches like probabilistic {\tt C}~\cite{DBLP:conf/icml/PaigeW14}, {\tt Tabular}~\cite{DBLP:conf/popl/GordonGRRBG14} and {\tt R2}~\cite{Nori:2014}. 

Probabilistic programs are sequential programs having two main features: (1) the ability to draw values at random from 
probability distributions, and (2) the ability to condition values of variables in a program through observations.
The semantics of languages without conditioning is well--understood.
Kozen~\cite{DBLP:journals/jcss/Kozen81} considered denotational semantics, whereas McIver and Morgan~\cite{McIver:2004} provided a weakest (liberal) precondition (\textsf{w}(\textsf{l})\textsf{p}) semantics; a corresponding operational semantics is given by Gretz \emph{et al.}~\cite{DBLP:journals/pe/GretzKM14}.
Other relevant works include probabilistic power--domains~\cite{DBLP:conf/lics/JonesP89}, semantics of constraint probabilistic programming languages~\cite{DBLP:conf/concur/GuptaJS97}, and semantics for stochastic $\lambda$--calculi~\cite{DBLP:journals/japll/Scott14}.

Conditioning of variables through observations is less well--understood and raises various semantic difficulties as we will discuss in this paper. 
Previous work on semantics for programs with \Observe statements~\cite{Nori:2014,Hur:2014} do neither consider the possibility of non--termination nor the powerful feature of non--determinism. 
In this paper, we thoroughly study a more general setting which accounts for non--termination by means of a very simple yet powerful probabilistic programming language supporting non--determinism and observations. 
Let us first study a few examples that illustrate the semantic intricacies.
The sample program snippet $P_\mathit{obs_1}$
$$
\PChoice{\Ass x 0}{\nicefrac 1 2}{\Ass x 1};~\Observe~x = 1
$$
assigns zero to the variable $x$ with probability $\nicefrac 1 2$ while $x$ is assigned one with the same likelihood, after which we condition to 
the outcome $x$ being one.
The \Observe statement blocks all runs violating its condition and prevents those runs from happening.
It differs, e.g., from program annotations like (probabilistic) \emph{assertions}~\cite{DBLP:conf/pldi/SampsonPMMGC14}. 
The interpretation of the program is the expected outcome conditioned on permitted runs.
For the sample program $P_\mathit{obs_1}$ this yields the outcome $1 \cdot 1$---there is one feasible run that happens with probability one with $x$ being one.
Whereas this is rather straightforward, a slight variant like $P_{\mathit{obs_2}}$ 
$$
\PChoice{\Ass x 0;~\Observe~x = 1}{\nicefrac 1 2}{\Ass x 1;~\Observe~x = 1}
$$
is somewhat more involved, as the entire left branch of the probabilistic choice is infeasible.
Is this program equivalent to the sample program $P_\mathit{obs_1}$?

The situation becomes more intricate when considering loopy programs that may diverge.
Consider the programs $P_{\mathit{div}}$ (left) and $P_{\mathit{andiv}}$ (right):
$$
\begin{array}{lcl}
\Ass x 1; & \quad \quad & \Ass x 1;\\
\While~(x = 1)~\{ & \quad \quad & \While~(x = 1)~\{\\
\qquad \Ass x 1\  & \quad \quad & \qquad \PChoice{\Ass x 1}{\nicefrac 1 2}{\Ass x 0};\\
\} & \quad \quad & \qquad \Observe~x = 1\\
 &  \quad \quad & \}
\end{array}
$$
Program $P_{\mathit{div}}$ diverges and therefore yields as expected outcome zero.
Due to the conditioning on $x{=}1$, $P_{\mathit{andiv}}$ admits just a single---diverging---feasible run but this 
run almost surely never happens.
Its conditional expected outcome can thus not be measured.
It should be noted that programs with (probabilistic) assertions must be loop--free to avoid similar problems~\cite{DBLP:conf/pldi/SampsonPMMGC14}.
Other approaches insist on the absence of diverging loops~\cite{DBLP:conf/sas/ChakarovS14}.

Intricacies also occur when conditioning is used in programs that may abort.
Consider the program
\begin{align*}
 & \bigl\{\Abort\bigr\} \, [ \nicefrac{1}{2} ] 
 \bigl\{ \PChoice{\Ass{x}{0}} {\nicefrac{1}{2}} {\Ass{x}{1}}; \\[3pt]
  &  \quad  \PChoice{\Ass{y}{0}} {\nicefrac{1}{2}} {\Ass{y}{1}}; \, \Observe \, x{=}0 \lor y{=}0  \bigr\}
\end{align*}
where \Abort is the faulty aborting program which by definition does nothing else but diverge.
The above program tosses a fair coin and depending on the outcome either diverges or tosses a fair coin twice.
It finally conditions on at least once heads ($x{=}0$ or $y{=}0$). 
What is the probability that the outcome of the last coin toss was heads?
The main issue here is how to treat the possibility of abortion.

Combining conditioning with non--determinism is complicated, too.\footnote{As 
stated in \cite{DBLP:conf/icse/GordonHNR14},  ``representing and inferring sets of distributions 
is more complicated than dealing with a single distribution, and hence there are several technical 
challenges in adding non--determinism to probabilistic programs''.}
Non--determinism is a powerful means to deal with unknown information, as well as 
to specify abstractions in situations where details are unimportant.
Let program $P_{\mathit{nondet}}$ be:
\begin{align*}
&\PChoice{\NDChoice {\Ass{x}{5}} {\Ass{x}{2}}}{\nicefrac 1 4}{\Ass{x}{2}};\\
&\Observe \, x > 3
\end{align*}
where with probability $\nicefrac 1 4$, $x$ is set to either $5$ or $2$ non--deterministically (denoted 
$\NDChoice {\Ass{x}{5}} {\Ass{x}{2}}$), while $x$ is set to $2$ with likelihood $\nicefrac 3 4$.
Resolving the non--deterministic choice in favour of setting $x$ to five yields an expectation of $5$ for $x$,
obtained as $5 \cdot \nicefrac 1 4$ rescaled over the single feasible run of $P_{\mathit{nondet}}$.
Taking the right branch however induces an infeasible run due to the violation of the condition $x > 3$, 
yielding a non--measurable outcome.

The above issues---loops, divergence, and non--determi\-nism---indicate that conditioning in 
probabilistic programs is far from trivial.
This paper presents a thorough semantic treatment of conditioning in a probabilistic extension of 
Dijk{\-}stra's guarded command language (known as \pGCL~\cite{McIver:2004}), an elementary though 
foundational language that includes (amongst others) parametric probabilistic choice.
We take several semantic viewpoints.
Reward Markov Decision Processes (RMDPs)~\cite{Puterman:1994} are used as the basis for an 
\emph{operational} semantics.
This semantics is rather simple and elegant while covering \emph{all} aforementioned phenomena.
In particular, it discriminates the programs $P_{\mathit{div}}$ and $P_{\mathit{andiv}}$ while it does not discriminate $P_{\mathit{obs_1}}$ and $P_{\mathit{obs_2}}$.

We also provide a \emph{weakest pre--condition} (\wp) semantics \`a la~\cite{McIver:2004}. 
This is typically defined inductively over the structure of the program.
We show that combining both non--determinism and conditioning \emph{cannot} be treated in this manner.
Given this impossibility result we present a \wp--semantics for fully probabilistic programs, i.e., 
programs without non--determinism. 
To treat possibly non--terminating programs, due to e.g., diverging loops or abortion, this is 
complemented by a weakest \emph{liberal} pre--condition (\wlp) semantics.
The \wlp--semantics yields the weakest pre--expectation---the probabilistic pendant of weakest 
pre--condition---under which program $P$ either does not terminate or establishes a post--expectation. 
It thus differs from the \wp--semantics in not guaranteeing termination.
The \emph{conditional} weakest pre--expectation (\cwp) of $P$ with respect to post--expectation $f$ 
is then given by normalizing $\wp[P](f)$ with respect to $\wlp[P]({\bf 1})$.
The latter yields the \wp under which $P$ either does not terminate or terminates while passing all 
\Observe statements.
This is proven to correspond to conditional expected rewards in the RMDP--semantics, extending a 
similar result for \pGCL~\cite{DBLP:journals/pe/GretzKM14}.
Our semantic viewpoints are thus consistent for fully probabilistic programs.
Besides, we show that conditioning is semantically a truly conservative extension.
That is to say, our semantics is backward compatible with the (usual) \pGCL semantics; this does not
apply to alternative approaches such as {\tt R2}~\cite{Nori:2014}.

Finally, we show several practical applications of our results.
We present two program transformations which entirely eliminate conditioning from any
program and prove their correctness using the \textsf{w}(\textsf{l})\textsf{p}--semantics.
In addition, we show how the \textsf{w}(\textsf{l})\textsf{p}--semantics can be used to determine conditional probabilities in a simplified version of the 
parametric anonymity protocol Crowds~\cite{DBLP:journals/tissec/ReiterR98}.

Summarized, we provide the first operational semantics for imperative probabilistic programming languages with conditioning and both probabilistic and non--deterministic choice. Furthermore we give a denotational semantics for the fully probabilistic case, which in contrast to~\cite{Nori:2014,Hur:2014}, where every program is assumed to terminate almost surely, takes the probability of non--termination into account. Finally, our semantics enables to prove the correctness of several program transformations that eliminate \Observe statements.
%


%% file: prelim.tex
In this section we present the probabilistic programming language used for our
approaches and recall the notions of expectation transformers and (conditional)
expected reward over Markov decision processes used to endow the language with a
formal semantics.

\paragraph{Probabilistic programs and expectation transformers} 
\hspace{-1pt}We adopt the \emph{probabilistic guarded command language} (\pGCL)
\cite{McIver:2004} for describing probabilistic programs. \pGCL is an extension
of Dijkstra's guarded command language (\GCL) \cite{Dijkstra:1976} with a binary
probabilistic choice operator and its syntax is given by clause
\label{pGCL:syntax}
\begin{align*}
\Cmd ~::=~  &\Skip \mid \Abort \mid \Ass x E \mid \Cmd;\Cmd \mid \Cond G {\Cmd } {\Cmd } \\ 
   &{}~\mid  \, \PChoice {\Cmd} {p} {\Cmd} \mid
\NDChoice {\Cmd}{\Cmd} \mid \WhileDo G \Cmd~.
\end{align*}
%
%
Here, $x$ belongs to $\Var$, the set of program variables; $E$ is an 
arithmetical expression over $\Var$, $G$ a Boolean expression over $\Var$ and
$p$ a real--valued parameter with domain $[0,\,1]$. Most of the \pGCL
instructions are self--explanatory; we elaborate only on the following: $\PChoice{P}{p}{Q}$ represents a \emph{probabilistic choice} where programs $P$ is executed with probability $p$ and program $Q$ with
probability $1{-}p$. $\NDChoice{P}{Q}$ represents a \emph{non--deterministic} choice between $P$ and $Q$.

\pGCL programs are given a formal semantics through the notion of
\emph{expectation transformers}. Let $\State$ be the set of \emph{program states}, where
a program state is a variable valuation. Now assume that $P$ is a
\emph{fully probabilistic} program, \ie a program without non--deterministic
choices. We can see $P$ as a mapping from an initial state $\sigma$ to a
distribution over final states $\sem{P}(\sigma)$. Given a random variable $f
\colon \State \To \Reals_{\geq 0}$, transformer $\wp[P]$ maps every initial state
$\sigma$ to the expected value $\EV_{\sem P (\sigma)}(f)$ of $f$ with respect to
the distribution of final states $\sem{P} (\sigma)$. Symbolically,
\begin{equation*}\label{eq:cwp-st}
 \wp[P](f)(\sigma) = \EV_{\sem{P}(\sigma)}(f)~.
\end{equation*}
In particular, if $f = \ToExp{A}$ is the characteristic function of some event
$A$, $\wp[P](f)$ retrieves the probability that the event occurred after the
execution of $P$. (Moreover, if $P$ is a deterministic program in \GCL,
$\EV_{\sem{P}(\sigma)}(\ToExp{A})$ is $\{0,1\}$--valued and we recover the ordinary
notion of weakest pre--condition introduced by Dijkstra~\cite{Dijkstra:1976}.)
%
%

In contrast to the fully probabilistic case, the execution of a
non--deterministic program $P$ may lead to multiple---rather than a
single---distributions of final states. To account for these kind of programs,
the definition of $\wp[P]$ is extended as follows:
\begin{equation*}\label{ec:wp-non-det}
 \wp[P](f)(\sigma) = \inf_{\mu' \in \sem{P}(\sigma)}  \EV_{\mu'}(f)
\end{equation*}
In other words, $\wp[P](f)$ represents the tightest lower bound that can be
guaranteed for the expected value of $f$ (we assume that non-deterministic
choices are resolved \emph{demonically}\footnote{Demonic schedulers induce the most pessimistic expected outcome while in~\cite{McIver:2001b} also \emph{angelic schedulers} are considered which guarantee the most optimistic outcome.}, attempting to minimize the expected value of
$f$).

In the following, we use the term \emph{expectation} to refer to a random
variable mapping program states to real values. The expectation transformer \wp then
transforms a post--expectation $f$ into a pre--expectation $\wp[P](f)$ and
can be defined inductively, following the rules in Figure~\ref{fig:cwp} (second
column), Page~\pageref{fig:cwp}.  
The transformer \wp also admits a liberal variant
$\wlp$, which differs from $\wp$ on the way in which non--termination is
treated.

Formally, the transformer \wp operates on \emph{unbounded expectations} in $\Ex = \State
\To \Reals_{\geq 0}^{\infty}$ and \wlp operates on \emph{bounded expectations} in $\BEx
= \State \To [0,\, 1]$. Here $\Reals_{\geq 0}^{\infty}$ denotes the set of
non--negative real values with the adjoined $\infty$ value. In order to
guarantee the well--definedness of \wp and \wlp we need to provide $\Ex$ and
$\BEx$ the structure of a directed--complete partial order. Expectations are
ordered pointwise, \ie $f\sqsubseteq g$ iff $f(\sigma) \leq g(\sigma)$ for every
state $\sigma\in\State$. The least upper bound of directed subsets is also defined
pointwise. 

In what follows we use bold fonts for constant expectations,
\eg $\CteFun{1}$ denotes the constant expectation $1$. Given an arithmetical
expression $E$ over program variables we simply write $E$ for the expectation
that in state $\sigma$ returns $\sigma(E)$. Given a Boolean expression $G$ over
program variables we use $\ToExp{G}$ to denote the $\{0,1\}$--valued expectation
that returns $1$ if $\sigma \models G$ and $0$ otherwise.

\paragraph{MDPs and conditional expected rewards}
%
Let $\Par$ be a finite set of parameters. A \emph{parametric distribution} over a countable set $S$ is a function $\mu\colon S\to\poly{\Par}$ with $\sum_{s\in S}\mu(s)=1$, where $\poly{\Par}$ denotes the set of all polynomials\footnote{Although parametric distributions are defined as polynomials over the parameters, we only use $p$ and $1-p$ for $p\in\Par$} over $\Par$. $\Distr(S)$ denotes the set of parametric distributions over $S$.
%
\begin{definition}[Parametric Discrete--time Reward Markov Decision Process]
  \label{def:mdp}
  Let $\AP$ be a set of atomic propositions. A \emph{parametric discrete--time reward Markov decision process (RMDP)} is a tuple $\rmdp=(S,\,\sinit,\,\Act,\,\pmdp,\,L,\,r)$ with a countable set of states $S$, a unique initial state $\sinit\in S$, a finite set of actions $\Act$, a transition probability function $\pmdp\colon S\times\Act\to\Distr(S)$ with $\forall (s,\,\act)\in S\times Act\suchthat \sum_{s'\in S}\pmdp(s,\,\act)(s') = 1$, a labeling function $L\colon S\to 2^{\AP}$, and a reward function $r\colon S\to\Reals_{\geq 0}$.
\end{definition}
\noindent A \emph{path} of $\rmdp$ is a finite or infinite sequence $\path = s_0
\act_0 s_1 \act_1\ldots$ such that $s_i\in S$, $\act_i\in\Act$, $s_0=\sinit$,
and $\pmdp(s_i,\,\act_i)(s_{i+1})>0$ for all $i\geq 0$. A finite path is denoted
by $\fpath=s_0 \act_0\ldots s_n$ for $n\in\Naturals$ with $\last \fpath = s_n$
and $|\pi|=n$. The $i$-th state $s_i$ of $\pi$ is denoted $\pi(i)$. The set of
all paths of $\rmdp$ is denoted by $\pathset^{\rmdp}$ and sets of infinite or
finite paths by $\pathsinf^{\rmdp}$ or $\pathsfin^{\rmdp}$,
respectively. $\pathset^{\rmdp}(s)$ is the set of paths starting in $s$ and
$\pathset^{\rmdp}(s,s')$ is the set of all finite paths starting in $s$ and
ending in $s'$. This is also lifted to sets of states.  If clear from the
context we omit the superscript $\rmdp$.

An MDP operates by a non--deterministic choice of an action $\act\in\Act$ that is
\emph{enabled} at state $s$ and a subsequent probabilistic determination of a
successor state according to $\pmdp(s,\act)$.  We denote the set of actions that
are enabled at $s$ by $\Act(s)$ and assume that $\Act(s) \neq \emptyset$ for
each state $s$.  A state $s$ with $|\Act(s)|=1$ is called \emph{fully
  probabilistic}, and in this case we use $\pmdp(s,\, s')$ as a shorthand for
$\pmdp(s,\,\act)(s')$ where $\Act(s)=\{\act\}$.
For resolving the non--deterministic choices, so--called \emph{schedulers} are
used.  In our setting, \emph{deterministic} schedulers suffice, which are
partial functions $\sched\colon\pathsfin^{\rmdp}\to \Act$ with
$\sched(\fpath)\in\Act(\last \fpath)$.  A deterministic scheduler is called
\emph{memoryless} if the choice depends only on the current state, yielding a
function $\sched\colon S\to\Act$.  The class of all (deterministic) schedulers
for $\rmdp$ is denoted by $\Sched^{\rmdp}$.

A \emph{parametric discrete--time reward Markov chain (RMC)} is an RMDP with only fully probabilistic states. For an RMC we use the notation $\rdtmc=(S$, $\sinit$, $\pdtmc$, $L$, $r)$ where $\pdtmc\colon S\to\Distr(S)$ is called a \emph{transition probability matrix}. 
For RMDP $\rmdp$, the fully probabilistic system $^\sched\rmdp$ induced by a
scheduler $\sched\in\Sched^{\rmdp}$ is an \emph{induced RMC}. A
\emph{probability measure} is defined on the induced RMCs.  The measure for RMC
$\rdtmc$ is given by $\Pr^\rdtmc\colon\pathsfin^\rdtmc\to\Ireal$ with
$\Pr^\rdtmc(\fpath)=\prod_{i=0}^{n-1}\pdtmc(s_i,s_{i+1})$, for $\fpath=s_0\ldots
s_n$.  The probability measure can be lifted to sets of (infinite) paths using a
cylinder set construction, see~\cite[Ch.\ 10]{DBLP:books/daglib/0020348}.  The
\emph{cumulated reward} of a finite path $\fpath=s_0\ldots s_n$ is given by
$r(\fpath)=\sum_{i=0}^{n-1}r(s_i)$ as the reward is ``earned'' when leaving the
state.
%

We consider \emph{reachability properties} of the form $\Finally T$ for a set of
target states $T=\{s\in S\mid T\in L(s)\}$ where $T$ is overloaded to be a set
of states and a label in $\AP$.  The set $\Finally T =
\{\pi\in\pathset(s_I,T)\mid \forall 0\leq i<|\pi|\mydot\,\pi(i)\not\in T\}$ shall be
prefix--free and contain all paths of $\rdtmc$ that visit a target state.
Analogously, the set $\neg\Finally T = \{\pi\in\pathset^{\rdtmc}(s_I)\mid
\forall i\geq 0\mydot\, \pi(i)\not\in T\}$ contains all paths that never reach a
state in $T$.  Let us first consider reward objectives for fully probabilistic
models, i.e., RMCs.  The \emph{expected reward} for a finite set of paths
$\Finally T\in\pathsfin^{\rdtmc}$ is
\begin{align*}
\ExpRew \rdtmc {\Finally T} ~\eqdef~ \sum_{\fpath\in\Finally T}\Pr^{\rdtmc}(\fpath) \cdot r(\fpath)~.
\end{align*}
For a reward bounded by one, the notion of the \emph{liberal} expected reward also takes the mere probability of \emph{not} reaching the target states into account:
\begin{align*}
\LExpRew \rdtmc {\Finally T} ~\eqdef~ \ExpRew \rdtmc {\Finally T} + \Pr^{\rdtmc}(\neg\Finally T)
\end{align*}
A liberal expected reward will later represent the probability of either establishing some condition or not terminating.

To explicitly exclude the probability of paths that reach ``undesired''
states, we let $U=\{s\in S\mid \bad\in L(s)\}$ 
%
and define the \emph{conditional expected reward} for the condition $\neg\Finally U$ by\footnote{Note that strictly formal one would have to define the intersection of sets of finite and possibly infinite paths by means of a cylinder set construction considering all infinite extensions of finite paths.}
\begin{align*}
    \CExpRew \rdtmc {\Finally T} {\neg\Finally U} ~\eqdef~ \frac{\ExpRew \rdtmc {\Finally T \cap \neg\Finally U}}{\Pr^{\rdtmc}(\neg\Finally U)}~.
\end{align*}
For details about conditional probabilities and expected rewards, we refer to~\cite{DBLP:conf/tacas/BaierKKM14}. Conditional \emph{liberal} expected rewards are defined by
\begin{align*}
    \CLExpRew \rdtmc {\Finally T} {\neg\Finally U} ~\eqdef~ \frac{\LExpRew \rdtmc {\Finally T \cap \neg\Finally U}}{\Pr^{\rdtmc}(\neg\Finally U)}~.
\end{align*}
Reward objectives for RMDPs are now defined using a \emph{demonic} scheduler $\sched\in\Sched^{\rmdp}$ minimizing probabilities and expected rewards for the induced RMC $\presuper \sched \rdtmc$. For the expected reward this yields
\begin{align*}
\ExpRew \rmdp {\Finally T} ~\eqdef~ \inf_{\sched\in\Sched^{\rmdp}}\ExpRew {\presuper \sched \rdtmc}{\Finally T}~.
\end{align*}
The scheduler for conditional expected reward properties minimizes the value of the quotient:
\begin{align*}
    &\CExpRew \rmdp {\Finally T} {\neg\Finally U}\\
    ~\eqdef~ &\inf_{\sched\in\Sched^{\rmdp}}\CExpRew {\presuper \sched \rdtmc} {\Finally T} {\neg\Finally U}\\
    ~=~ &\inf_{\sched\in\Sched^{\rmdp}} \frac{\ExpRew {\presuper \sched \rdtmc} {\Finally T \cap \neg\Finally U}}{\Pr^{{\presuper \sched \rdtmc}}(\neg\Finally U)}
\end{align*}
The liberal reward notions for RMDPS are analogous.
Regarding the quotient minimization we assume ``$\frac 0 0 < 0$" as we see $\frac 0 0$---being undefined---to be less favorable than $0$.

%

%% file: cpGCL.tex


As mentioned in Section~\ref{sec:prelim}, \pGCL programs can be considered as
distribution transformers.  Inspired by~\cite{DBLP:conf/icse/GordonHNR14}, we
extend \pGCL with \Observe statements to obtain \emph{conditional} \pGCL
(\cpGCL, for short).  This is done by extending the syntax of \pGCL
(p.~\pageref{pGCL:syntax}) with $\Observe~G$ where $G$ is a
Boolean expression over the program variables.  When a program's execution
reaches $\Observe~G$ with a current variable valuation
$\sigma \not\models G$, further execution of the program is blocked as with an
\texttt{assert} statement~\cite{DBLP:journals/toplas/Nelson89}.  In contrast to \texttt{assert}, however, the
\Observe statements do not only block further execution but \emph{condition}
resulting distributions on the program's state to only those executions
satisfying the observations. Consider two small example programs:
\begin{align*}
    & \PChoice{\Ass{x}{0}} {p}{\Ass{x}{1}}; &
    & \qquad \PChoice{\Ass{x}{0}} {p}{\Ass{x}{1}};\\[-1pt]
    & \PChoice{\Ass{y}{0}}{q}{\Ass{y}{-1}} &
    & \qquad \PChoice{\Ass{y}{0}}{q}{\Ass{y}{-1}};\\[-2pt]
    & &  & \qquad \Observe ~ x+y=0
\end{align*}
The left program establishes that the probability of $x{=}0$ is $p$,
whereas for the right program this probability is $\tfrac{pq}{pq+(1-p)(1-q)}$.
The left program admits all (four) runs, two of which satisfy $x{=}0$.  Due to
the \Observe statement requiring $x{+}y{=}0$, the right program, however, admits
only two runs ($x{=}0, y{=}0$ and $x{=}1, y{=}{-}1$), satisfying
$x{=}0$.

In Section~\ref{sec:denotational} we will focus on the subclass of fully
probabilistic programs in \cpGCL, which we denote \cfpGCL. 

%% file: operational.tex
\begin{figure*}[th]
\scriptsize
\abovedisplayskip=0pt
\begin{align*}
&(\textbf{terminal})\,\frac{\vphantom{\langle}}{\langle {\downarrow},\, \sigma \rangle ~\longrightarrow~ \sink}\qquad\qquad(\textbf{skip})\,\frac{\vphantom{\langle}}{\langle \Skip,\, \sigma \rangle ~\longrightarrow~ \langle {\downarrow},\, \sigma \rangle}\qquad\qquad
(\textbf{abort})\,\frac{\vphantom{\langle}}{\langle \Abort,\, \sigma \rangle ~\longrightarrow~ \langle \Abort,\, \sigma \rangle}\qquad\qquad
(\textbf{undesired})\,\frac{\vphantom{\langle}}{\undesired ~\longrightarrow~ \sink}\qquad\quad\\
&(\textbf{assign})\,\frac{\vphantom{\langle}}{\langle \Ass x E,\, \sigma \rangle ~\longrightarrow~ \langle {\downarrow},\, \sigma[x \leftarrow \llbracket E \rrbracket_\sigma] \rangle}\qquad\qquad(\textbf{observe})\,\frac{\sigma \models G}{\langle \Observe\, G,\, \sigma \rangle ~\longrightarrow~ \langle {\downarrow} ,\, \sigma \rangle}\qquad\qquad
\frac{\sigma \not\models G}{\langle \Observe\, G,\, \sigma \rangle ~\longrightarrow~ \undesired}\\
&(\textbf{concatenate})\,\frac{}{\langle\downarrow;{Q},\, \sigma \rangle ~\longrightarrow~ \langle Q,\, \sigma \rangle}\qquad\qquad
\frac{\langle P,\, \sigma \rangle ~\longrightarrow~ \undesired}{\langle {P};{Q},\, \sigma \rangle ~\longrightarrow~ \undesired}\qquad\qquad
\frac{\langle P,\, \sigma \rangle ~\longrightarrow~ \mu}{\langle {P};{Q},\, \sigma \rangle ~\longrightarrow~ \nu},\,\textnormal{where } \forall P'.\,\nu(\langle {P'};{Q}, \sigma'\rangle) := \mu(\langle P',\, \sigma'\rangle)\\
&(\textbf{if})\,\frac{\sigma \models G}{\langle \Cond G P Q,\, \sigma \rangle ~\longrightarrow~ \langle P,\, \sigma \rangle}\qquad\qquad
\frac{\sigma \not\models G}{\langle \Cond G P Q,\, \sigma \rangle ~\longrightarrow~ \langle Q,\, \sigma \rangle}\\
&(\textbf{while})\,\frac{\sigma \models G}{\langle \WhileDo G P,\, \sigma \rangle ~\longrightarrow~ \langle {P};{\WhileDo G P},\, \sigma \rangle}\qquad\qquad
\frac{\sigma \not\models G}{\langle \WhileDo G P,\, \sigma \rangle ~\longrightarrow~ \langle {\downarrow},\, \sigma \rangle}\\
&(\textbf{prob.\ choice})\,\frac{\vphantom{\langle}}{\langle \PChoice P p Q,\, \sigma \rangle ~\longrightarrow~ \nu},\,\textnormal{where } \nu(\langle P,\, \sigma\rangle) := p,\, \nu(\langle Q,\, \sigma\rangle) := 1 - p\\
&(\textbf{non--det.\ choice})\,\frac{\vphantom{\langle}}{\langle \NDChoice P Q ,\, \sigma \rangle ~ \xrightarrow{~\mathit{left}~} ~ \langle P,\, \sigma \rangle}\qquad\qquad
\frac{\vphantom{\langle}}{\langle \NDChoice P Q ,\, \sigma \rangle ~\xrightarrow{~\mathit{right}~}~ \langle Q,\, \sigma \rangle}
\end{align*}\normalsize
%
\caption{Rules for the construction of the operational RMDPs. If not stated
  otherwise, $\langle s \rangle {\longrightarrow} \langle t \rangle$ is a shorthand for $\langle s \rangle {\longrightarrow} \mu\in\Distr(\State)$ with $\mu(\langle t \rangle)=1$. 
A terminal state of the form $\langle {\downarrow},\, \sigma \rangle$ indicates successful termination.
Terminal states and \undesired go to the \sink state. 
\Skip without context terminates successfully.
\Abort self--loops, i.e.\ diverges. 
$\Ass x E$ alters the variable valuation according to the assignment then terminates successfully. 
For the concatenation, $\langle\term;{Q},\, \sigma \rangle$ indicates successful termination of the first program, so the execution continues with $\langle{Q},\, \sigma \rangle$.
If for $P;\,Q$ the execution of $P$ leads to $\undesired$, $P;\,Q$ does so, too.
Otherwise, for $\langle P,\sigma \rangle{\longrightarrow}\mu$, $\mu$ is lifted such that $Q$ is concatenated to the support of $\mu$.
If for the conditional choice $\sigma\models G$ holds, $P$ is executed, otherwise $Q$.
The case for $\While$ is similar.
For the probabilistic choice, a distribution $\nu$ is created according to $p$. For $\NDChoice P Q$, we call $P$ the $\mathit{left}$ choice and $Q$ the $\mathit{right}$ choice for actions $\mathit{left},\mathit{right}\in\Act$.
For the \Observe statement, if $\sigma\models G$ \Observe acts like $\Skip$. 
Otherwise, the execution leads directly to \undesired indicating a violation of the \Observe statement.
}
\label{fig:rmdprules}
\end{figure*}
This section presents an operational semantics for \cpGCL using RMDPs as underlying model inspired by~\cite{DBLP:journals/pe/GretzKM14}.
Schematically, the operational RMDP of a \cpGCL program shall have the following structure:
\begin{center}
\scalebox{0.9}{\input{pics/schema}}
\end{center}
Terminating runs eventually end up in the \sink state; other runs are diverging (never reach \sink).
A program terminates either successfully, \ie a run passes a \exit--labelled state, or terminates due to a false observation, \ie a run passes $\undesired$.
Squiggly arrows indicate reaching certain states via possibly multiple paths and states; the clouds indicate that there might be several states of the particular kind.
The \exit--labelled states are the \emph{only ones} with positive reward.
Note that the sets of paths that eventually reach \undesired, eventually reach \exit, or diverge, are pairwise disjoint. 
\begin{definition}[Operational \cpGCL semantics]
The \emph{operational semantics} of $P\in\cpGCL$ for $\sigma\in\State$ and $f\in\Ex$ is the RMDP $\Rmdp \sigma {f} P = (S$, $\langle P,\,\sigma\rangle$, $\Act$, $\pdtmc$, $L$, $r)$, such that $S$ is the smallest set of states with $\undesired \in S$, $\sink \in S$, and $\langle Q,\, \tau \rangle,\langle \term,\, \tau \rangle \in S$ for $Q \in \pGCL$ and $\tau \in \State$. $\langle P,\,\sigma\rangle\in S$ is the initial state. $\Act=\{\mathit{left},\,\mathit{right}\}$ is the set of actions. $\pdtmc$ is formed according to the rules given in Figure~\ref{fig:rmdprules}. The labelling and the reward function are given by:
\begin{align*}
L(s) ~\eqdef~&\begin{cases}
\{\exit\},& \textnormal{if } s = \langle\term,\, \tau\rangle, \textnormal{ for some } \tau\in\State\\
\{\sinklabel\},& \textnormal{if } s = \sink\\
\{\bad\},& \textnormal{if } s = \undesired\\
\emptyset, &\textnormal{otherwise,}
\end{cases}\\
r(s) ~\eqdef~&\begin{cases}
f(\tau),& \textnormal{if } s = \langle\term,\, \tau\rangle, \textnormal{ for some } \tau\in\State\\
0,&\textnormal{otherwise}
\end{cases}
\end{align*}
where a state of the form $\langle \term,\, \tau \rangle$ denotes a terminal state in which no program is left to be executed.
\end{definition}
\noindent To determine the \emph{conditional expected outcome of program $P$}
given that all observations are true, we need to determine the \emph{expected
  reward to reach $\sink$~from the initial state conditioned on not
  reaching $\undesired$} under a demonic scheduler.  For $\Rmdp \sigma {f} P$
this is given by $\CExpRew{\Rmdp \sigma {f} P}{\lozenge
  \sinklabel}{\neg\lozenge\bad}$. Recall for the condition $\neg\Finally\bad$
that all paths not eventually reaching \undesired either diverge (thus collect
reward 0) or pass by a \exit--labelled state and eventually reach \sink. This
gives us:
%
\begin{align*}
&\CExpRew{\Rmdp \sigma {f} P}{\lozenge \sinklabel}{\neg\lozenge\bad}\\
=~ &\inf_{\sched\in\Sched^{\Rmdp \sigma {f} P}}\frac{\ExpRew{\presuper \sched {\Rmdp \sigma {f} P}}{\lozenge \sinklabel  \cap \neg\lozenge\bad}}{\Pr^{\presuper \sched {\Rmdp \sigma {f} P}}({\neg\lozenge\bad})}~\\
=~ &\inf_{\sched\in\Sched^{\Rmdp \sigma {f} P}}\frac{\ExpRew{\presuper \sched {\Rmdp \sigma {f} P}}{\lozenge \sinklabel}}{\Pr^{\presuper \sched {\Rmdp \sigma {f} P}}({\neg\lozenge\bad})}
\end{align*}
This is analogous for $\CLExpRew{\Rdtmc{\sigma}{f}{P}}{\lozenge \sinklabel}{\neg\lozenge\bad}$.
\begin{example}
Consider the program $P\in\cpGCL$:
\begin{align*}
&\PChoice{\NDChoice {\Ass{x}{5}} {\Ass{x}{2}}}{q}{\Ass{x}{2}};\\
&\Observe \, x>3
\end{align*}
where with parametrized probability $q$ a non--deterministic choice between $x$ being assigned $2$ or $5$ is executed, and with probability $1-q$, $x$ is directly assigned $2$.
Let for readability $P_1= \NDChoice {\Ass{x}{5}} {\Ass{x}{2}}$, $P_2= \Ass{x}{2}$, $P_3= \Observe \, x>3$, and $P_4=\Ass x 5$. 
The operational RMDP $\Rmdp {\sigma_I} x {P}$ for an arbitrary initial variable valuation $\sigma_I$ and post--expectation $x$ is depicted below.
\begin{center}
\scalebox{0.8}{\input{pics/operational_example}}\\
\end{center}
\noindent The only state with positive reward is $s' \coloneqq
\langle\term,\,\sigma_I\subst x 5\rangle$ and its reward is indicated by number
$5$.  Assume first a scheduler choosing action $\mathit{left}$ in state $\langle
P_1;\,P_3,\,\sigma_I\rangle$.  In the induced RMC the only path accumulating
positive reward is the path $\pi$ going from $\langle P,\, \sigma_I\rangle$ via
$s'$ to $\sink$ with $r(\pi)=5$ and $\Pr(\pi)=q$.  This gives an expected reward
of $5\cdot q$.  The overall probability of not reaching $\undesired$ is also
$q$.  The conditional expected reward of eventually reaching $\sink$ given that
$\undesired$ is not reached is hence $\frac{5\cdot q}{q}=5$.  Assume now the
\emph{minimizing} scheduler choosing $\mathit{right}$ at state $\langle
P_1;\,P_3,\,\sigma_I\rangle$.  In this case there is no path having positive
accumulated reward in the induced RMC, yielding an expected reward of $0$.  The
probability of not reaching $\undesired$ is also $0$.  
The conditional expected reward in this case is undefined ($\nicefrac{0}{0}$) thus the $\mathit{right}$ branch is preferred over the $\mathit{left}$ branch.

In general, the operational RMDP is not finite, even if the program terminates almost--surely (i.e.\ with probability 1).
\end{example}

%% file: pics/schema.tex
\begin{tikzpicture}[->,>=stealth',shorten >=1pt,node distance=2.5cm,semithick,minimum size=1cm]
\tikzstyle{every state}=[draw=none]
  \draw[white, use as bounding box] (-1.1,-1.8) rectangle (6.4,1.5);
   \node [state, initial, initial text=] (init) {$\langle \mathpzc{init} \rangle$};  
   \node [cloud, draw=black,cloud puffs=15, cloud puff arc= 150,
        minimum width=1.5cm, minimum height=.75cm, aspect=1] (exit) [right of=init] {$\exit$};
   \node [state] (bad) [above=0.3cm of exit] {$\undesired$};
   \node [state] (sink) [right of=exit] {$\sink$};
   \node [cloud, draw=black,cloud puffs=15, cloud puff arc= 150,
        minimum width=1.5cm, minimum height=.75cm, aspect=1] (diverge) [below=0.5 cm of exit] {$\phantom{\exit}$};

    \node [] (divergetext) [below=-0.825 cm of diverge] {\small$\mathpzc{diverge}$};

   \node [state] (haken1) at (2.2, .1) {\scriptsize $\exit$};
   \node [state] (haken2) at (2.25, -.1) {\tiny $\exit$};
    \node [state] (haken3) at (2.7, -.1) {\scriptsize $\exit$};
   \node [state] (haken4) at (2.75, .1) {\tiny $\exit$};
   \node [state] (haken5) at (2.95, .0) {\tiny $\exit$};

  \path [] 
      (init) edge [decorate,decoration={snake, post length=2mm}] (exit)
      (init) edge [decorate,decoration={snake, post length=2mm}] (bad)
      (init) edge [decorate,decoration={snake, post length=2mm}] (diverge)
      (exit) edge [] (sink)
      (bad) edge [] (sink)
      (sink) edge [loop right] (sink)
      (diverge) edge [loop right,decorate,decoration={snake, post length=2mm}] (diverge)
  ;
\end{tikzpicture}

%% file: pics/operational_example.tex
\begin{tikzpicture}[->,>=stealth',shorten >=1pt,node distance=2.7cm,semithick,minimum size=1cm]
\tikzstyle{every state}=[draw=none]
  \draw[white, use as bounding box] (-6.6,-6.3) rectangle (3.4,0.3);
   \node [state, initial where=above, initial text=] (init) at (0,0) {$\langle P,\,\sigma_I \rangle$};  
    \node [] (s1) [on grid, below=1.5 cm of init, xshift=-2.2cm] {$\langle P_1;\,P_3,\,\sigma_I\rangle$};   
        \node [] (s11) [on grid, below=1.5 cm of init, xshift=2.2cm] {$\langle P_2;\,P_3,\,\sigma_I\rangle$};   
        \node [] (s12) [on grid, left=2 cm of s1, xshift=-1.2cm] {$\langle P_4;\,P_3,\,\sigma_I\rangle$};   
        \node [] (s111) [on grid, below=1.5 cm of s11] {$\langle\term;\,P_3,\,\sigma_I\subst x 2\rangle$};   
        \node [] (s121) [on grid, below=1.5 cm of s12] {$\langle\term;\,P_3,\,\sigma_I\subst x 5\rangle$};   
         \node [] (s1111) [on grid, below=1.5 cm of s111] {$\langle\,P_3,\,\sigma_I\subst x 2\rangle$};   
         \node [] (s1211) [on grid, right=3.3 cm of s121] {$\langle P_3,\,\sigma_I\subst x 5\rangle$};   
         \node [label={[yshift=0.1cm, gray] 180:$5$}] (s12111) [on grid, below=1.5 cm of s1211] {$\langle\term,\,\sigma_I\subst x 5\rangle$};   
 \node [] (undesired) [on grid, below=1.5 cm of s1111, xshift=-.5cm] {$\undesired$};  
  \node [] (sink) [on grid, below=1.5 cm of s12111, xshift=1.5cm] {$\sink$};  

\path
      (init) edge [] node [left, near start] {\scriptsize{$q$}} (s1)
      (init) edge [] node [right, near start] {\scriptsize{$1-q$}} (s11)
      (s1) edge [] node [below=-0.3cm] {\scriptsize{$\mathit{left}$}} (s12)
      (s1) edge [] node [below=-0.3cm] {\scriptsize{$\mathit{right}$}} (s11)
      (s12) edge [] (s121)
      (s11) edge [] (s111)
      (s111) edge [] (s1111)
      (s121) edge [] (s1211)
      (s1111) edge [] (undesired)
      (s1211) edge [] (s12111)
      (undesired) edge [] (sink)
      (s12111) edge [] (sink)
      (sink) edge [loop above] (sink)
;
\end{tikzpicture}

%% file: denotational.tex
This section presents an expectation transformer semantics for the fully
probabilistic fragment \cfpGCL of \cpGCL.  We formally relate this to the
\wp/\wlp--semantics of \pGCL as well as to the operational semantics from the
previous section.

\subsection{Conditional Expectation Transformers}

\label{sec:expec_transform}
\input{expectation_transformer}

\subsection{Correspondence Theorem}\label{sec:equiv}
\input{equiv}

%% file: expectation_transformer.tex
An expectation transformer semantics for the fully probabilistic fragment of
\cpGCL is defined using the operators:
\begin{align*}
\cwp[\,\cdot\,]  \colon &\Ex \times \BEx \To \Ex \times \BEx\\
\cwlp[\,\cdot\,]  \colon &\BEx \! \times \BEx \To \BEx \! \times \BEx
\end{align*}
These functions can intuitively be viewed as the counterpart of $\wp$ and $\wlp$
respectively, as shortly shown.  The weakest conditional pre--expectation
$\qcwp[P](f)$ of $P\in\cfpGCL$ with respect to post--expectation $f$ is now
given as
\begin{equation*}\label{eq:cwp-quotient}
\qcwp[P](f) ~\eqdef~ \frac{\cwp_1[P](f, \CteFun{1})}{\cwp_2[P](f, \CteFun{1})}~,
\end{equation*}
where $\cwp_1[P] (f,g)$ (resp.~$\cwp_2[P](f,g)$) denotes the first (resp.~second) 
component of $\cwp[P] (f,g)$ and $\CteFun{1}$ is the constant expectation one.
The weakest liberal conditional pre--expectation $\qcwlp[P](f)$ is defined analogously.
In words, $\qcwp[P](f)(\sigma)$ represents the expected value of $f$ with respect 
to the distribution of final states obtained from executing $P$ in state $\sigma$, 
given that all \Observe statements occurring along the runs of $P$ were satisfied. 
The quotient defining $\qcwp[P](f)$ is interpreted is the same way as the
quotient
\begin{align*}
   \frac{\Pr(A \cap B)}{\Pr(B)}
\end{align*}
encoding the conditional probability $\Pr(A | B)$. However, here we measure the
expected value of random variable $f$\footnote{In fact, $\qcwp[P](f)(\sigma)$
  corresponds to the notion of \emph{conditional expected value} or in simpler
  terms, the expected value over a conditional distribution.}.  The denominator
$\cwp_2[P](f,\CteFun{1})(\sigma)$ measures the probability that $P$ satisfies
all the observations (occurring along valid runs) from the initial state
$\sigma$.  If $\cwp_2[P](f,\CteFun{1})(\sigma)=0$, program $P$ is
\emph{infeasible} from state $\sigma$ and in this case $\qcwp[P](f)(s)$ is not
well--defined (due to the division by zero).  This corresponds to the conditional
probability $\Pr(A | B)$ being not well--defined when $\Pr(B)=0$.

\begin{figure*}[ht]
\begin{center}
\renewcommand{\arraystretch}{1.25}
\begin{tabular}{lll}
\hline
$\boldsymbol{P}$ & $\boldsymbol{\wp[P](f)}$ & $\boldsymbol{\cwp[P](f,\, g)}$\\
\hline
$\Skip$ & $f$ & $ (f,\, g)$\\
$\Abort$ & $\ExpZero$ & $ (\ExpZero,\, \ExpOne)$\\
$\Ass{x}{E}$ & $f\subst{x}{E}$ & $(f\subst{x}{E},\, g\subst{x}{E})$\\
$\Observe \, G$ & $\ToExp{G} \cdot  f$ & $\ToExp{G} \cdot  (f,\, g)$\\
$P_1;~P_2$ & $(\wp[P_1] \circ \wp[P_2]) (f)$ & $(\cwp[P_1] \circ \cwp[P_2]) (f, g)$\\
$\Cond{G}{P_1}{P_2}$ & $\ToExp{G} \cdot \wp[P_1](f) + \ToExp{\lnot G} \cdot \wp[P_2](f)$ & $\ToExp{G} \cdot \cwp[P_1](f,\,g) + \ToExp{\lnot G} \cdot \cwp[P_2](f,\,g)$\\
$\PChoice{P_1}{p}{P_2}$ & $p \cdot \wp[P_1](f)  + (1 - p) \cdot \wp[P_2](f)$ & $p \cdot \cwp[P_1](f,\,g)  + (1 - p) \cdot \cwp[P_2](f,\,g)$\\
$\NDChoice{P_1}{P_2}$ & $\lambda \sigma \mydot \min \{\wp[P_1](f)(\sigma),\, \wp[P_2](f)(\sigma)\}$ & --- not defined ---\\
$\WhileDo{G}{P'}$ & $\lfp \hat{f} \mydot \left( \ToExp{G} \cdot \wp[P'] (\hat{f})  + \ToExp{\lnot G} \cdot f\right)$ & $\lfp_{\scriptscriptstyle{\sqsubseteq,\sqsupseteq}} (\hat{f},\,\hat{g}) \mydot \left( \ToExp{G} \cdot \cwp[P'] (\hat{f},\,\hat{g})  + \ToExp{\lnot G} \cdot (f,\,g)\right)$\\\noalign{\smallskip}\noalign{\smallskip}
\hline
$\boldsymbol{P}$ & $\boldsymbol{\wlp[P](f)}$ & $\boldsymbol{\cwlp[P](f,\, g)}$\\
\hline
$\Abort$ & $\ExpOne$ & $ (\ExpOne,\, \ExpOne)$\\
$\WhileDo{G}{P'}$ & $\gfp \hat{f} \mydot \left( \ToExp{G} \cdot \wp[P'] (\hat{f})  + \ToExp{\lnot G} \cdot f\right)$ & $\gfp_{\scriptscriptstyle{\sqsubseteq,\sqsubseteq}} (\hat{f},\,\hat{g}) \mydot \left( \ToExp{G} \cdot \cwp[P'] (\hat{f},\,\hat{g})  + \ToExp{\lnot G} \cdot (f,\,g)\right)$
\end{tabular}
\renewcommand{\arraystretch}{1}
\end{center}
\caption{Definitions for the \wp/\wlp and \cwp/\cwlp operators. The \wlp (\cwlp) operator differs from \wp (\cwp) only for \Abort and the \While--loop.
A scalar multiplication $a \cdot (f,\, g)$ is meant componentwise yielding $(a\cdot f,\, a\cdot g)$.
Likewise an addition $(f,\, g) + (f',\, g')$ is also meant componentwise yielding $(f+f',\, g+g')$.}
\label{fig:cwp}
\end{figure*}

%
The operators $\cwp$ and $\cwlp$ are defined inductively on the
program structure, see Figure~\ref{fig:cwp} (last column).  Let us briefly explain this.
$\cwp[\Skip]$ behaves as the identity since $\Skip$ has no effect on the program
state. $\cwp[\Abort]$ maps any pair of post--expectations to the pair of
constant pre--expectations $(\CteFun{0}, \CteFun{1})$.  Assignments induce a
substitution on expectations, \ie $\cwp[\Ass{x}{E}]$ maps $(f, g)$ to
pre--expectation $(f\subst{x}{E}, \allowbreak g\subst{x}{E})$, where
$h\subst{x}{E}(\sigma) = h (\sigma\subst{x}{E})$ and $\sigma\subst{x}{E}$
denotes the usual variable update on states.  $\cwp[P_1;P_2]$ is obtained as the
functional composition (denoted $\circ$) of $\cwp[P_1]$ and
$\cwp[P_2]$. $\cwp[\Observe \, G]$ restricts post--expectations to those states
that satisfy $G$; states that do not satisfy $G$ are mapped to $0$.
$\cwp[\Cond{G}{P_1}{P_2}]$ behaves either as $\cwp[P_1]$ or $\cwp[P_2]$ according to the
evaluation of $G$. $\cwp[\PChoice{P_1}{p}{P_2}]$ is obtained as a convex
combination of $\cwp[P_1]$ and $\cwp[P_2]$, weighted according to $p$.
$\cwp[\WhileDo{G}{P'}]$ is defined using standard fixed point
techniques.\footnote{We define $\cwp[\WhileDo{G}{P}]$ by the least fixed
  point \wrt the order $(\sqsubseteq,\sqsupseteq)$ in $\Ex \times \BEx$.  This way we
  encode the greatest fixed point in the second component \wrt the order $\sqsubseteq$ over $\BEx$ as the
  least fixed point \wrt the dual order $\sqsupseteq$.}

The $\cwlp$ transformer follows the same rules as $\cwp$, except for the \Abort 
and \While statements. $\cwlp[\Abort]$ takes any post--expectation to pre--expectation 
$(\CteFun{1}, \CteFun{1})$ and $\cwlp[\WhileDo{G}{P}]$ is defined as a \emph{greatest} 
fixed point rather than a least fixed point.

\begin{example}\label{ex:run-ex-cwp}
\noindent
Consider the program $P'$
$$
\begin{array}{ll}
   _1 & \PChoice{\Ass{x}{0}} {\nicefrac{1}{2}} {\Ass{x}{1}};\\[3pt]
   _2 & \Ite \, (x=1) \: \big\{  \PChoice{\Ass{y}{0}} {\nicefrac{1}{2}}
    {\Ass{y}{2}}  \big\} \\[3pt]
    & \qquad \big\{ \PChoice{\Ass{y}{0}} {\nicefrac{4}{5}} {\Ass{y}{3}} \big\};\\[2pt]
  _3  & \Observe \: y=0
\end{array}
$$
Assume we want to compute the conditional expected value of the expression
$10{+}x$ given that the observation $y{=}0$ is passed. This expected value is
given by $\qcwp[P'](10{+}x)$ and the computation of $\cwp[P'](10{+}x,
\CteFun{1})$ goes as follows:
\begin{align*}
\MoveEqLeft[1]
\cwp[P'](10{+}x, \CteFun{1}) \\
& =
\cwp[P'_{1\text{-}2}] (\cwp[\Observe \: y=0](10{+}x, \CteFun{1})) \\
& =
\cwp[P'_{1\text{-}2}](f,g) ~\text{where}~ (f,g)= \ToExp{y=0} \cdot (10{+}x, \CteFun{1})\\
& =
\cwp[P'_{1\text{-}1}] (\cwp[\Cond{x{=}1}{\ldots}{\ldots}](f,g)) \\
& =
\cwp[P'_{{1\text{-}1}}] (\ToExp{x{=}1} \cdot (h,i) + \ToExp{x{\neq}1} \cdot
(h',i'))~~\text{where} \\
& \: \:
\begin{aligned}
 (h,i) &= \cwp[\PChoice{\Ass{y\!}{\!0}} {\nicefrac{1}{2}} {\Ass{y\!}{\!2}}](f,g) \\  
&= \tfrac{1}{2} \cdot (10 + x,\, \ExpOne)~, \text{ and }\\
 (h',i') &= \cwp[\PChoice{\Ass{y\!}{\!0}} {\nicefrac{4}{5}} {\Ass{y\!}{\!3}}
 ](f,g)\\
&= \tfrac{4}{5} \cdot (10 + x,\, \ExpOne)
\end{aligned}\\
&= \tfrac{1}{2} \cdot \tfrac 4 5 \cdot (\CteFun{10} + 0,\,\ExpOne) + \tfrac{1}{2} \cdot \tfrac 1 2 \cdot (\CteFun{10} + 1,\, \ExpOne) \\
& =
\left(\CteFun{4},\, \CteFun{\tfrac 2 5}\right) + \left( \CteFun{\tfrac{11}{4}},\, \CteFun{\tfrac 1 4}\right)
= 
\left(\CteFun{\tfrac{27}{4}},\CteFun{\tfrac{13}{20}} \right)
\end{align*}
Then $\qcwp[P'](10{+}x) =\CteFun{\tfrac{135}{13}}$ and the conditional expected
value of $10{+}x$ is approximately $10.38$.
\end{example}

In the rest of this section we investigate some properties of the
expectation transformer semantics of \cfpGCL.  As every fully probabilistic
\pGCL program is contained in \cfpGCL, we first study the relation of the
\cwllp-- to the \wllp--semantics of \pGCL. To that end, we extend the weakest
(liberal) pre--expectation operator to \cpGCL as follows:
\[
\wp[\Observe \, G](f) = \ToExp{G} \cdot  f  
\quad
\wlp[\Observe \, G](f) = \ToExp{G} \cdot  f~.
\]

To relate the \cwllp-- and \wllp--semantics we heavily rely on the following
result which says that \cwp (resp.~\cwlp) can be \emph{decoupled} as the product
$\wp \times \wlp$ (resp.~$\wlp \times \wlp$).
\begin{theorem}[Decoupling of \cwllp]
\label{thm:cwp-decuop}
For $P \in \cfpGCL$, $f \in \Ex$, and $f',g \in \BEx$:
\begin{align*}
\cwp [P] (f,\, g) ~&=~ \big(\wp[P](f),\, \wlp[P](g)\big)\\
\cwlp [P] (f',\, g) ~&=~ \big(\wlp[P](f),\, \wlp[P](g)\big)
\end{align*}
\end{theorem}
\begin{proof}
By induction on the program structure. See Appendix~\ref{sec:proof-thm-sepfix}
for details.
\end{proof}
\noindent
Let \fpGCL denote the fully probabilistic fragment of \pGCL.
We show that the \qcwp--semantics is a \emph{conservative extension} of the \wp--semantics 
for \fpGCL.  
The same applies to the weakest liberal pre--expectation semantics.  
\begin{theorem}[Compatibility with the \wllp--semantics]
\label{thm:cwp-ext-wp}
For $P \in \fpGCL$, $f \in \Ex$, and $g \in \BEx$:
\begin{align*}
\wp[P](f) = \qcwp[P](f) 
\quad \text{and} \quad
\wlp[P](g) = \qcwlp[P](g)
\end{align*}
\end{theorem}
\begin{proof}
  By Theorem~\ref{thm:cwp-decuop} and the fact that $\qcwlp[P](\CteFun{1}) =
  \CteFun{1}$ (see Lemma~\ref{thm:cwp-healthy-cond}).
\end{proof}
\noindent
We now investigate some elementary properties of \qcwp and \qcwlp such as monotonicity
and linearity.
\begin{lemma}[Elementary properties of $\qcwp$ and $\qcwlp$]
\label{thm:cwp-healthy-cond}
For every $P \in \cfpGCL$ with at least one feasible execution (from every initial state), 
post--expectations $f,g \in \Ex$ and non--negative real constants $\alpha,\beta$:
\begin{enumerate}[label=\roman*), itemsep=1ex]
\item \label{thm:cwp-monot}
$f \sqsubseteq g$ implies $\qcwp[P](f) \sqsubseteq \qcwp[P](g)$ and likewise for \qcwlp.
\item \label{thm:cwp-linear} 
$\qcwp[P](\alpha \cdot f + \beta \cdot g) = 
\alpha \cdot \qcwp[P](f) + \beta \cdot \qcwp[P](g)$.
%
\item \label{thm:cwp-other} 
$\qcwp[P](\CteFun{0}) = \CteFun{0}$ and $\qcwlp[P](\CteFun{1}) = \CteFun{1}$.
\end{enumerate}
\end{lemma}
\begin{proof}
Using Theorem~\ref{thm:cwp-decuop} one can show that the transformers \qcwp/\qcwlp 
inherit these properties from the transformers \wp/\wlp.  
For details we refer to Appendix~\ref{proof:cwp-healthy}.
\end{proof}
\noindent
We conclude this section by discussing alternative approaches for providing an expectation 
transformer semantics for $P\in\cfpGCL$. By Theorem~\ref{thm:cwp-decuop}, 
the transformers $\qcwlp[P]$ and $\qcwlp[P]$ can be recast as:
\[
f \mapsto \frac{\wp[P](f)}{\wlp[P](\CteFun{1})}
\quad \text{and} \quad
f \mapsto \frac{\wlp[P](f)}{\wlp[P](\CteFun{1})}~,
\]
respectively.  Recall that $\wlp[P](\CteFun{1})$ yields the weakest pre--expectation under 
which $P$ either does not terminate or does terminate while passing all \Observe--statements.
An alternative is to normalize using \wp in the denominator instead of \wlp, yielding:
\[
f \mapsto \frac{\wp[P](f)}{\wp[P](\CteFun{1})}
\qquad \text{and} \qquad
f \mapsto \frac{\wlp[P](f)}{\wp[P](\CteFun{1})}
\]
The transformer on the right is not meaningful, as the denominator
$\wp[P](\CteFun{1})(\sigma)$ may be smaller than the numerator
$\wlp[P](f)(\sigma)$ for some state $\sigma\in\State$. 
This would lead to probabilities exceeding one. 
The transformer on the left normalizes \wrt the terminating executions.
This interpretation corresponds to the semantics of the probabilistic programming language {\tt R2}~\cite{Nori:2014,Hur:2014} and is only meaningful if programs terminate almost surely (\ie with probability one).

A noteworthy consequence of adopting this semantics is that $\Observe \, G$ is equivalent to $\WhileDo{\lnot G}{\Skip}$~\cite{Hur:2014}\label{par:othersemantics}, see the discussion in Section~\ref{sec:app}.

Let us briefly compare the four alternatives. To that end consider the 
program $P$ below
\begin{align*}
 & \bigl\{\Abort\bigr\} \, [ \nicefrac{1}{2} ] 
 \bigl\{ \PChoice{\Ass{x}{0}} {\nicefrac{1}{2}} {\Ass{x}{1}}; \\[3pt]
  &  \quad  \PChoice{\Ass{y}{0}} {\nicefrac{1}{2}} {\Ass{y}{1}}; \, \Observe \, x=0
  \lor y=0  \bigr\}
\end{align*}
$P$ tosses a fair coin and according to the outcome either diverges or tosses a 
fair coin twice and observes at least once heads ($y{=}0 \vee x{=}0$). We measure the probability 
that the outcome of the last coin toss was heads according to each transformer:
{\small
\begin{align*}
\frac{\wp[P](\ToExp{y{=}0})}{\wlp[P](\CteFun{1})} 
& = \frac{2}{7}
& \frac{\wlp[P](\ToExp{y{=}0})}{\wlp[P](\CteFun{1})} 
& =\frac{6}{7} \\[4pt]
\frac{\wp[P](\ToExp{y{=}0})}{\wp[P](\CteFun{1})} 
& = \frac{2}{3}
&  \frac{\wlp[P](\ToExp{y{=}0})}{\wp[P](\CteFun{1})} 
& = 2
\end{align*}
}
As mentioned before, the transformer $f \mapsto \frac{\wlp[P](f)}{\wp[P](\CteFun{1})}$ is not significant as 
it yields a ``probability'' exceeding one.  Note that our \cwp--semantics yields a probability of $y{=}0$ on termination---while passing all \Observe--statements---of $\frac{2}{7}$.  As shown before, this is a conservative and natural extension of the $\wp$--semantics.  This does not apply to the {\tt R2}--semantics, as this would require an adaptation of rules for \Abort and \While.

%% file: equiv.tex
We now investigate the connection between the operational semantics of Section~\ref{sec:ope} (for fully
probabilistic programs) and the \cwp--semantics.
We start with some auxiliary results.
The first result establishes a relation between (liberal) expected rewards and weakest (liberal) pre--expectations. 

\begin{lemma}
\label{lem:ExpRew_is_wp}
For $P \in \cfpGCL$, $f \in \Ex, g \in \BEx$, and $\sigma \in \State$:
\begin{align*}
\ExpRew{\Rdtmc{\sigma}{f}{P}}{\lozenge \sink} &= \wp[P](f)(\sigma) \tag{i}\\
\LExpRew{\Rdtmc{\sigma}{g}{P}}{\lozenge \sink} &= \wlp[P](g)(\sigma)\tag{ii}
\label{eq:elr}
\end{align*}
\end{lemma}

\begin{proof} By induction on $P$, see Appendix~\ref{proofof:ExpRew_is_wp} and \ref{proofof:LExpRew_is_wlp}.  \end{proof}
\noindent
The next result establishes that the probability to never reach \undesired in the RMC of program 
$P$ coincides with the weakest liberal pre--expectation of $P$ w.r.t. post--expectation \ExpOne:

\begin{lemma}
\label{lem:Pr_not_bad_is_wlp_1}
For $P \in \cfpGCL$, $g \in \BEx$, and $\sigma \in \State$:
\begin{align*}
\Prr{\Rdtmc{\sigma}{g}{P}}({\neg \lozenge \bad}) ~=~ \wlp[P](\ExpOne)(\sigma)
\end{align*}
\end{lemma}

\begin{proof}
See Appendix \ref{proofof:lem:Pr_not_bad_is_wlp_1}
\end{proof}
\noindent
We now have all prerequisites in order to present the main result of this section: the correspondence between
the operational and expectation transformer semantics of \cfpGCL programs.
It turns out that the weakest (liberal) pre--ex{\-}pec{\-}ta{\-}tion $\qcwp[P](f)(\sigma)$ (respectively $\qcwlp[P](f)(\sigma)$) coincides with the conditional (liberal) expected 
reward in the RMC $\Rdtmc{\sigma}{f}{P}$ of terminating while never violating an \Observe-statement, i.e., avoiding the $\undesired$ states.

\begin{theorem}[Correspondence theorem]
\label{thm:correspondence:cwp}
For $P \in \cfpGCL$, $f \in \Ex$, $g \in \BEx$ and $\sigma \in \State$,
\begin{align*}
\CExpRew{\Rdtmc{\sigma}{f}{P}}{\lozenge \sinklabel}{\neg\lozenge\bad} ~&=~ \qcwp[P](f)(\sigma) \\
\CLExpRew{\Rdtmc{\sigma}{g}{P}}{\lozenge \sinklabel}{\neg\lozenge\bad} ~&=~ \qcwlp[P](g) (\sigma)~.
\end{align*}
\end{theorem}

\begin{proof}
The proof makes use of Lemmas \ref{lem:ExpRew_is_wp}, \ref{lem:Pr_not_bad_is_wlp_1}, and Theorem~\ref{thm:cwp-decuop}. For details see Appendix \ref{proofof:thm:correspondence:cwp}.
\end{proof}
\noindent
Theorem~\ref{thm:correspondence:cwp} extends a previous result~\cite{DBLP:journals/pe/GretzKM14} that established a connection between an operational and the \wp/\wlp semantics for \pGCL programs to the fully probabilistic fragment of \cpGCL.

%% file: applications.tex
In this section we study approaches that make use of our semantics in order to analyze fully probabilistic programs with observations.
We first present a program transformation based on \emph{hoisting} \Observe statements in a way that probabilities of conditions are extracted, allowing for a subsequent analysis on an observation--free program.
Furthermore, we discuss how observations can be replaced by loops and vice versa. 
Finally, we use a well--known case study to demonstrate the direct applicability of our \cwp--semantics.

\subsection{Observation Hoisting}\label{sec:obs-hoisiting}
\input{hoisting}
\subsection{Replacing Observations by Loops}\label{sec:obs-to-loop}
\input{sugar}

\subsection{Replacing Loops by Observations}\label{sec:iid}
\input{iid}

\subsection{The \texttt{Crowds} Protocol}\label{sec:crowds}
\input{crowds}

%% file: hoisting.tex

In what follows we give a semantics--preserving transformation for
removing observations from $\cfpGCL$ programs. Intuitively, the program
transformation ``hoists'' the \Observe statements while updating
the probabilities in case of probabilistic choices. Given $P\in\cfpGCL$, the
transformation delivers a semantically equivalent \Observe--free program
$\hat{P}\in\fpGCL$ and---as a side product---an expectation $\hat{h} \in\BEx$ that
captures the probability of the original program to establish all \Observe
statements. For intuition, reconsider the program from
Example~\ref{ex:run-ex-cwp}. The transformation yields the program
$$
 \begin{array}{ll}
    & \PChoice{\Ass{x}{0}} {\nicefrac{8}{13}} {\Ass{x}{1}};\\[3pt]
    & \Ite \: (x=1) \:  \{  \PChoice{\Ass{y}{0}} {1}
    {\Ass{y}{2}}  \} \\[3pt]
    & \qquad  \{ \PChoice{\Ass{y}{0}} {1} {\Ass{y}{3}} \}
  \end{array}
$$
and expectation $\hat{h} = \CteFun{\tfrac{13}{20}}$. By eliminating dead code in both probabilistic choices and coalescing the branches in the conditional, we can simplify the program to:
\[
\PChoice{\Ass{x}{0}} {\nicefrac{8}{13}} {\Ass{x}{1}}; \, \Ass{y}{0}
\]
As a sanity check note that the expected value of $10+x$ in this program is equal to $10 \cdot \tfrac{8}{13}+11 \cdot \frac{5}{13}=\tfrac{135}{13}$, which agrees with the result obtained in Example~\ref{ex:run-ex-cwp} by analyzing the original program.
Formally, the program transformation is given by a function
\[
\Tr : \cfpGCL \times \BEx \To \cfpGCL \times \BEx~.
\]
To apply the transformation to a program $P$ we need to determine $\Tr
(P,\CteFun{1})$, which gives the semantically equivalent program $\hat{P}$ and
the expectation $\hat{h}$.

The transformation is defined in Figure~\ref{fig:prog-transf} and works by
inductively computing the weakest pre--expectation that guarantees the
establishment of all \Observe statements and updating the probability parameter
of probabilistic choices so that the pre--expectations of their branches are
established in accordance with the original probability parameter. The
computation of these pre--expectations is performed following the same rules as
the \wlp operator. The correctness of the transformation is established
by the following Theorem, which states that a program and its transformed version
share the same terminating and non--terminating behavior.
\begin{theorem}[Program Transformation Correctness]\label{thm:prog-trans-sound}
  Let $P \in \cfpGCL$ admit at least one feasible run for every initial state and $\Tr (P,\CteFun{1}) = (\hat{P},\,\hat{h})$. Then for any $f 
  \in \Ex$ and $g \in \BEx$,
\begin{align*}
\wp[\hat{P}](f) = \qcwp[P](f) 
\quad \text{and} \quad
\wlp[\hat{P}](g) = \qcwlp[P](g).
\end{align*}
\end{theorem} 
\begin{proof} See Appendix~\ref{sec:proof-hoist-trans}.
\end{proof}

\noindent A similar program transformation has been given in~\cite{Nori:2014}. Whereas 
they use random assignments to introduce randomization in their programming 
model, we use probabilistic choices.  Consequently, they can hoist \Observe statements 
only until the occurrence of a random assignment, while we are able to hoist \Observe
statements through probabilistic choices and completely remove them from
programs. Another difference is that their semantics only accounts for terminating 
program behaviors and thus can guarantee the correctness of the program 
transformation for terminating behaviors only. Our semantics is more expressive and
enables establishing the correctness of the program transformation for non--terminating
program behavior, too.
%
%
\begin{figure*}[ht]
\hrule
$$
\begin{array}{l@{\;\;}c@{\;\;}l}
\\[-5pt]
\Tr(\Skip, f)  &=&  (\Skip, f) \\[3pt]
\Tr(\Abort, f)  &=& (\Abort, \CteFun{1}) \\[3pt]
\Tr(\Ass{x}{E}, f)  &=&  (\Ass{x}{E}, f [E/x]) \\[3pt]
\Tr(\Observe \, G, f)  &=&  (\Skip, \ToExp{G} \cdot f) \\[3pt]
\Tr(\Cond{G}{P}{Q}, f)  &=& (\Cond{G}{P'}{Q'}, 
\ToExp{G} \cdot f_P + \ToExp{\lnot G} \cdot f_Q) \\
&& \text{where }  (P',f_P) = \Tr(P,f), \; 
(Q',f_Q) = \Tr(Q,f)  \\[3pt]
\Tr(\PChoice{P}{p}{Q}, f)  &=&  (\PChoice{P'}{p'}{Q'}, p \cdot f_P +
(\CteFun{1}{-}p) \cdot f_Q) \\
&& \text{where }  (P',f_P) = \Tr(P,f), \; (Q',f_Q) =
\Tr(Q,f), \text{and } p' = \frac{p \cdot f_P }{p \cdot f_P + (\CteFun{1}{-}p) \cdot f_Q} \\[3pt]
\Tr(\WhileDo{G}{P}, f) & = & (\WhileDo{G}{P'}, f') \\
&& \text{where }  f' = \gfp X \mydot \left( \ToExp{G} \cdot (\pi_2 \circ\Tr) (P,X) + \ToExp{\lnot
  G} \cdot f \right), \text{and }  (P',\ourunderscore) = \Tr(P,f') \\[3pt]
%
\Tr(P;Q, f)  &=& (P';Q', f'') ~ \text{where} ~ (Q',f') =
\Tr(Q,f), \; (P',f'') = \Tr(P,f') 
\end{array}
$$
\caption{Program transformation for eliminating \Observe statements in \cfpGCL.}
\label{fig:prog-transf}
\end{figure*}

%% file: sugar.tex
For semantics that normalize with respect to the terminating behavior of programs, \Observe statements can  readily be replaced by a loop~\cite{DBLP:conf/sigsoft/ClaretRNGB13,Hur:2014}.
%
%
%
%
In our setting a more intricate transformation is required to eliminate
observations from programs. Briefly stated, the idea is to restart a violating
run from the initial state until it satisfies all encountered observations.
To achieve this we consider a fresh variable \continue and
transform a given program $P\in\cfpGCL$ into a new program $P'$ as described below: 
\needspace{2\baselineskip}
\begin{align*}
\Observe\, G \quad\to&\quad \Cond{\neg G}{\Ass{\continue}{\true}}{\Skip}\\
\Abort \quad\to&\quad \Cond{\neg \continue}{\Abort}{\Skip} \\
\WhileDo{G}{\ldots} \quad\to&\quad \WhileDo{G \wedge \neg \continue}{\ldots}
\end{align*}
For conditional and probabilistic choice, we apply the above rules recursively
to the subprograms.

The aim of the transformation is twofold. First, the program $P'$ flags the
violation of an \Observe statement through the variable \continue. If a violation
occurs, $\continue$ is set to \true while in contrast to the original program we
continue the program execution. As a side effect, we may introduce some
subsequent diverging behavior which would not be present in the original
program (since the execution would have already been blocked). The second aim of the
transformation is to avoid this possible diverging--behavior. This is achieved
by blocking \While--loops and \Abort statements once $\continue$ is set to
$\true$.


Now we can get rid of the observations in $P$ by repeatedly executing $P'$ from
the same initial state till \continue is set to \false (which would intuitively
correspond to $P$ passing all its observations).

This is implemented by program $P''$ below:
\begin{align*}
{}&s_1,\ldots,s_n \coloneqq x_1,\ldots,x_n;~\continue\coloneqq\true;\\
{}&\While(\continue)\,\{\,x_1,\ldots,x_n \coloneqq s_1,\ldots,s_n;~P'\,\}
\end{align*}
Here, $s_1,\ldots,s_n$ are fresh variables and $x_1,\ldots,x_n$ are all program
variables of $P$.  The first assignment stores the initial state in the
variables $s_i$ and the first line of the loop body, ensures that the loop
always starts with the same (initial) values.

\begin{theorem}[]
\label{thm:sugar}
Let programs $P$ and $P''$ be as above. Then 
\[\qcwp[P](f)= \wp[P''](f)\enspace.\]
\end{theorem}
\begin{proof}
See Appendix \ref{sec:proof-sugar}.
\end{proof}
\begin{example}
\label{ex:transform}
Consider the following \cpGCL program:
\begin{align*}
&{}\PChoice{x\coloneqq 0}{p}{x\coloneqq 1};~\PChoice{y\coloneqq 0}{p}{y\coloneqq 1}\\
&{}\Observe\, x \neq y;
\end{align*}
We apply the program transformation to it and obtain:
\begin{align*}
{}&s_1,s_2 \coloneqq x,y;~\continue\coloneqq\true;\\
{}&\While(\continue)\{\\
{}&\quad x,y \coloneqq s_1,s_2;~\continue \coloneqq \false;\\
&{}\quad\PChoice{\Ass x  0}{p}{\Ass x 1};\\
&{}\quad\PChoice{\Ass y 0}{p}{\Ass y 1};\\
&{}\quad\If(x = y)\{\, \continue \coloneqq \true\}\\
{}&\}
\end{align*}
This program is simplified by a data flow analysis: The variables $s_1$ and $s_2$ are irrelevant because $x$ and $y$ are overwritten in every iteration.
Furthermore, there is only one observation so that its predicate can be pushed directly into the loop's guard.
Then the initial values of $x$ and $y$ may be arbitrary but they must be equal to make sure the loop is entered.
This gives the final result
\begin{align*}
{}&x,y \coloneqq 0,0;\\
{}&\While(x = y)\{\\
&{}\quad\PChoice{x \coloneqq 0}{p}{x \coloneqq 1};~\PChoice{y \coloneqq 0}{p}{y \coloneqq 1}\\
{}&\}
\end{align*}
This program is a simple algorithm that repeatedly uses a biased coin to simulate an unbiased coin flip.
A proof that $x$ is indeed distributed uniformly over $\{0,1\}$ has been previously shown \eg in~\cite{DBLP:conf/qest/GretzKM13}.
\end{example}
Theorem~\ref{thm:sugar} shows how to define and effectively calculate the conditional expectation using a straightforward program transformation and the well established notion of \wp.
However in practice it will often be infeasible to calculate the fixed point of the outer loop or to find a suitable loop invariant -- even though it exists. 

%% file: iid.tex
In this section we provide an overview on how the aforementioned result can be ``applied backwards'' in order to replace a loop by an \Observe statement. 
This is useful as it is easier to analyze a loop--free program with observations than a program with loops for which fixed points need to be determined.

The transformation presented in Section~\ref{sec:obs-to-loop} yields programs of a certain form: In every loop iteration the variable values are initialized independently from their values after the previous iteration. Hence the loop iterations generate a sequence of program variable valuations that are \emph{independent and identically distributed} (iid loop), cf. Example~\ref{ex:transform} where no ``data flow'' between iterations of the loop occurs.

In general, if $\Loop = \While(G)\{P\}$ is an iid loop we can obtain a program $Q = P; \Observe\, \neg G$ with
\[\wp[\Loop](f) = \qcwp[Q](f)\enspace\]
for any expectation $f\in\Ex$.
To see this, apply Theorem~\ref{thm:sugar} to program $Q$.
Let the resulting program be \Loop'.
As in Example~\ref{ex:transform}, note that there is only one \Observe statement at the end of \Loop' and furthermore there is no data flow between iterations of \Loop'.
Hence by the same simplification steps we arrive at the desired program \Loop.

%% file: crowds.tex
To demonstrate the applicability of the \cwp-semantics to a practical example, consider the \texttt{Crowds}-protocol~\cite{DBLP:journals/tissec/ReiterR98}.
A set of nodes forms a fully connected network called the \emph{crowd}.
Crowd members would like to exchange messages with a server without revealing their identity to the server.
To achieve this, a node \emph{initiates communication} by sending its message to a randomly chosen crowd member, possibly itself.
Upon receiving a message a node probabilistically decides to either \emph{forward} the message once again to a randomly chosen node in the network or to relay it to the server directly.
A commonly studied attack scenario is that some malicious nodes called \emph{collaborators} join the crowd and participate in the protocol with the aim to reveal the identity of the sender.
The following \cpGCL-program $P$ models this protocol where $p$ is the forward probability and $c$ is the fraction of collaborating nodes in the crowd. The initialization corresponds to the communication initiation.
\begin{align*}
\Init:\quad{}&\PChoice{\Ass {\mathit{intercepted}} {1}}{c}{\Ass {\mathit{intercepted}} {0}};\\
{}&\Ass {\mathit{delivered}} {0};~\Ass {\mathit{counter}} {1}\\
\Loop:\quad{}&\While(\mathit{delivered}=0)\,\{\\
{}&\qquad \bigl\{\Ass {\mathit{counter}} {\mathit{counter}+1};\\ 
{}&\qquad \PChoice{\Ass {\mathit{intercepted}} {1}}{c}{\Skip}\bigr\}\\
{}&\qquad    [p]\\
{}&\qquad    \{\Ass {\mathit{delivered}} {1}\}\\
{}&    \};\\
{}&\Observe ({\mathit{counter}\leq k})
\end{align*}
Our goal is to determine the probability of a message not being intercepted by a collaborator.
We condition this by the observation that a message is forwarded at most $k$ times.

Note that the operational semantics of $P$ produce an \emph{infinite parametric RMC} since the value of $k$ is fixed but arbitrary.
Using Theorem~\ref{thm:cwp-decuop} we express the probability that a message is not intercepted given that it was rerouted no more than $k$ times by
\begin{equation}
\qcwp[P]([\neg\textit{intercepted}]) = \frac{\wp[P]([\neg\textit{intercepted}])}{\wlp[P](\CteFun{1})}
\label{eqn:goal}
\end{equation}
The computation of this quantity requires to find fixed points, \cf Appendix~\ref{sec:calculation} for details.
As a result we obtain a closed form solution parametrized in $p$, $c$, and $k$:
\begin{align*}
(1-c)(1-p)\frac{1-(p(1-c))^k}{1-p(1-c)}\cdot\frac{1}{1-p^k}
\end{align*}

The automation of such analyses remains a challenge and is part of ongoing and future work.

%% file: nondet.tex

In this section we argue why (under mild assumptions) it is not possible to come up with a denotational semantics in the style of conditional pre--expectation transformers (CPETs for short) for full \cpGCL.
To show this, it suffices to consider a simple fragment of \cpGCL containing only assignments, observations, probabilistic and non--deterministic choices. 
Let $x$ be the only program variable that can be written or read in this fragment.
We denote this fragment by $\cpGCL^{-}$. 
Assume $D$ is some appropriate domain for \emph{representing} conditional expectations of the program variable $x$ with respect to some \emph{fixed} initial state $\sigma_0$ and let $\llbracket \:\cdot\: \rrbracket \colon D \rightarrow \mathbb R \cup \{\bot\}$ be an interpretation function such that for any $d \in D$ we have that $\llbracket d \rrbracket$ \emph{is equal to} the (possibly undefined) conditional expected value of $x$.

\begin{definition}[Inductive CPETs]
\label{def:inductiveCPETs}
A \emph{CPET} is a function $\cwp^*\colon \cpGCL^{-} \rightarrow D$ such that for any $P \in \cpGCL^{-}$, $\llbracket \cwp[P] \rrbracket = \CExpRew{\Rmdp {\sigma_0} x P}{\Finally \sinklabel}{\neg\Finally\bad}$.
$\cwp^*$ is called \emph{inductive}, if there exists some function $\mathcal K \colon \cpGCL^{-} \times [0,\, 1] \times \cpGCL^{-} \rightarrow D$ such that for any $P_1, P_2\in\cpGCL^{-}$,
\begin{align*}
\cwp^*[\PChoice {P_1} p {P_2}] ~=~ \mathcal K(\cwp^*[P_1],\, p,\, \cwp^*[P_2])~,
\end{align*}
and some function $\mathcal N\colon \cpGCL^{-} \times \cpGCL^{-} \rightarrow D$ with
\begin{align*}
\cwp^*[\NDChoice{P_1}{P_2}] ~=~ \mathcal N(\cwp^*[P_1],\, \cwp^*[P_2])~,
\end{align*}
where $\forall d_1,d_2 \in D\colon \mathcal N(d_1,\, d_2)  ~\in~ \{d_1,\, d_2\}$.
\end{definition}
\noindent
This definition suggests that the conditional pre--ex{\-}pec{\-}ta{\-}tion of $\PChoice{P_1}{p}{P_2}$ is determined only by the conditional pre--expectation of $P_1$, the conditional pre--expectation of $P_2$, and the probability $p$.
Furthermore the above definition suggests that the conditional pre--expectation of $\NDChoice{P_1}{P_2}$ is also determined by the conditional pre--expectation of $P_1$ and the conditional pre--expectation of $P_2$ only. 
Consequently, the non--deterministic choice can be resolved by replacing it either by $P_1$ or $P_2$. 
While this might seem like a strong limitation, the above definition is compatible with the interpretation of non--deterministic choice as demonic choice: 
The choice is deterministically driven towards the worst option.
The requirement $N(d_1,\, d_2)  ~\in~ \{d_1,\, d_2\}$ is also necessary for interpreting non--deterministic choice as an abstraction where implementational details are not important.

As we assume a fixed initial state and a fixed post--expectation, the non--deterministic choice turns out to be deterministic once the pre--expectations of $P_1$ and $P_2$ are known.
Under the above assumptions (which do apply to the \wp and \wlp transformers) we claim:

\begin{theorem}
\label{thm:nocpet}
There exists no inductive CPET.
\end{theorem}
%
\begin{proof}
The proof goes by contradiction.
Consider the program $P = \PChoice{P_1}{\nicefrac 1 2}{P_5}$ with 
\begin{align*}
P_1 ~&=~ \Ass x 1\\
P_5 ~&=~ \NDChoice{P_2}{P_4}\\
P_2 ~&=~ \Ass x 2\\
P_4 ~&=~ \PChoice{\Observe~\texttt{false}}{\nicefrac 1 2}{P_{2+\varepsilon}}\\
P_{2+\varepsilon} ~&=~ \Ass x {2 + \varepsilon}~,
\end{align*}
where $\varepsilon > 0$.
A schematic depiction of the RMDP $\Rmdp{\sigma_0}{x}{P}$ is given in Figure~\ref{fig:rmdpschem_inpaper}.
\begin{figure}[b]
	\centering
    \input{pics/nondetproof}
    \caption{Schematic depiction of the RMDP $\Rmdp{\sigma_0}{x}{P}$}
    \label{fig:rmdpschem_inpaper}
\end{figure}
Assume there exists an inductive CPET $\cwp^*$ over some appropriate domain $D$. Then,
\begin{align*}
\cwp^*[P_1] ~&=~ d_1, \text{ with } \llbracket d_1 \rrbracket = 1\\
\cwp^*[P_2] ~&=~ d_2, \text{ with } \llbracket d_2 \rrbracket = 2\\
\cwp^*[P_{2 + \varepsilon}] ~&=~ d_{2 + \varepsilon}, \text{ with } \llbracket d_{2 + \varepsilon} \rrbracket = {2 + \varepsilon}\\
\cwp^*[\texttt{observe false}] ~&=~ \mathfrak{of}, \text{ with } \llbracket \mathfrak{of} \rrbracket = \bot
\end{align*}
for some appropriate $d_1,\, d_2 ,\, d_{2 + \varepsilon},\, \mathfrak{of} \in D$.
By Definition \ref{def:inductiveCPETs}, $\cwp^*$ being inductive requires the existence of a function $\mathcal K$, such that
\begin{align*}
\cwp^*[P_4] ~=~ &\mathcal{K}\left(\cwp^*[\Observe~\texttt{false}],\, \nicefrac{1}{2},\, \cwp^*[P_{2+\varepsilon}]\right)\\
~=~ &\mathcal{K}(\mathfrak{of},\, \nicefrac 1 2,\, d_{2 + \varepsilon})~.
\end{align*}
In addition, there must be an $\mathcal N$ with:
\begin{align*}
\cwp^*[P_5] ~=~ &\mathcal{N}\left(\cwp^*[P_2],\, \cwp^*[P_4]\right)\\
~=~ &\mathcal N (d_2,\, \mathcal{K}(\mathfrak{of},\, \nicefrac 1 2,\, d_{2 + \varepsilon}))~.
\end{align*}
Since $P_4$ is a probabilistic choice between an infeasible branch and $P_{2 + \varepsilon}$, the expected value for $x$ has to be rescaled to the feasible branch. 
Hence $P_4$ yields $\llbracket \cwp^*[P_4] \rrbracket = 2 + \varepsilon$, whereas $\llbracket \cwp^*[P_2] \rrbracket = 2$. 
Thus:
\begin{align}
\llbracket d_2 \rrbracket ~\lneq~ \llbracket \mathcal{K}(\mathfrak{of},\, \nicefrac 1 2,\, d_{2 + \varepsilon})\rrbracket \label{x2neqprob_inpaper}
\end{align}
As non--deterministic choice is demonic, we have:
\begin{align}
\cwp^*[P_5] ~=~ \mathcal N (d_2,\, \mathcal{K}(\mathfrak{of},\, \nicefrac 1 2,\, d_{2 + \varepsilon})) ~=~ d_{2}\label{P4eqx2_inpaper}
\end{align}
As $\mathcal{N}\left(\cwp^*[P_2],\, \cwp^*[P_4]\right) \in \{\cwp^*[P_2],\, \cwp^*[P_4]\}$ we can resolve non--determinism in $P$ by either rewriting $P$ to $\PChoice{P_1}{\nicefrac 1 2}{P_2}$ which gives 
\begin{align*}
\llbracket \cwp^*\PChoice{P_1}{\nicefrac 1 2}{P_2} \rrbracket ~=~ &\frac{3}{2}~,
\intertext{or we rewrite $P$ to $\PChoice{P_1}{\nicefrac 1 2}{P_4}$, which gives}
\llbracket \cwp^*\PChoice{P_1}{\nicefrac 1 2}{P_4} \rrbracket ~=~ &\frac{4 + \varepsilon}{3}~.
\end{align*} 
For a sufficiently small $\varepsilon$ the second option should be preferred by a demonic scheduler. 
This, however, suggests:
\begin{align*}
\cwp^*[P_5] ~=~ &\mathcal N (d_2,\, \mathcal{K}(\mathfrak{of},\, \nicefrac 1 2,\, d_{2 + \varepsilon}))\\
=~ &\mathcal{K}(\mathfrak{of},\, \nicefrac 1 2,\, d_{2 + \varepsilon})
\end{align*}
Together with Equality (\ref{P4eqx2_inpaper}) we get $d_2 = \mathcal{K}(\mathfrak{of},\, \nicefrac 1 2,\, d_{2 + \varepsilon})$, which implies $\llbracket d_2 \rrbracket = \llbracket \mathcal{K}(\mathfrak{of},\, \nicefrac 1 2,\, d_{2 + \varepsilon})\rrbracket$.
This is a contradiction to Inequality (\ref{x2neqprob_inpaper}).
\end{proof}
As an immediate corollary of Theorem \ref{thm:nocpet} we obtain the following statement:
\begin{corollary}
\label{cor:cwpnondet}
We cannot extend the \cwp rules in Figure \ref{fig:cwp} for non--deterministic programs such that Theorem \ref{thm:correspondence:cwp} extends to full \cpGCL.
\end{corollary}
%
%
\noindent
This result is related to the fact that for minimizing conditional (reachability) probabilities in RMDPs positional, \ie history--independent, schedulers are insufficient~\cite{DBLP:conf/tacas/AndresR08}.
Intuitively speaking, if a \emph{history--dependent} scheduler is required, this necessitates the inductive definition of $\cwp^*$ to take the context of a statement (if any) into account.
This conflicts with the principle of an inductive definition.
Investigating the precise relationship with the result of~\cite{DBLP:conf/tacas/AndresR08} requires further study.

%% file: pics/nondetproof.tex
\scalebox{.75}{\begin{tikzpicture}[->,>=stealth',shorten >=1pt,node distance=2.5cm,semithick]
\tikzstyle{every state}=[draw=none]
  \draw[white, use as bounding box] (-2.1,1.05) rectangle (5.8,-6.2);
   \node [state, initial, initial text=, initial where=above] (P) {$\langle  P \rangle$};  
   \node [state] (P1) [below left of=P, label={[yshift=0.25cm, gray] 270:$1$}] {$\langle P_1 \rangle$};
   \node [state] (P5) [below right of=P] {$\langle P_5 \rangle$};
   \node [state] (P2) [below left of=P5, label={[yshift=0.25cm, gray] 270:$2$}] {$\langle P_2 \rangle$};
   \node [state] (P4) [below right of=P5] {$\langle P_4 \rangle$};
   \node [state] (bad) [below left of=P4] {$\bad$};
   \node [state] (P2e) [below right of=P4, label={[yshift=0.25cm, gray] 270:$2 + \varepsilon$}] {$\langle P_{2 + \varepsilon} \rangle$};

\path [] (P) edge [] node [above left] {$\frac{1}{2}$} (P1);
\path [] (P) edge [] node [above right] {$\frac{1}{2}$} (P5);
\path [] (P5) edge [dotted] (P2);
\path [] (P5) edge [dotted] (P4);
\path [] (P4) edge [] node [above left] {$\frac{1}{2}$} (bad);
\path [] (P4) edge [] node [above right] {$\frac{1}{2}$} (P2e);
\end{tikzpicture}}

%% file: conclusion.tex

This paper presented an extensive treatment of semantic issues in probabilistic programs with conditioning.
Major contributions are the treatment of non--terminating programs (both operationally and for weakest
liberal pre--expectations), our results on combining non--determinism with conditioning, as well as the 
presented program transformations.
We firmly believe that a thorough understanding of these semantic issues provides a main cornerstone for enabling 
automated analysis techniques such as loop invariant synthesis~\cite{DBLP:conf/sas/ChakarovS14,DBLP:conf/sas/KatoenMMM10}, 
program analysis~\cite{DBLP:conf/esop/CousotM12} and model checking~\cite{DBLP:conf/tacas/BaierKKM14} 
to the class of probabilistic programs with conditioning.
Future work consists of investigating conditional invariants and a further investigation of non--determinism 
in combination with conditioning.

%% file: appendix.tex
\input{denoapp}

\input{opapp}
\input{appapp}

\input{crowdsapp}

%% file: denoapp.tex

\subsection{Continuity of \wp and \wlp}

\begin{lemma}[Continuity of \wp/\wlp]
\label{thm:wp-ext-is-cont}
Consider the extension of \wp and \wlp to \cpGCL given by
\begin{align*}
\wp[\Observe \, G](f) ~&=~ \ToExp{G} \cdot  f\\
 \wlp[\Observe \, G](g) ~&=~ \ToExp{G} \cdot  g~.
\end{align*}
Then for every $P \in \cpGCL$ the expectation transformers $\wp[P]\colon\Ex \To \Ex$ and $\wlp[P]\colon\BEx \To \BEx$ are continuous mappings over $(\Ex,\, \sqsubseteq)$ and $(\BEx,\, \sqsupseteq)$, respectively.
\end{lemma}
\begin{proof}
For proving the continuity of $\wp$ we have to show that for any directed subset $D\subseteq \Ex$ we have 
\begin{align}
\label{eq:goal_cont_lem}
\sup_{f \in D} \wp[P](f) ~=~ \wp[P]\left(\sup_{f \in D} f\right)~.
\end{align}
This can be shown by structural induction on $P$.
All cases except for the \Observe statement have been covered in \cite{DBLP:journals/pe/GretzKM14}.
It remains to show that Equality (\ref{eq:goal_cont_lem}) holds for $P=\Observe ~G$:
\begin{align*}
\sup_{f \in D} \wp[\Observe ~G](f) 
&~=~ \sup_{f \in D} \ToExp{G} \cdot  f  \\
&~=~ \ToExp{G} \cdot \sup_{f \in D}  f  \\
&~=~ \wp[\Observe \,G](\sup_{f \in D} f) 
\end{align*}
The proof for the liberal transformer \wlp is analogous.
\end{proof}

\subsection{Proof of Theorem~\ref{thm:cwp-decuop}}
\label{sec:proof-thm-sepfix}

\begingroup
\def\thelemma{\ref{thm:cwp-decuop}}
\begin{theorem}[Decoupling of $\cwp / \cwlp$]
For $P \in \cfpGCL$, $f \in \Ex$, and $f',g \in \BEx$:
\begin{align*}
\cwp [P] (f,\, g) ~&=~ \big(\wp[P](f),\, \wlp[P](g)\big)\\
\cwlp [P] (f',\, g) ~&=~ \big(\wlp[P](f'),\, \wlp[P](g)\big)
\end{align*}
\end{theorem}
\addtocounter{lemma}{-1}
\endgroup

\begin{proof}
The proof of Theorem \ref{thm:cwp-decuop} goes by induction over all \cfpGCL programs.
For the induction base we have:

\paragraph{The Effectless Program \Skip.} 

For \cwp we have:
\begin{align*}
\cwp[\Skip](f,\, g) ~=~ &(f,\, g)\\
~=~&\big(\wp[\Skip](f),\, \wlp[\Skip](g)\big)
\end{align*}
The argument for \cwlp is completely analogous.

\paragraph{The Faulty Program \Abort.} 

For \cwp we have:
\begin{align*}
\cwp[\Abort](f,\, g) ~=~ &(\ExpZero,\, \ExpOne)\\
~=~&\big(\wp[\Abort](f),\, \wlp[\Abort](g)\big)
\end{align*}
Analogously for \cwlp we have:
\begin{align*}
\cwlp[\Abort](f',\, g) ~=~ &(\ExpOne,\, \ExpOne)\\
~=~&\big(\wlp[\Abort](f'),\, \wlp[\Abort](g)\big)
\end{align*}

\paragraph{The Assignment $\Ass{x}{E}$.} 

For \cwp we have:
\begin{align*}
\cwp[\Ass{x}{E}](f,\, g) ~=~ &(f\subst x E,\, g\subst x E)\\
~=~&\big(\wp[\Ass{x}{E}](f),\, \wlp[\Ass{x}{E}](g)\big)
\end{align*}
The argument for \cwlp is completely analogous.

\paragraph{The Observation $\Observe~G$.} 

For \cwp we have:
\begin{align*}
&\cwp[\Observe~G](f,\, g)\\
=~ &(f \cdot \ToExp{G},\, g \cdot \ToExp{G})\\
=~&\big(\wp[\Observe~G](f),\, \wlp[\Observe~G](g)\big)
\end{align*}
The argument for \cwlp is completely analogous.

\paragraph{The Induction Hypothesis:} Assume in the following that for two arbitrary but fixed programs $P, Q \in \cfpGCL$ it holds that both
\begin{align*}
\cwp [P] (f,\, g) ~=~ &\big(\wp[P](f), \wlp[P](g)\big), \textnormal{ and }\\
\cwlp [P] (f',\, g) ~=~ &\big(\wlp[P](f'), \wlp[P](g)\big)~.
\end{align*}
Then for the induction step we have:

\paragraph{The Concatenation $P;\, Q$.} 

For \cwp we have:
\begin{align*}
&\cwp[P;\, Q](f,\, g)\\
=~ &\cwp[P] (\cwp[Q](f,\, g)\\
=~ &\cwp[P] \big(\wp[Q](f),\, \wlp[Q](g)\big)\tag{I.H. on $Q$}\\
=~ &\big(\wp[P](\wp[Q](f)),\, \wlp[P](\wlp[Q](g))\big)\tag{I.H. on $P$}\\
=~&\big(\wp[P;\,Q](f),\, \wlp[P;\,Q](g)\big)
\end{align*}
The argument for \cwlp is completely analogous.

\paragraph{The Conditional Choice $\Cond G P Q$.} 

For \cwp we have:
\begin{align*}
&\cwp[\Cond G P Q](f,\, g)\\
=~ &\ToExp{G} \cdot \cwp[P](f,\, g) + \ToExp{\neg G} \cdot \cwp[Q](f,\, g)\\
=~ &\ToExp{G} \cdot \big(\wp[P](f),\, \wlp[P](g)\big)\tag{I.H.}\\
 &~~~~{} + \ToExp{\neg G} \cdot \big(\wp[Q](f),\, \wlp[Q](g)\big)\\
 =~ &\big(\ToExp{G} \cdot \wp[P](f) + \ToExp{\neg G} \cdot \wp[Q](f),\\
 &~~~~\ToExp{G} \cdot \wlp[P](g) + \ToExp{\neg G} \cdot \wlp[Q](g)\big)\\
=~&\big(\wp[\Cond G P Q](f),\\
 &~~~~\wlp[\Cond G P Q](g)\big)
\end{align*}
The argument for \cwlp is completely analogous.

\paragraph{The Probabilistic Choice $\PChoice P p Q$.} 

For \cwp we have:
\begin{align*}
&\cwp[\PChoice P p Q](f,\, g)\\
=~ &p \cdot \cwp[P](f,\, g) + (1-p) \cdot \cwp[Q](f,\, g)\\
=~ &p\cdot \big(\wp[P](f),\, \wlp[P](g)\big)\tag{I.H.}\\
 &~~~~{} + (1 - p) \cdot \big(\wp[Q](f),\, \wlp[Q](g)\big)\\
 =~ &\big(p \cdot \wp[P](f) + (1-p) \cdot \wp[Q](f),\\
 &~~~~p \cdot \wlp[P](g) + (1 - p) \cdot \wlp[Q](g)\big)\\
=~&\big(\wp[\PChoice P p Q](f),\,\wlp[\PChoice P p Q](g)\big)
\end{align*}
The argument for \cwlp is completely analogous.

\paragraph{The Loop $\WhileDo G P$.} 

For \cwp we have:
\begin{align*}
&\cwp[\WhileDo G P](f,\, g)\\
=~ &\textstyle\lfp_{{\sqsubseteq},{\sqsupseteq}} (X_1,\, X_2)\mydot \ToExp{G} \cdot \cwp[P](X_1,\, X_2) + \ToExp{\neg G} \cdot (f,\, g)\\
=~ &\textstyle\lfp_{{\sqsubseteq},{\sqsupseteq}} (X_1,\, X_2)\mydot \ToExp{G} \cdot \big(\wp[P](X_1),\, \wlp[P](X_2)\big)\\
&~~~~{} + \ToExp{\neg G} \cdot (f,\, g)\tag{I.H.}\\
=~ &\textstyle\lfp_{{\sqsubseteq},{\sqsupseteq}} (X_1,\, X_2)\mydot \big(\ToExp{G} \cdot \wp[P](X_1) + \ToExp{\neg G} \cdot f,\\
&~~~~{} \ToExp{G} \cdot \wlp[P](X_2) + \ToExp{\neg G} \cdot g\big)
\end{align*}
Now let $H(X_1,\, X_2) = \big(\ToExp{G} \cdot \wp[P](X_1) + \ToExp{\neg G} \cdot f,\, \ToExp{G} \cdot \wlp[P](X_2) + \ToExp{\neg G} \cdot g\big)$ and let $H_1(X_1,\, X_2)$ be the projection of $H(X_1$, $X_2)$ to the first component and let $H_2(X_1,\, X_2)$ be the projection of $H(X_1,\, X_2)$ to the second component. 

Notice that the value of $H_1(X_1,\, X_2)$ does not depend on $X_2$ and that it is given by
\begin{align*}
H_1(X_1,\, \ourunderscore) ~=~ \ToExp{G} \cdot \wp[P](X_1) + \ToExp{\neg G} \cdot f~.
\end{align*}
By the continuity of $\wp$ (Lemma \ref{thm:wp-ext-is-cont}) we can establish that $H_1$ is continuous. 
Analogously the value of $H_2(X_1,\, X_2)$ does not depend on $X_1$ and it is given by
\begin{align*}
H_2(\ourunderscore,\, X_2) ~=~ \ToExp{G} \cdot \wlp[P](X_2) + \ToExp{\neg G} \cdot g~.
\end{align*}
By the continuity of $\wlp$ (Lemma \ref{thm:wp-ext-is-cont}) we can establish that $H_2$ is continuous.  

As both $H_1$ and $H_2$ are continuous, we can apply Beki\'{c}'s Theorem \cite{Bekic:1984} which tells us that the least fixed point of $H$ is given as $\left(\widehat{X_1},\, \widehat{X_2}\right)$ with
\begin{align*}
\widehat{X_1} ~=~ &\textstyle\lfp_{\sqsubseteq}X_1\mydot H_1\big(X_1,\, \lfp_{\sqsupseteq}X_2\mydot H_2(X_1,\, X_2)\big)\\
=~ &\textstyle\lfp_{\sqsubseteq}X_1\mydot H_1\big(X_1,\, \ourunderscore\big)\\
=~ &\textstyle\lfp_{\sqsubseteq}X_1\mydot \ToExp{G} \cdot \wp[P](X_1) + \ToExp{\neg G} \cdot f\\
=~ &\wp[\WhileDo G P](f)
\intertext{and}
\widehat{X_2} ~=~ &\textstyle\lfp_{\sqsupseteq}X_2\mydot H_2\big(\lfp_{\sqsupseteq}X_1\mydot H_1(X_1,\, X_2),\, X_2\big)\\
=~ &\textstyle\lfp_{\sqsupseteq}X_2\mydot H_2\big(\ourunderscore,\,  X_2\big)\\
=~ &\textstyle\lfp_{\sqsupseteq}X_2\mydot \ToExp{G} \cdot \wlp[P](X_2) + \ToExp{\neg G} \cdot g\\
=~ &\textstyle\gfp_{\sqsubseteq}X_2\mydot \ToExp{G} \cdot \wlp[P](X_2) + \ToExp{\neg G} \cdot g\\
=~ &\wlp[\WhileDo G P](g)~,
\end{align*}
which gives us in total
\begin{align*}
&\cwp[\WhileDo G P](f,\, g) ~=~ \left(\widehat{X_1},\, \widehat{X_2}\right)\\
 =~ &\big( \wp[\WhileDo G P](f),\, \wlp[\WhileDo G P](f)\big)~.
\end{align*}
The argument for \cwlp is completely analogous.
\end{proof}


\subsection{Linearity of \wp}

\begin{lemma}[Linearity of \wp]
\label{thm:wp-linear}%
For any $P \in \cfpGCL$, any post--expectations
$f, g \in \Ex$ and any non--negative real constants $\alpha,\beta$,
\[
\wp[P](\alpha \cdot f + \beta \cdot g) = \alpha \cdot \wp[P](f) + \beta
  \cdot \wp[P](g)\enspace.
\]
%
\end{lemma}
\begin{proof}
The proof proceeds by induction on the structure of $P$.

\paragraph{The Effectless Program \Skip.} 
\begin{align*}
\MoveEqLeft[2]
\wp[\Skip](\alpha \cdot f + \beta \cdot g) \\
& = 
\alpha \cdot f + \beta \cdot g \\
& = 
\alpha \cdot \wp[\Skip](f) + \beta \cdot \wp[\Skip](g)
\end{align*}

\paragraph{The Faulty Program \Abort.} 
\begin{align*}
\MoveEqLeft[2]
\wp[\Abort](\alpha \cdot f + \beta \cdot g)  \\
& = 
\CteFun{0}\\
&=
\alpha \cdot \wp[\Abort](f) + \beta \cdot \wp[\Abort](g)
\end{align*}

\paragraph{The Assignment $\Ass{x}{E}$.} 
\begin{align*}
\MoveEqLeft[2]
\wp[\Ass{x}{E}](\alpha \cdot f + \beta \cdot g)\\
& = 
(\alpha \cdot f + \beta \cdot g)\subst{x}{E}\\
& = 
\alpha \cdot f \subst{x}{E} + \beta \cdot g \subst{x}{E}\\
& =
\alpha \cdot \wp[\Ass{x}{E}](f) + \beta \cdot \wp[\Ass{x}{E}](g)
\end{align*}

\paragraph{The Observation $\Observe~G$.} 
\begin{align*}
\MoveEqLeft[2]
\wp[\Observe~G] (\alpha \cdot f + \beta \cdot g)\\
& =
\ToExp{G} \cdot (\alpha \cdot f + \beta \cdot g)\\
& = 
\alpha \cdot \ToExp{G} \cdot f + \beta \cdot \ToExp{G} \cdot  g\\
& = 
\alpha \cdot \wp[\Observe~G](f) + \beta \cdot \wp[\Observe~G](g)
\end{align*}

\paragraph{The Concatenation $P;\, Q$.} 
\begin{align*}
\MoveEqLeft[2]
\wp[P; Q](\alpha \cdot f + \beta \cdot g)\\
& =
\wp[P](\wp[Q](\alpha \cdot f + \beta \cdot g))\\
& = 
\wp[P](\alpha \cdot \wp[Q](f) + \beta \cdot \wp[Q](g)) \tag{I.H.\ on $Q$}\\
& = 
\alpha \cdot \wp[P](\wp[Q](f)) \\
& \quad + \beta \cdot \wp[P](\wp[Q](g)) \tag{I.H.\ on $P$}\\
& =
\alpha \cdot \wp[P; Q](f) + \beta \cdot \wp[P; Q](g)
\end{align*}

\paragraph{The Conditional Choice $\Cond G P Q$.} 
\begin{align*}
\MoveEqLeft[2]
\wp[\Cond G P Q](\alpha \cdot f + \beta \cdot g) \\
& = 
\ToExp{G} \cdot \wp[P](\alpha \cdot f + \beta \cdot g) \\
& \quad + \ToExp{\neg G} \cdot \wp[Q](\alpha \cdot f + \beta \cdot g)\\
& = 
\ToExp{G} \cdot (\alpha \cdot \wp[P](f) + \beta \cdot \wp[P](g)) \\
& \quad + \ToExp{\neg G} \cdot (\alpha \cdot \wp[Q](f) + \beta \cdot \wp[Q](g))
\tag{I.H.}\\
& =
\alpha \cdot (\ToExp{G} \cdot \wp[P](f) + \ToExp{\neg G} \cdot \wp[Q](f)) \\
& \quad + \beta \cdot (\ToExp{G} \cdot \wp[P](g) + \ToExp{\neg G} \cdot \wp[Q](g)) \\
& =
\alpha \cdot \wp[\Cond G P Q](f) \\
& \quad + \beta \cdot \wp[\Cond G P Q](g)
\end{align*}

\paragraph{The Probabilistic Choice $\PChoice P p Q$.} 
\begin{align*}
\MoveEqLeft[2]
\wp[\PChoice P p Q](\alpha \cdot f + \beta \cdot g) \\
& = 
p \cdot \wp[P](\alpha \cdot f + \beta \cdot g) \\
& \quad + (\CteFun{1} - p) \cdot \wp[Q](\alpha \cdot f + \beta \cdot g)\\
& = 
p \cdot (\alpha \cdot \wp[P](f) + \beta \cdot \wp[P](g)) \\
& \quad + (\CteFun{1} - p) \cdot (\alpha \cdot \wp[Q](f) + \beta \cdot \wp[Q](g))
\tag{I.H.}\\
& =
\alpha \cdot (p \cdot \wp[P](f) + (\CteFun{1} - p) \cdot \wp[Q](f)) \\
& \quad + \beta \cdot (p \cdot \wp[P](g) + (\CteFun{1} - p) \cdot \wp[Q](g)) \\
& =
\alpha \cdot \wp[\PChoice P p Q](f) \\
& \quad  + \beta \cdot \wp[\PChoice P p Q](g)
\end{align*}

\paragraph{The Loop $\WhileDo G P$.} 
The main idea of the proof is to show that linearity holds for the $n$-th
unrolling of the loop and then use a continuity argument to show that the
property carries over to the loop.

The fact that linearity holds for the $n$--unrolling of the loop is
formalized by formula $H^n(\CteFun{0}) = \alpha \cdot I^n(\CteFun{0}) + \beta
\cdot J^n(\CteFun{0})$, where
\begin{align*}
H(X) &= \ToExp{G} \cdot \wp[P](X) + \ToExp{\neg G} \cdot (\alpha \cdot f + \beta
\cdot g)\\
I(X) &= \ToExp{G} \cdot \wp[P](X) + \ToExp{\neg G} \cdot f \\
J(X) &= \ToExp{G} \cdot \wp[P](X) + \ToExp{\neg G} \cdot g
\end{align*}
We prove this formula by induction on $n$. The base case $n=0$ is immediate. For
the inductive case we reason as follows
\begin{align*}
\MoveEqLeft[1]
H^{n+1}(\CteFun{0}) \\
& = 
H (H^n(\CteFun{0})) \\
& =  
H (\alpha \cdot I^n(\CteFun{0}) + \beta \cdot
J^n(\CteFun{0})) \tag{I.H.\ on $n$}\\
& =  
\ToExp{G} \cdot  \wp[P](\alpha \cdot I^n(\CteFun{0}) + \beta \cdot
J^n(\CteFun{0})) \\
& \quad + \ToExp{\neg G} \cdot (\alpha \cdot f + \beta \cdot g) \\
& =  
\ToExp{G} \cdot  (\alpha \cdot \wp[P](I^n(\CteFun{0})) + \beta \cdot
\wp[P](J^n(\CteFun{0}))) \\
& \quad + \ToExp{\neg G} \cdot (\alpha \cdot f + \beta \cdot g) \tag{I.H.\ on
  $P$}\\
& =
\alpha \cdot  (\ToExp{G} \cdot \wp[P](I^n(\CteFun{0})) + \ToExp{\neg G} \cdot f)
\\
& \quad + \beta \cdot (\ToExp{G} \cdot \wp[P](J^n(\CteFun{0})) + \ToExp{\neg G} \cdot g) \\
& =
\alpha \cdot I(I^n(\CteFun{0})) + \beta \cdot J(J^n(\CteFun{0})) \\
& =
\alpha \cdot I^{n+1}(\CteFun{0}) + \beta \cdot J^{n+1}(\CteFun{0})
\end{align*}

Now we turn to the proof of the main claim. We apply the Kleene Fixed Point Theorem
to deduce that the least fixed points of $H$, $I$ and $J$ can be built by
iteration from expectation $\CteFun{0}$ since the three transformers are
continuous (due to the continuity of $\wp$ established in
Lemma~\ref{thm:wp-ext-is-cont}). Then we have
\belowdisplayskip=-1\baselineskip
\begin{align*}
\MoveEqLeft[2]
\wp[\WhileDo{G}{P}](\alpha \cdot f + \beta \cdot g) \\
& =
\bigsqcup_n H^n(\CteFun{0}) \\
& = 
\bigsqcup_n \alpha \cdot I^n(\CteFun{0}) + \beta \cdot J^n(\CteFun{0}) \\
& = 
\alpha \cdot \bigsqcup_n  I^n(\CteFun{0}) + \beta \cdot \bigsqcup_n
J^n(\CteFun{0}) \\
& = 
\alpha \cdot \wp[\WhileDo{G}{P}](f) \\ 
& \quad + \beta \cdot \wp[\WhileDo{G}{P}](g)
\end{align*}
\end{proof}

\subsection{Proof of  Lemma~\ref{thm:cwp-healthy-cond}}
\label{proof:cwp-healthy}

\begingroup
\def\thelemma{\ref{thm:cwp-healthy-cond}}
\begin{lemma}[Elementary properties of $\qcwp$ and $\qcwlp$]
For every $P \in \cfpGCL$ with at least one feasible execution (from every initial state), 
post--expectations $f,g \in \Ex$ and non--negative real constants $\alpha,\beta$:
\begin{enumerate}[label=\roman*), itemsep=1ex]
\item \label{thm:cwp-monot}
$f \sqsubseteq g$ implies $\qcwp[P](f) \sqsubseteq \qcwp[P](g)$ and likewise for \qcwlp (monotonicity).
\item \label{thm:cwp-linear} 
$\qcwp[P](\alpha \cdot f + \beta \cdot g) = 
\alpha \cdot \qcwp[P](f) + \beta \cdot \qcwp[P](g)$. 
%
\item \label{thm:cwp-other} 
$\qcwp[P](\CteFun{0}) = \CteFun{0}$ and $\qcwlp[P](\CteFun{1}) = \CteFun{1}$.
\end{enumerate}
\end{lemma}
\endgroup

  \paragraph{\textnormal{\textit{Proof of \ref{thm:cwp-monot}}}} We do the proof for transformer
  $\qcwp$; the proof for $\qcwp$ is analogous. On view of
  Theorem~\ref{thm:cwp-decuop}, the monotonicity of $\qcwp$ reduces to the
  monotonicity of \wp which follows immediately from its continuity (see 
  Lemma~\ref{thm:wp-ext-is-cont}).

  \paragraph{\textnormal{\textit{Proof of \ref{thm:cwp-linear}}}} Once again, on view of
  Theorem~\ref{thm:cwp-decuop}, the linearity of $\qcwp$ follows from the linearity
  of \wp, which we prove in Lemma~\ref{thm:wp-linear}.\footnote{We cannot adopt
    the results from the original work~\cite{McIver:2004} because their analyses
    is restricted to bounded expectations.}

\paragraph{\textnormal{\textit{Proof of \ref{thm:cwp-other}}}} Let us begin by proving that
$\qcwp[P](\CteFun{0}) = \CteFun{0}$. On account of
Theorem~\ref{thm:cwp-decuop} this assertion reduces to $\wp[P](\CteFun{0}) \allowbreak=\allowbreak
\CteFun{0}$, which has already been proved for \pGCL
programs (see \eg~\cite{McIver:2004}). Therefore we only have to deal with the case of
$\Observe$ statements and the claim holds since $\wp[\Observe \, G](\CteFun{0})
= \ToExp{G} \cdot \CteFun{0} = \CteFun{0}$.  Finally formula
$\qcwlp[P](\CteFun{1}) \allowbreak=\allowbreak \CteFun{1}$ follows immediately from
Theorem~\ref{thm:cwp-decuop}. \qed

%% file: opapp.tex
\subsection{Proof of Lemma \ref{lem:ExpRew_is_wp} (i)}
\label{proofof:ExpRew_is_wp}

For proving Lemma \ref{lem:ExpRew_is_wp} (i) we rely on the fact that allowing a bounded while--loop to be executed for an increasing number of times approximates the behavior of an unbounded while--loop.
We first define bounded while--loops formally:

\begin{definition}[Bounded \texttt{while}--Loops]
Let $P \in \pGCL$. Then we define:
\begin{align*}
\WhileDok{0}{G}{P} ~&\eqdef~  \Abort\\
\WhileDok{k+1}{G}{P} ~&\eqdef~  \Cond{G}{P^k}{\Skip}\\
P^k ~&\eqdef~ P;\: \WhileDok{k}{G}{P}
\end{align*}
\end{definition}

We can now establish that by taking the supremum on the bound $k$ we obtain the full behavior of the unbounded while--loop:

\begin{lemma}
\label{lem:whilekwhile}
Let $G$ be a predicate, $P \in \pGCL$, and $f \in \Ex$.
Then it holds that
\begin{align*}
\sup_{k \in \Naturals} \wp[\WhileDok{k}{G}{P}](f) = \wp[\WhileDo{G}{P}](f)~.
\end{align*}
\end{lemma}

\begin{proof}
For any predicate $G$, any program $P \in \pGCL$, and any expectation $f \in \Ex$ let
\begin{align*}
F(X) ~=~  \ToExp{G}\cdot\wp[P](X) + \ToExp{\neg G}\cdot f~.
\end{align*}
We first show by induction on $k \in \Naturals$ that
\begin{align*}
\wp[\WhileDok{k}{G}{P}](f) ~=~ F^k(\ExpZero)~.
\end{align*}
For the induction base we have $k = 0$.
In that case we have
\begin{align*}
&\wp[\WhileDok{0}{G}{P}](f)\\
=~ &\wp[\Abort](f)\\
=~ &\ExpZero\\
=~ &F^0(\ExpZero)~.
\end{align*}
As the induction hypothesis assume now that 
\begin{align*}
\wp[\WhileDok{k}{G}{P}](f) ~=~ F^k(\ExpZero)(f)
\end{align*}
holds for some arbitrary but fixed $k$.
Then for the induction step we have
\begin{align*}
&\wp[\WhileDok{k+1}{G}{P}](f)\\
=~ &\wp[P;~\Cond{G}{\WhileDok{k}{G}{P}}{\Skip}](f)\\
=~ &(\ToExp{G} \cdot \wp[P]\circ\wp[\WhileDok{k}{G}{P}] \\
&\qquad+ \ToExp{\neg G} \cdot \wp[\Skip])(f)\\
=~ &\ToExp{G} \cdot \wp[P](\wp[\WhileDok{k}{G}{P}](f)) \\
& \qquad + \ToExp{\neg G} \cdot \wp[\Skip](f)\\
=~ &\ToExp{G} \cdot \wp[P](F^k(\ExpZero)) + \ToExp{\neg G} \cdot f\tag{I.H.}\\
=~ &F^{k+1}(\ExpZero)(f)~.
\end{align*}
We have by now established that
\begin{align*}
\wp[\WhileDok{k}{G}{P}](f) ~=~ F^k(\ExpZero)
\end{align*}
holds for every $k \in \Naturals$. Ergo, we can also establish that
\belowdisplayskip=-1\baselineskip
\begin{align*}
&\sup_{k \in \Naturals}\wp[\WhileDok{k}{G}{P}](f)\\
=~ &\sup_{k \in \Naturals} F^k(\ExpZero)\\
=~ &\lfp X.~ F(X)\\
=~&\wp[\WhileDo{G}{P}](f)~.
\end{align*}
\end{proof}

With Lemma \ref{lem:whilekwhile} in mind, we can now restate and prove Lemma \ref{lem:ExpRew_is_wp} (i):

\begingroup
\def\thelemma{\ref{lem:ExpRew_is_wp} (i)}
\begin{lemma}
For $P \in \cfpGCL$, $f \in \Ex, g \in \BEx$, and $\sigma \in \State$:
\begin{align*}
\ExpRew{\Rdtmc{\sigma}{f}{P}}{\lozenge \sink} &= \wp[P](f)(\sigma)
\end{align*}
\end{lemma}
\addtocounter{lemma}{-1}
\endgroup


\begin{proof}

The proof goes by induction over all \cfpGCL programs. 
For the induction base we have:

\textbf{The Effectless Program \Skip.} The RMC for this program is of the following form:\footnote{If transitions have probability $1$, we omit this in our figures. Moreover, all states---with the exception of $\sink$---are left out if they are not reachable from the initial state.}
\begin{center}
    \input{pics/skip0}
\end{center}
In the above RMC we have $\Pi := \paths{\langle \Skip,\, \sigma \rangle}{\sink} \allowbreak=\allowbreak \{\hat\pi_1\}$ with $\hat\pi_1 = \langle \Skip,\, \sigma \rangle \rightarrow \langle \term,\, \sigma\rangle \rightarrow \sink$.
Then we have for the expected reward:
\begin{align*}
&\ExpRew{\Rdtmc{\sigma}{f}{\Skip}}{\lozenge \sinklabel}\\
 ~=~&\sum_{\hat\pi \in \Pi}\Pr(\hat\pi) \cdot r(\hat\pi)\\
 ~=~&\Pr(\hat\pi_1) \cdot r(\hat\pi_1)\\
 ~=~&1 \cdot f(\sigma)\\
 ~=~ &f(\sigma)\\
 ~=~ &\wp[\Skip](f)(\sigma)
\end{align*}

\textbf{The Faulty Program \Abort.} The RMC for this program is of the following form:
\begin{center}
    \input{pics/abort0}
\end{center}
In this RMC we have $\Pi := \paths{\langle \Abort,\, \sigma \rangle}{\sink} = \emptyset$.
Then we have for the expected reward:
\begin{align*}
&\ExpRew{\Rdtmc{\sigma}{f}{\Abort}}{\lozenge \sinklabel}\\
 ~=~&\sum_{\hat\pi \in \Pi}\Pr(\hat\pi) \cdot r(\hat\pi)\\
 ~=~&\sum_{\hat\pi \in \emptyset}\Pr(\hat\pi) \cdot r(\hat\pi)\\
 ~=~&0\\
 ~=~&\ExpZero(\sigma)\\
 ~=~ &\wp[\Abort](f)(\sigma)
\end{align*}

\textbf{The Assignment $\Ass x E$.} The RMC for this program is of the following form:
\begin{center}
    \input{pics/ass0}
\end{center}
In this RMC we have $\Pi := \paths{\langle \Ass x E,\, \sigma \rangle}{\sink} = \{\hat\pi_1\}$ with $\hat\pi_1 = \langle \Ass x E,\, \sigma \rangle \rightarrow \langle \term,\, \sigma[E / x]\rangle \rightarrow \sink$.
Then we have for the expected reward:
\begin{align*}
&\ExpRew{\Rdtmc{\sigma}{f}{\Ass x E}}{\lozenge \sinklabel}\\
 ~=~&\sum_{\hat\pi \in \Pi}\Pr(\hat\pi) \cdot r(\hat\pi)\\
 ~=~&\Pr(\hat\pi_1) \cdot r(\hat\pi_1)\\
 ~=~&1 \cdot f(\sigma[E / x])\\
 ~=~&f(\sigma[E / x])\\
 ~=~&f[E / x](\sigma)\\
 ~=~&\wp[\Ass x E](f)(\sigma)
\end{align*}

\textbf{The Observation $\Observe~ G$.} For this program there are two cases: In Case 1 we have $\sigma \models G$, so we have $\ToExp{G}(\sigma) = 1$. 
The RMC in this case is of the following form:
\begin{center}
    \input{pics/observe1} 
\end{center}
In this RMC we have $\Pi := \paths{\langle \Observe~G,\, \sigma \rangle}{\sink} \allowbreak=\allowbreak \{\hat\pi_1\}$ with $\hat\pi_1 = \langle \Observe~G,\, \sigma \rangle \rightarrow \langle \term,\, \sigma\rangle \rightarrow \sink$.
Then we have for the expected reward:
\begin{align*}
&\ExpRew{\Rdtmc{\sigma}{f}{\Observe~G}}{\lozenge \sinklabel}\\
 ~=~&\sum_{\hat\pi \in \Pi}\Pr(\hat\pi) \cdot r(\hat\pi)\\
 ~=~&\Pr(\hat\pi_1) \cdot r(\hat\pi_1)\\
 ~=~&1 \cdot f(\sigma)\\
 ~=~&\ToExp{G}(\sigma) \cdot f(\sigma)\\
 ~=~&(\ToExp{G}\cdot f)(\sigma)\\
 ~=~&\wp[\Observe~G](f)(\sigma)
\end{align*}
In Case 2 we have $\sigma \not\models G$, so we have $\ToExp{G}(\sigma) = 0$. 
The RMC in this case is of the following form:
\begin{center}
    \input{pics/observe2}
\end{center}
In this RMC we have $\Pi := \paths{\langle \Observe~G,\, \sigma \rangle}{\sink} \allowbreak=\allowbreak \{\hat\pi_1\}$ with $\hat\pi_1 = \langle \Observe~G,\, \sigma \rangle \rightarrow \undesired \rightarrow \sink$.
Then for the expected reward we also have:
\begin{align*}
&\ExpRew{\Rdtmc{\sigma}{f}{\Observe~G}}{\lozenge \sinklabel}\\
 ~=~&\sum_{\hat\pi \in \Pi}\Pr(\hat\pi) \cdot r(\hat\pi)\\
 ~=~&\Pr(\hat\pi_1) \cdot r(\hat\pi_1)\\
 ~=~&1 \cdot 0\\
 ~=~& 0\\
 ~=~& 0 \cdot f(\sigma)\\
 ~=~&\ToExp{G}(\sigma) \cdot f(\sigma)\\
 ~=~&(\ToExp{G}\cdot f)(\sigma)\\
 ~=~&\wp[\Observe~G](f)(\sigma)
\end{align*}

\textbf{The Concatenation $P;\, Q$.} For this program the RMC is of the following form:
\begin{center}
    \input{pics/concat1}
\end{center}
In this RMC every path in $\paths{\langle P;\, Q,\, \sigma \rangle}{\sink}$ starts with $\langle P;\, Q,\, \sigma \rangle$, eventually reaches $\langle \term;\, Q,\, \sigma'\rangle$, and then immediately after that reaches $\langle Q,\, \sigma'\rangle$ which is the initial state of $\Rdtmc{\sigma'}{f}{Q}$ for which the expected reward is given by $\ExpRew{\Rdtmc{\sigma'}{f}{Q}}{\lozenge \sinklabel}$.
By this insight we can transform the above RMC into the RMC with equal expected reward below:
\begin{center}
    \input{pics/concat2}
\end{center}But the above RMC is exactly $\Rdtmc{\sigma}{\lambda \tau.\ExpRew{\Rdtmc{\tau}{f}{Q}}{\lozenge \sinklabel}}{P}$ for which the expected reward is also known by the induction hypothesis. 
So we have
\begin{align*}
&\ExpRew{\Rdtmc{\sigma}{f}{P;\, Q}}{\lozenge \sinklabel}\\
 ~=~&\ExpRew{\Rdtmc{\sigma}{\lambda \tau. \ExpRew{\Rdtmc{\tau}{f}{Q}}{\lozenge \sinklabel}}{P}}{\lozenge \sinklabel}\\
 ~=~&\ExpRew{\Rdtmc{\sigma}{\wp[Q](f)}{P}}{\lozenge \sinklabel}\tag{I.H. on $Q$}\\
 ~=~&\wp[P](\wp[Q](f))(\sigma)\tag{I.H. on $P$}\\
 ~=~&\wp[P;\, Q](f)
\end{align*}

\textbf{The Conditional Choice $\Cond G P Q$.} For this program there are two cases: In Case 1 we have $\sigma \models G$, so we have $\ToExp{G}(\sigma) = 1$ and $\ToExp{\neg G}(\sigma) = 0$. 
The RMC in this case is of the following form:
\begin{center}
 \input{pics/cond11}
\end{center}
In this RMC every path in $\paths{\langle \Cond G P Q,\, \sigma \rangle}{\sink}$ starts with $\langle \Cond G P Q,\, \sigma \rangle \rightarrow \langle P,\, \sigma \rangle \rightarrow \cdots$.
As the state $\langle \Cond G P Q,\, \sigma \rangle$ collects zero reward, the expected reward of the above RMC is equal to the expected reward of the following RMC:
\begin{center}
    \input{pics/cond12}
\end{center}
But the above RMC is exactly $\Rdtmc{\sigma}{f}{P}$ for which the expected reward is known by the induction hypothesis. 
So we have
\begin{align*}
&\ExpRew{\Rdtmc{\sigma}{f}{\Cond G P Q}}{\lozenge \sinklabel}\\
 ~=~&\ExpRew{\Rdtmc{\sigma}{f}{P}}{\lozenge \sinklabel}\\
 ~=~&\wp[P](f)(\sigma)\tag{I.H.}\\
 ~=~&1 \cdot \wp[P](f)(\sigma) + 0 \cdot \wp[Q](f)(\sigma) \\
 ~=~&\ToExp{G}(\sigma) \cdot \wp[P](f)(\sigma) + \ToExp{\neg G}(\sigma) \cdot \wp[Q](f)(\sigma) \\
 ~=~&\wp[\Cond G P Q](f)(\sigma)~.
\end{align*}
In Case 2 we have $\sigma \not\models G$, so we have $\ToExp{G}(\sigma) = 0$ and $\ToExp{\neg G}(\sigma) = 1$. 
The RMC in this case is of the following form:
\begin{center}
    \input{pics/cond21}
\end{center}
In this RMC every path in $\paths{\langle \Cond G P Q,\, \sigma \rangle}{\sink}$ starts with $\langle \Cond G P Q,\, \sigma \rangle \rightarrow \langle Q,\, \sigma \rangle \rightarrow \cdots$.
As the state $\langle \Cond G P Q,\, \sigma \rangle$ collects zero reward, the expected reward of the above RMC is equal to the expected reward of the following RMC:
\begin{center}
    \input{pics/cond22}
\end{center}
But the above RMC is exactly $\Rdtmc{\sigma}{f}{Q}$ for which the expected reward is known by the induction hypothesis. 
So we also have
\begin{align*}
&\ExpRew{\Rdtmc{\sigma}{f}{\Cond G P Q}}{\lozenge \sinklabel}\\
 ~=~&\ExpRew{\Rdtmc{\sigma}{f}{Q}}{\lozenge \sinklabel}\\
 ~=~&\wp[Q](f)(\sigma)\tag{I.H.}\\
 ~=~&0 \cdot \wp[P](f)(\sigma) + 1 \cdot \wp[Q](f)(\sigma) \\
 ~=~&\ToExp{G}(\sigma) \cdot \wp[P](f)(\sigma) + \ToExp{\neg G}(\sigma) \cdot \wp[Q](f)(\sigma) \\
 ~=~&\wp[\Cond G P Q](f)(\sigma)~.
\end{align*}

\textbf{The Probabilistic Choice $\PChoice P p Q$}. For this program the RMC is of the following form:
\begin{center}
    \input{pics/prob1}
\end{center}In this RMC every path in $\paths{\langle \PChoice P p Q,\, \sigma \rangle}{\sink}$ starts with $\langle \PChoice P p Q,\, \sigma \rangle$ and immediately after that reaches $\langle P,\, \sigma\rangle$ with probability $p$ or $\langle Q,\, \sigma\rangle$ with probability $1-p$.
$\langle P,\, \sigma\rangle$ is the initial state of $\Rdtmc{\sigma}{f}{P}$ and $\langle Q,\, \sigma\rangle$ is the initial state of $\Rdtmc{\sigma}{f}{Q}$.
By this insight we can transform the above RMC into the RMC with equal expected reward below:
\begin{center}
    \input{pics/prob2}
\end{center}The expected reward of the above RMC is given by $p \cdot \ExpRew{\Rdtmc{\sigma}{f}{P}}{\lozenge \sinklabel} + (1 - p) \cdot \ExpRew{\Rdtmc{\sigma}{f}{Q}}{\lozenge \sinklabel}$, so in total we have for the expected reward:
\begin{align*}
&\ExpRew{\Rdtmc{\sigma}{f}{\PChoice P p Q}}{\lozenge \sinklabel}\\
 ~=~&p \cdot \ExpRew{\Rdtmc{\sigma}{f}{P}}{\lozenge \sinklabel} \\
 &\qquad{}+ (1 - p) \cdot \ExpRew{\Rdtmc{\sigma}{f}{Q}}{\lozenge \sinklabel}\\
 ~=~&p \cdot \wp[P](f)(\sigma) + (1 - p) \cdot \wp[Q](f)(\sigma)\tag{I.H.}\\
 ~=~&\wp[\PChoice P p Q](f)~.
\end{align*}

\textbf{The Loop $\WhileDo G Q$.} By Lemma \ref{lem:whilekwhile} we have
\begin{align*}
&\wp[\WhileDo G P](f) ~=~ \sup_{k \in \Naturals} \wp[\WhileDok k G P](f)
\end{align*}
and as $\WhileDok k G P$ is a purely syntactical construct (made up from \Abort, \Skip, conditional choice, and $P$) we can (using what we have already established on \Abort, \Skip, conditional choice, and using the induction hypothesis on $P$) also establish that
\begin{align*}
&\wp[\WhileDo G P](f)\\
=~ &\sup_{k \in \Naturals} \ExpRew{\Rdtmc{\sigma}{f}{\WhileDok k G P}}{\lozenge \sinklabel}~.
\end{align*}
It is now left to show that
\begin{align}
&\sup_{k \in \Naturals} \ExpRew{\Rdtmc{\sigma}{f}{\WhileDok k G P}}{\lozenge \sinklabel}\label{eq:supkwhiledok}\\
=~ &\ExpRew{\Rdtmc{\sigma}{f}{\WhileDo G P}}{\lozenge \sinklabel}\label{eq:whiledo}~.
\end{align}
While the above is intuitively evident, it is a tedious and technically involved task to prove it. 
Herefore we just provide an intuition thereof:
For showing $\eqref{eq:supkwhiledok} \leq \eqref{eq:whiledo}$, we know that every path in the RMDP $\Rdtmc{\sigma}{f}{\WhileDok k G P}$ either terminates properly or is prematurely aborted (yielding 0 reward) due to the fact that the bound of less than $k$ loop iterations was reached. The RMDP $\Rdtmc{\sigma}{f}{\WhileDo G P}$ for the unbounded while--loop does not prematurely abort executions, so left--hand--side is upper bounded by the right--hand--side of the equation. 
For showing $\eqref{eq:supkwhiledok} \geq \eqref{eq:whiledo}$, we know that a path that collects positive reward is necessarily finite. Therefore there exists some $k \in \Naturals$ such that $\Rdtmc{\sigma}{f}{\WhileDok k G P}$ includes this path. Taking the supremum over $k$ we eventually include every path  in $\Rdtmc{\sigma}{f}{\WhileDo G P}$ that collects positive reward.
\end{proof}

\subsection{Proof of Lemma \ref{lem:ExpRew_is_wp} (ii)}

\label{proofof:LExpRew_is_wlp}

\begingroup
\def\thelemma{\ref{lem:ExpRew_is_wp} (ii)}
\begin{lemma}
For $P \in \cfpGCL$, $f \in \Ex, g \in \BEx$, and $\sigma \in \State$:
\begin{align*}
\LExpRew{\Rdtmc{\sigma}{g}{P}}{\lozenge \sink} &= \wlp[P](g)(\sigma)
\end{align*}
\end{lemma}
\addtocounter{lemma}{-1}
\endgroup

\begin{proof}
The proof goes by induction over all \cfpGCL programs. 
For the induction base we have:
\textbf{The Effectless Program \Skip.} The RMC for this program is of the following form:
\begin{center}
    \input{pics/skip0}
\end{center}
In this RMC we have $\Pi := \paths{\langle \Skip,\, \sigma \rangle}{\sink} = \{\hat\pi_1\}$ with $\hat\pi_1 = \langle \Skip,\, \sigma \rangle \rightarrow \langle \term,\, \sigma\rangle \rightarrow \sink$.
Then we have for the liberal expected reward:
\begin{align*}
&\LExpRew{\Rdtmc{\sigma}{g}{\Skip}}{\lozenge \sinklabel}\\
 ~=~&\sum_{\hat\pi \in \Pi}\Pr(\hat\pi) \cdot r(\hat\pi) + \Pr(\neg\lozenge \sink)\\
 ~=~&\Pr(\hat\pi) \cdot r(\hat\pi) + 0\\
 ~=~&1 \cdot g(\sigma)\\
 ~=~ &g(\sigma)\\
 ~=~ &\wlp[\Skip](g)(\sigma)
\end{align*}

\textbf{The Faulty Program \Abort.} The RMC for this program is of the following form:
\begin{center}
    \input{pics/abort0}
\end{center}
In this RMC we have $\Pi := \paths{\langle \Abort,\, \sigma \rangle}{\sink} = \emptyset$.
Then we have for the liberal expected reward:
\begin{align*}
&\ExpRew{\Rdtmc{\sigma}{g}{\Abort}}{\lozenge \sinklabel}\\
 ~=~&\sum_{\hat\pi \in \Pi}\Pr(\hat\pi) \cdot r(\hat\pi) + \Pr(\neg\lozenge \sink)\\
 ~=~&\sum_{\hat\pi \in \emptyset}\Pr(\hat\pi) \cdot r(\hat\pi) + 1\\
 ~=~&0 + 1\\
  ~=~&1\\
 ~=~&\ExpOne(\sigma)\\
 ~=~ &\wlp[\Abort](g)(\sigma)
\end{align*}

\textbf{The Assignment $\Ass x E$.} The RMC for this program is of the following form:
\begin{center}
    \input{pics/ass0}
\end{center}
In this RMC we have $\Pi := \paths{\langle \Ass x E,\, \sigma \rangle}{\sink} = \{\hat\pi_1\}$ with $\hat\pi_1 = \langle \Ass x E,\, \sigma \rangle \rightarrow \langle \term,\, \sigma[E / x]\rangle \rightarrow \sink$.
Then we have for the liberal expected reward:
\begin{align*}
&\LExpRew{\Rdtmc{\sigma}{g}{\Ass x E}}{\lozenge \sinklabel}\\
 ~=~&\sum_{\hat\pi \in \Pi}\Pr(\hat\pi) \cdot r(\hat\pi) + \Pr(\neg\lozenge \sink)\\
 ~=~&\Pr(\hat\pi_1) \cdot r(\hat\pi_1) + 0\\
 ~=~&1 \cdot g(\sigma[E / x])\\
 ~=~&g(\sigma[E / x])\\
 ~=~&g[E / x](\sigma)\\
 ~=~&\wlp[\Ass x E](g)(\sigma)
\end{align*}

\textbf{The Observation $\Observe~ G$.} For this program there are two cases: In Case 1 we have $\sigma \models G$, so we have $\ToExp{G}(\sigma) = 1$. 
The RMC in this case is of the following form:
\begin{center}
    \input{pics/observe1} 
\end{center}
In this RMC we have $\Pi := \paths{\langle \Observe~G,\, \sigma \rangle}{\sink} \allowbreak=\allowbreak \{\hat\pi_1\}$ with $\hat\pi_1 = \langle \Observe~G,\, \sigma \rangle \rightarrow \langle \term,\, \sigma\rangle \rightarrow \sink$.
Then we have for the liberal expected reward:
\begin{align*}
&\LExpRew{\Rdtmc{\sigma}{g}{\Observe~G}}{\lozenge \sinklabel}\\
 ~=~&\sum_{\hat\pi \in \Pi}\Pr(\hat\pi) \cdot r(\hat\pi) + \Pr(\neg\lozenge \sink)\\
 ~=~&\Pr(\hat\pi_1) \cdot r(\hat\pi_1) + 0\\
 ~=~&1 \cdot g(\sigma)\\
 ~=~&\ToExp{G}(\sigma) \cdot g(\sigma)\\
 ~=~&(\ToExp{G}\cdot g)(\sigma)\\
 ~=~&\wlp[\Observe~G](g)(\sigma)
\end{align*}
In Case 2 we have $\sigma \not\models G$, so we have $\ToExp{G}(\sigma) = 0$. 
The RMC in this case is of the following form:
\begin{center}
    \input{pics/observe2}
\end{center}
In this RMC we have $\Pi := \paths{\langle \Observe~G,\, \sigma \rangle}{\sink} \allowbreak=\allowbreak \{\hat\pi_1\}$ with $\hat\pi_1 = \langle \Observe~G,\, \sigma \rangle \rightarrow \undesired \rightarrow \sink$.
Then we have for the liberal expected reward:
\begin{align*}
&\LExpRew{\Rdtmc{\sigma}{g}{\Observe~G}}{\lozenge \sinklabel}\\
 ~=~&\sum_{\hat\pi \in \Pi}\Pr(\hat\pi) \cdot r(\hat\pi)  + \Pr(\neg\lozenge \sink)\\
 ~=~&\Pr(\hat\pi_1) \cdot r(\hat\pi_1) + 0\\
 ~=~&1 \cdot 0\\
 ~=~& 0\\
 ~=~& 0 \cdot g(\sigma)\\
 ~=~&\ToExp{G}(\sigma) \cdot g(\sigma)\\
 ~=~&(\ToExp{G}\cdot g)(\sigma)\\
 ~=~&\wlp[\Observe~G](g)(\sigma)
\end{align*}

\textbf{The Concatenation $P;\, Q$.} 
For this program the RMC is of the following form:
\begin{center}
 \scalebox{.84}{\input{pics/libconcat1}}
\end{center}
In this RMC every path in $\paths{\langle P;\, Q,\, \sigma \rangle}{\sink}$ starts with $\langle P;\, Q,\, \sigma \rangle$, eventually reaches $\langle \term;\, Q,\, \sigma\rangle$, and then immediately after that reaches $\langle Q,\, \sigma\rangle$ which is the initial state of $\Rdtmc{\sigma}{g}{Q}$. 
Every diverging path either diverges because the program $P$ diverges or because the program $Q$ diverges.
If we attempt to make the RMC smaller (while preserving the liberal expected reward) by cutting it off at states of the form $\langle \term;\, Q, \tau\rangle$, we have to assign to them the liberal expected reward $\LExpRew{\Rdtmc{\tau}{g}{Q}}{\lozenge \sinklabel}$ in order to not loose the non--termination probability caused by $Q$.
By this insight we can now transform the above RMC into the RMC with equal liberal expected reward below:
\begin{center}
 \input{pics/libconcat2}
\end{center}
But the above RMC is exactly $\Rdtmc{\sigma}{\LExpRew{\Rdtmc{\sigma}{g}{Q}}{\lozenge \sinklabel}}{P}$ for which the liberal expected reward is known by the induction hypothesis. 
So we have for the liberal expected reward:
\begin{align*}
&\LExpRew{\Rdtmc{\sigma}{g}{P;\, Q}}{\lozenge \sinklabel}\\
 ~=~&\LExpRew{\Rdtmc{\sigma}{\LExpRew{\Rdtmc{\sigma}{g}{Q}}{\lozenge \sinklabel}}{P}}{\lozenge \sinklabel}\\
 ~=~&\LExpRew{\Rdtmc{\sigma}{\wlp[Q](g)}{P}}{\lozenge \sinklabel}\tag{I.H. on $Q$}\\
 ~=~&\wlp[P](\wlp[Q](g))(\sigma)\tag{I.H. on $P$}\\
 ~=~&\wlp[P;\, Q](g)~.
\end{align*}

\textbf{The Conditional Choice $\Cond G P Q$.} For this program there are two cases: In Case 1 we have $\sigma \models G$, so we have $\ToExp{G}(\sigma) = 1$ and $\ToExp{\neg G}(\sigma) = 0$. 
The RMC in this case is of the following form:
\begin{center}
 \input{pics/cond11}
\end{center}
As the state $\langle \Cond G P Q,\, \sigma \rangle$ collects zero reward, the expected reward of the above RMC is equal to the expected reward of the following RMC:
\begin{center}
 \input{pics/cond12}
\end{center}
But the above RMC is exactly $\Rdtmc{\sigma}{g}{P}$ for which the expected reward is known by Lemma . 
A similar argument can be applied to the probability of not eventually reaching $\sink$.
So we have for the liberal expected reward:
\begin{align*}
&\LExpRew{\Rdtmc{\sigma}{g}{\Cond G P Q}}{\lozenge \sinklabel}\\
 ~=~&\ExpRew{\Rdtmc{\sigma}{g}{\Cond G P Q}}{\lozenge \sinklabel}\\
        &{} + \Pr^{\Rdtmc{\sigma}{g}{\Cond G P Q}}(\neg\lozenge \sink)\\
 ~=~&\ExpRew{\Rdtmc{\sigma}{g}{P}}{\lozenge \sinklabel}  + \Pr^{\Rdtmc{\sigma}{g}{P}}(\neg\lozenge \sink)\\
 ~=~&\wlp[P](g)(\sigma)\tag{I.H.}\\
 ~=~&1 \cdot \wlp[P](g)(\sigma) + 0 \cdot \wlp[Q](g)(\sigma) \\
 ~=~&\ToExp{G}(\sigma) \cdot \wlp[P](g)(\sigma) + \ToExp{\neg G}(\sigma) \cdot \wlp[Q](g)(\sigma) \\
 ~=~&\wlp[\Cond G P Q](g)(\sigma)~.
\end{align*}
In Case 2 we have $\sigma \not\models G$, so we have $\ToExp{G}(\sigma) = 0$ and $\ToExp{\neg G}(\sigma) = 1$. 
The RMC in this case is of the following form:
\begin{center}
 \input{pics/cond21}
\end{center}
In this RMC every path in $\paths{\langle \Cond G P Q,\, \sigma \rangle}{\sink}$ starts with $\langle \Cond G P Q,\, \sigma \rangle \rightarrow \langle Q,\, \sigma \rangle \rightarrow \cdots$.
As the state $\langle \Cond G P Q,\, \sigma \rangle$ collects zero reward, the expected reward of the above RMC is equal to the expected reward of the following RMC:
\begin{center}
 \input{pics/cond22}
\end{center}
But the above RMC is exactly $\Rdtmc{\sigma}{g}{Q}$ for which the expected reward is known by the induction hypothesis. 
A similar argument can be applied to the probability of not eventually reaching $\sink$.
So we also have for the liberal expected reward:
\begin{align*}
&\LExpRew{\Rdtmc{\sigma}{g}{\Cond G P Q}}{\lozenge \sinklabel}\\
 ~=~&\ExpRew{\Rdtmc{\sigma}{g}{\Cond G P Q}}{\lozenge \sinklabel}\\
        &{} + \Pr^{\Rdtmc{\sigma}{g}{\Cond G P Q}}(\neg\lozenge \sink)\\
 ~=~&\ExpRew{\Rdtmc{\sigma}{g}{Q}}{\lozenge \sinklabel}  + \Pr^{\Rdtmc{\sigma}{g}{Q}}(\neg\lozenge \sink)\\
 ~=~&\wlp[Q](g)(\sigma)\tag{I.H.}\\
 ~=~&0 \cdot \wlp[P](g)(\sigma) + 1 \cdot \wlp[Q](g)(\sigma) \\
 ~=~&\ToExp{G}(\sigma) \cdot \wlp[P](g)(\sigma) + \ToExp{\neg G}(\sigma) \cdot \wlp[Q](g)(\sigma) \\
 ~=~&\wlp[\Cond G P Q](g)(\sigma)~.
\end{align*}

\textbf{The Probabilistic Choice $\PChoice P p Q$.} For this program the RMC is of the following form:
\begin{center}
    \input{pics/prob1}
\end{center}In this RMC every path in $\paths{\langle \PChoice P p Q,\, \sigma \rangle}{\sink}$ starts with $\langle \PChoice P p Q,\, \sigma \rangle$ and immediately after that reaches $\langle P,\, \sigma\rangle$ with probability $p$ or $\langle Q,\, \sigma\rangle$ with probability $1-p$.
$\langle P,\, \sigma\rangle$ is the initial state of $\Rdtmc{\sigma}{f}{P}$ and $\langle Q,\, \sigma\rangle$ is the initial state of $\Rdtmc{\sigma}{f}{Q}$.
The same holds for all paths that do not eventually reach \sink.
By this insight we can transform the above RMC into the RMC with equal liberal expected reward below:
\begin{center}
    \input{pics/prob2}
\end{center}The liberal expected reward of the above RMC is given by $p \cdot \LExpRew{\Rdtmc{\sigma}{f}{P}}{\lozenge \sinklabel} + (1 - p) \cdot \LExpRew{\Rdtmc{\sigma}{f}{Q}}{\lozenge \sinklabel}$, so in total we have for the liberal expected reward:
\begin{align*}
&\LExpRew{\Rdtmc{\sigma}{f}{\PChoice P p Q}}{\lozenge \sinklabel}\\
 ~=~&p \cdot \LExpRew{\Rdtmc{\sigma}{f}{P}}{\lozenge \sinklabel} \\
 &\qquad {}+ (1 - p) \cdot \LExpRew{\Rdtmc{\sigma}{f}{Q}}{\lozenge \sinklabel}\\
 ~=~&p \cdot \wlp[P](f)(\sigma) + (1 - p) \cdot \wlp[Q](f)(\sigma)\tag{I.H.}\\
 ~=~&\wlp[\PChoice P p Q](f)~.
\end{align*}

\textbf{The Loop $\WhileDo G Q$.}

The argument is dual to the case for the (non--liberal) expected reward.
\end{proof}

\subsection{Proof of Lemma \ref{lem:Pr_not_bad_is_wlp_1}}
\label{proofof:lem:Pr_not_bad_is_wlp_1}

\begingroup
\def\thelemma{\ref{lem:Pr_not_bad_is_wlp_1}}
\begin{lemma}
For $P \in \cfpGCL$, $g \in \BEx$, and $\sigma \in \State$:
\begin{align*}
\Prr{\Rdtmc{\sigma}{g}{P}}({\neg \lozenge \bad}) ~=~ \wlp[P](\ExpOne)(\sigma)~.
\end{align*}
\end{lemma}
\addtocounter{lemma}{-1}
\endgroup

\begin{proof}
First, observe that paths on reaching \exit or \bad imme{\-}di{\-}ate{\-}ly move to the state \sink. 
Moreover, all paths that never visit $\bad$ either (a) visit a terminal \exit--state (which are the only states that can possibly collect positive reward) or (b) diverge and never reach \sink and therefore neither reach \exit nor \bad. 
Furthermore the set of ``(a)--paths" and the set of ``(b)--paths" are disjoint.
Thus:
\begin{align*}
&\Prr{\Rdtmc{\sigma}{f}{P}}({\neg \lozenge \bad})\\
~=~ &\Prr{\Rdtmc{\sigma}{f}{P}}({\lozenge \exit}) + \Prr{\Rdtmc{\sigma}{f}{P}}({\neg \lozenge \sinklabel})\\
\intertext{and by assigning reward one to every \exit--state, and zero to all other states, we can turn the probability measure into an expected reward, yielding}
~=~ &\ExpRew{\Rdtmc{\sigma}{\ExpOne}{P}}{\lozenge \exit} + \Prr{\Rdtmc{\sigma}{g}{P}}({\neg\lozenge \sinklabel})\\
\intertext{As every path that reaches sink over a \bad--state cumulates zero reward, we finally get:}
~=~ &\ExpRew{\Rdtmc{\sigma}{\ExpOne}{P}}{\lozenge \sinklabel} + \Prr{\Rdtmc{\sigma}{g}{P}}({\neg\lozenge \sinklabel})\\
~=~ &\LExpRew{\Rdtmc{\sigma}{\ExpOne}{P}}{\lozenge \sinklabel}\\
~=~ &\wlp[P](\ExpOne)\tag{Lemma \ref{lem:ExpRew_is_wp}}\\[-1.5\baselineskip]
\end{align*}
\end{proof}

\subsection{Proof of Theorem \ref{thm:correspondence:cwp}}
\label{proofof:thm:correspondence:cwp}

\begingroup
\def\thelemma{\ref{thm:correspondence:cwp}}
\begin{theorem}[Correspondence theorem]
For $P \in \cfpGCL$, $f \in \Ex$, $g \in \BEx$ and $\sigma \in \State$,
\begin{align*}
\CExpRew{\Rdtmc{\sigma}{f}{P}}{\lozenge \sinklabel}{\neg\lozenge\bad} ~&=~ \qcwp[P](f)(\sigma) \\
\CLExpRew{\Rdtmc{\sigma}{g}{P}}{\lozenge \sinklabel}{\neg\lozenge\bad} ~&=~ \qcwlp[P](g) (\sigma)~.
\end{align*}
\end{theorem}
\addtocounter{lemma}{-1}
\endgroup

\begin{proof}
\belowdisplayskip=-1\baselineskip
We prove only the first equation. 
The proof of the second equation goes along the same arguments.
\begin{align*}
&\CExpRew{\Rdtmc{\sigma}{f}{P}}{\lozenge \sinklabel}{\neg\lozenge\bad}\\
~=~ &\frac{\ExpRew{\Rdtmc{\sigma}{f}{P}}{\lozenge \sinklabel}}{\Prr{\Rdtmc{\sigma}{f}{P}}(\neg\lozenge \bad)}\\
~=~ &\frac{\wp[P](f)}{\wlp[P](\ExpOne)}\tag{Lemmas \ref{lem:ExpRew_is_wp}, \ref{lem:Pr_not_bad_is_wlp_1}}\\
~=~ &\frac{\cwp_1[P](f,\, \ExpOne)}{\cwp_2[P](f,\, \ExpOne)}\tag{Theorem \ref{thm:cwp-decuop}}\\
~=~ &\qcwp[P](f)
\end{align*}
\end{proof}

%% file: pics/skip0.tex
\begin{tikzpicture}[->,>=stealth',shorten >=1pt,node distance=2.5cm,semithick,minimum size=1cm]
\tikzstyle{every state}=[draw=none]
  \draw[white, use as bounding box] (-1.6,-0.85) rectangle (6.5,.25);
   \node [state, initial, initial text=,label={[yshift=0.76cm, gray] 270:$0$}] (init) {$\langle \Skip,\, \sigma\rangle$};  
   \node [state,label={[yshift=0.5cm, gray] 270:$f(\sigma)$}] (exit) [right of=init] {$\langle\term,\,\sigma\rangle$};
   \node [state,label={[yshift=0.5cm, gray] 270:$0$}] (sink) [right of=exit] {$\sink$};

  \path [] 
      (init) edge [] (exit)
      (exit) edge [] (sink)
      (sink) edge [loop right] (sink)
  ;
\end{tikzpicture}

%% file: pics/abort0.tex
\begin{tikzpicture}[->,>=stealth',shorten >=1pt,auto,node distance=3.75cm,semithick,minimum size=1cm]
\tikzstyle{every state}=[draw=none]
  \draw[white, use as bounding box] (-1.5,-0.75) rectangle (5.15,.39);
   \node [state, initial, initial text=,label={[yshift=0.9cm, gray] 270:$0$}] (init) {$\langle \Abort,\, \sigma\rangle$};  
      \node [state,label={[yshift=0.55cm, gray] 270:$0$}] (sink) [right of=init] {$\sink$};

  \path []
       (init) edge [loop right] (init)
       (sink) edge [loop right] (sink);
       ;
\end{tikzpicture}

%% file: pics/ass0.tex
\begin{tikzpicture}[->,>=stealth',shorten >=1pt,node distance=2.5cm,semithick,minimum size=1cm]
\tikzstyle{every state}=[draw=none]
  \draw[white, use as bounding box] (-1.6,-0.85) rectangle (6.4,.25);
   \node [state, initial, initial text=,label={[yshift=0.9cm, gray] 270:$0$}] (init) {$\langle \Ass x E,\, \sigma\rangle$};  
   \node [state,label={[yshift=0.89cm, gray] 270:$f(\sigma[E/x])$}] (exit) [right of=init] {$\langle\term,\,\sigma[E/x]\rangle$};
   \node [state,label={[yshift=0.5cm, gray] 270:$0$}] (sink) [right of=exit] {$\sink$};

  \path [] 
      (init) edge [] (exit)
      (exit) edge [] (sink)
      (sink) edge [loop right] (sink)
  ;
\end{tikzpicture}

%% file: pics/observe1.tex
\begin{tikzpicture}[->,>=stealth',shorten >=1pt,node distance=2.7cm,semithick,minimum size=1cm]
\tikzstyle{every state}=[draw=none]
  \draw[white, use as bounding box] (-1.85,-0.85) rectangle (6.2,.25);
   \node [state, initial, initial text=,label={[yshift=1.2cm, gray] 270:$0$}] (init) {$\langle \Observe~ G,\, \sigma\rangle$};  
   \node [state,label={[yshift=0.5cm, gray] 270:$f(\sigma)$}] (exit) [right=.8 cm of init] {$\langle\term,\,\sigma\rangle$};
   \node [state,label={[yshift=0.5cm, gray] 270:$0$}] (sink) [right=.8 of exit] {$\sink$};

  \path [] 
      (init) edge [] (exit)
      (exit) edge [] (sink)
      (sink) edge [loop right] (sink)
  ;
\end{tikzpicture}

%% file: pics/observe2.tex
\begin{tikzpicture}[->,>=stealth',shorten >=1pt,node distance=2.7cm,semithick,minimum size=1cm]
\tikzstyle{every state}=[draw=none]
  \draw[white, use as bounding box] (-1.9,-0.8) rectangle (6.2,.25);
   \node [state, initial, initial text=,label={[yshift=1.2cm, gray] 270:$0$}] (init) {$\langle \Observe~ G,\, \sigma\rangle$};  
   \node [state,label={[yshift=0.35cm, gray] 270:$0$}] (undesired) [right=.95 cm of init] {$\undesired$};
   \node [state,label={[yshift=0.5cm, gray] 270:$0$}] (sink) [right=.8 of exit] {$\sink$};

  \path [] 
      (init) edge [] (exit)
      (exit) edge [] (sink)
      (sink) edge [loop right] (sink)
  ;
\end{tikzpicture}

%% file: pics/concat1.tex
\begin{tikzpicture}[->,>=stealth',shorten >=1pt,node distance=2.7cm,semithick,minimum size=1cm]
\tikzstyle{every state}=[draw=none]
  \draw[white, use as bounding box] (-1.5,-2.5) rectangle (6.8,.2);
   \node [state, initial, initial text=,label={[yshift=.9cm, gray] 270:$0$}] (init) {$\langle P;\, Q,\, \sigma\rangle$};  
   \node [state,label={[yshift=.9cm, gray] 270:$0$}] (termp) [on grid, right=2.5 cm of init] {$\langle \term;\,Q,\,\sigma'\rangle$};
      \node [state,label={[yshift=.7cm, gray] 270:$0$}] (q) [on grid, right=2 cm of termp] {$\langle Q,\,\sigma'\rangle$};
      \node [state,label={[yshift=.9cm, gray] 270:$0$}] (termp') [on grid, below=1.8 cm of termp] {$\langle \term;\,Q,\,\sigma''\rangle$};
       \node [state,label={[yshift=.7cm, gray] 270:$0$}] (q') [on grid, right=2 cm of termp'] {$\langle Q,\,\sigma''\rangle$};
\node [] (dummy) [below= 0.8 cm of init] {$\vdots$};
   
    \node [] (dummy1) [on grid, right =2cm of q] {$\ldots$};
    \node [] (dummy2) [on grid, right =2cm of q'] {$\ldots$};
  \path [] 
      (init) edge [decorate,decoration={snake, post length=2mm}] (termp)
      (init) edge [decorate,decoration={snake, post length=2mm}] (dummy)
      (init) edge [decorate,decoration={snake, post length=2mm}] (termp')
      (termp) edge [] (q)
      (termp') edge [] (q')
      (q) edge [decorate,decoration={snake, post length=2mm}] (dummy1)
      (q') edge [decorate,decoration={snake, post length=2mm}] (dummy2)
  ;
\end{tikzpicture}

%% file: pics/concat2.tex
\begin{tikzpicture}[->,>=stealth',shorten >=1pt,node distance=2.7cm,semithick,minimum size=1cm]
\tikzstyle{every state}=[draw=none]
  \draw[white, use as bounding box] (-1.2,-2.7) rectangle (6.1,.25);
   \node [state, initial, initial text=,label={[yshift=.8cm, gray] 270:$0$}] (init) {$\langle P,\, \sigma\rangle$};  
   \node [state,label={[yshift=.6cm, gray] 280:$\ExpRew{\Rdtmc{\sigma'}{f}{Q}}{\lozenge \sinklabel}$}] (termp) [on grid, right=2.5 cm of init] {$\langle \term,\,\sigma'\rangle$};
      \node [state,label={[yshift=.6cm, gray] 280:$\ExpRew{\Rdtmc{\sigma''}{f}{Q}}{\lozenge \sinklabel}$}] (termp') [on grid, below=1.8 cm of termp] {$\langle \term,\,\sigma''\rangle$};
\node [] (dummy) [below= 1.2 cm of init] {$\vdots$};
   
  \path [] 
      (init) edge [decorate,decoration={snake, post length=2mm}] (termp)
      (init) edge [decorate,decoration={snake, post length=2mm}] (termp')
(init) edge [decorate,decoration={snake, post length=2mm}] (dummy)
  ;
\end{tikzpicture}

%% file: pics/cond11.tex
\begin{tikzpicture}[->,>=stealth',shorten >=1pt,node distance=2.7cm,semithick,minimum size=1cm]
\tikzstyle{every state}=[draw=none]
  \draw[white, use as bounding box] (-2.8,-2.2) rectangle (5.25,.2);
   \node [state, initial, initial text=,label={[yshift=1.8cm, gray] 270:$0$}] (init) {$\langle \Cond G P Q~ G,\, \sigma\rangle$};  
   \node [state,label={[yshift=.6cm, gray] 260:$0$}] (p) [right=1.8 cm of init] {$\langle P,\,\sigma\rangle$};
    \node [] (dummy) [below =.7cm of p] {$\vdots$};
  \path [] 
      (init) edge [] (p)
      (p) edge [decorate,decoration={snake, post length=2mm}] (dummy)
  ;
\end{tikzpicture}

%% file: pics/cond12.tex
\begin{tikzpicture}[->,>=stealth',shorten >=1pt,node distance=2.7cm,semithick,minimum size=1cm]
\tikzstyle{every state}=[draw=none]
  \draw[white, use as bounding box] (-1.3,-.6) rectangle (2.2,.2);
      \node [state, initial, initial text=,label={[yshift=.7cm, gray] 270:$0$}] (p) [] {$\langle P,\,\sigma\rangle$};
    \node [] (dummy) [right =.7cm of p] {$\ldots$};
  \path [] 
      (p) edge [decorate,decoration={snake, post length=2mm}] (dummy)
  ;
\end{tikzpicture}

%% file: pics/cond21.tex
\begin{tikzpicture}[->,>=stealth',shorten >=1pt,node distance=2.7cm,semithick,minimum size=1cm]
\tikzstyle{every state}=[draw=none]
  \draw[white, use as bounding box] (-2.8,-2.2) rectangle (5.3,.2);
   \node [state, initial, initial text=,label={[yshift=1.8cm, gray] 270:$0$}] (init) {$\langle \Cond G P Q~ G,\, \sigma\rangle$};  
   \node [state,label={[yshift=.6cm, gray] 260:$0$}] (q) [right=1.8 cm of init] {$\langle Q,\,\sigma\rangle$};
    \node [] (dummy) [below =.7cm of q] {$\vdots$};
  \path [] 
      (init) edge [] (q)
      (q) edge [decorate,decoration={snake, post length=2mm}] (dummy)
  ;
\end{tikzpicture}

%% file: pics/cond22.tex
\begin{tikzpicture}[->,>=stealth',shorten >=1pt,node distance=2.7cm,semithick,minimum size=1cm]
\tikzstyle{every state}=[draw=none]
  \draw[white, use as bounding box] (-1.3,-.65) rectangle (2.2,.2);
      \node [state, initial, initial text=,label={[yshift=.7cm, gray] 270:$0$}] (q) [] {$\langle Q,\,\sigma\rangle$};
    \node [] (dummy) [right =.7cm of q] {$\ldots$};
  \path [] 
      (q) edge [decorate,decoration={snake, post length=2mm}] (dummy)
  ;
\end{tikzpicture}

%% file: pics/prob1.tex
\begin{tikzpicture}[->,>=stealth',shorten >=1pt,node distance=2.7cm,semithick,minimum size=1cm]
\tikzstyle{every state}=[draw=none]
  \draw[white, use as bounding box] (-1.9,-2.65) rectangle (5.2,.65);
   \node [state, initial, initial text=,label={[yshift=1.2cm, gray] 270:$0$}] (init) {$\langle \PChoice P p Q, \sigma\rangle$};  
   
   \node [state,label={[yshift=.7cm, gray] 270:$0$}] (p) [on grid, right=3 cm of init] {$\langle P,\,\sigma\rangle$};
      \node [state,label={[yshift=.7cm, gray] 270:$0$}] (q) [on grid, below  =2 cm of p] {$\langle Q,\,\sigma\rangle$};
      
    \node [] (dummy1) [on grid, right =2cm of p] {$\ldots$};
    \node [] (dummy2) [on grid, right =2cm of q] {$\ldots$};
  \path [] 
      (init) edge [] node [above] {\scriptsize{$p$}} (p)
      (init) edge [] node [above] {\scriptsize{$1-p$}} (q)
      (p) edge [decorate,decoration={snake, post length=2mm}] (dummy1)
      (q) edge [decorate,decoration={snake, post length=2mm}] (dummy2)
  ;
\end{tikzpicture}

%% file: pics/prob2.tex
\begin{tikzpicture}[->,>=stealth',shorten >=1pt,node distance=2.7cm,semithick,minimum size=1cm]
\tikzstyle{every state}=[draw=none]
  \draw[white, use as bounding box] (-1.9,-2.65) rectangle (6.45,.6);
   \node [state, initial, initial text=,label={[yshift=1.1cm, gray] 270:$0$}] (init) {$\langle \PChoice P p Q, \sigma\rangle$};  
   
   \node [state,label={[yshift=.8cm, gray] 280:$\ExpRew{\Rdtmc{\sigma}{f}{P}}{\lozenge \sinklabel}$}] (p) [on grid, right=3 cm of init] {$\langle P,\,\sigma\rangle$};
      \node [state,label={[yshift=.8cm, gray] 280:$\ExpRew{\Rdtmc{\sigma}{f}{Q}}{\lozenge \sinklabel}$}] (q) [on grid, below  =2 cm of p] {$\langle Q,\,\sigma\rangle$};
      
  \path [] 
      (init) edge [] node [above] {\scriptsize{$p$}} (p)
      (init) edge [] node [above] {\scriptsize{$1-p$}} (q)
  ;
\end{tikzpicture}

%% file: pics/libconcat1.tex
\begin{tikzpicture}[->,>=stealth',shorten >=1pt,node distance=2.7cm,semithick,minimum size=1cm, inner sep=.5mm]
\tikzstyle{every state}=[draw=none]
  \draw[white, use as bounding box] (-1.3,-2.5) rectangle (8.6,2.7);
   \node [state, initial, initial text=,label={[yshift=.9cm, gray] 270:$0$}] (init) {$\langle P;\, Q,\, \sigma\rangle$};  
   \node [state,label={[yshift=.9cm, gray] 270:$0$}] (termp) [on grid, right=2.5 cm of init] {$\langle \term;\,Q,\,\sigma'\rangle$};
      \node [state,label={[yshift=.7cm, gray] 270:$0$}] (q) [on grid, right=2 cm of termp] {$\langle Q,\,\sigma'\rangle$};
      \node [state,label={[yshift=.9cm, gray] 270:$0$}] (termp') [on grid, below=1.8 cm of termp] {$\langle \term;\,Q,\,\sigma''\rangle$};
       \node [state,label={[yshift=.7cm, gray] 270:$0$}] (q') [on grid, right=2 cm of termp'] {$\langle Q,\,\sigma''\rangle$};
       \node [state] (diverge1) [above=.7cm of termp] {$\mathpzc{diverge}$\ldots};
       \node [state] (diverge2) [above right=1cm of dummy1] {$\mathpzc{diverge}$\ldots};
   
    \node [] (dummy1) [on grid, right =2cm of q] {$\ldots$};
    \node [] (dummy2) [on grid, right =2cm of q'] {$\ldots$};
  \path [] 
      (init) edge [decorate,decoration={snake, post length=2mm}] (termp)
      (init) edge [decorate,decoration={snake, post length=2mm}] (termp')
      (termp) edge [] (q)
      (termp') edge [] (q')
      (q) edge [decorate,decoration={snake, post length=2mm}] (dummy1)
      (q') edge [decorate,decoration={snake, post length=2mm}] (dummy2)
      (init) edge [decorate,decoration={snake, post length=2mm}] (diverge1)
      (diverge1) edge [decorate,decoration={snake, post length=2mm}, loop right] (diverge1)
      (q) edge [decorate,decoration={snake, post length=2mm}] (diverge2)
      (diverge2) edge [decorate,decoration={snake, post length=2mm}, loop right] (diverge2)
  ;
\end{tikzpicture}

%% file: pics/libconcat2.tex
\begin{tikzpicture}[->,>=stealth',shorten >=1pt,node distance=2.7cm,semithick,minimum size=1cm]
\tikzstyle{every state}=[draw=none]
  \draw[white, use as bounding box] (-1.15,-2.7) rectangle (6.3,1.95);
   \node [state, initial, initial text=,label={[yshift=.4cm, gray] 270:$0$}] (init) {$\langle P,\, \sigma\rangle$};  
   \node [state,label={[yshift=.6cm, gray] 280:$\LExpRew{\Rdtmc{\sigma'}{f}{Q}}{\lozenge \sinklabel}$}] (termp) [on grid, right=2.5 cm of init] {$\langle \term,\,\sigma'\rangle$};
      \node [state,label={[yshift=.6cm, gray] 280:$\LExpRew{\Rdtmc{\sigma''}{f}{Q}}{\lozenge \sinklabel}$}] (termp') [on grid, below=1.8 cm of termp] {$\langle \term,\,\sigma''\rangle$};
\node [state] (diverge1) [above=0cm of termp] {$\mathpzc{diverge}$\ldots};
   
  \path [] 
      (init) edge [decorate,decoration={snake, post length=2mm}] (termp)
      (init) edge [decorate,decoration={snake, post length=2mm}] (termp')
(diverge1) edge [decorate,decoration={snake, post length=2mm}, loop right] (diverge1)
      (init) edge [decorate,decoration={snake, post length=2mm}] (diverge1)
  ;
\end{tikzpicture}

%% file: appapp.tex
\subsection{Proof of  Theorem~\ref{thm:prog-trans-sound}}
\label{sec:proof-hoist-trans}

\begingroup
\def\thelemma{\ref{thm:prog-trans-sound}}
\begin{theorem}[Program Transformation Correctness]
  Let $P \in \cfpGCL$ admit at least one feasible run for every initial state and $\Tr (P,\CteFun{1}) = (\hat{P},\,\hat{h})$. Then for any $f 
  \in \Ex$ and $g \in \BEx$,
\begin{align*}
\wp[\hat{P}](f) = \qcwp[P](f) 
\quad \text{and} \quad
\wlp[\hat{P}](g) = \qcwlp[P](g).
\end{align*}
\end{theorem}
\addtocounter{lemma}{-1}
\endgroup
\noindent
In view of Theorem~\ref{thm:cwp-decuop}, the proof reduces to showing equations $\hat{h} \cdot \wp[\hat{P}] (f) = \wp[P](f)$, $\hat{h} \cdot \wlp[\hat{P}] (f) = \wlp[P](f)$ and $\hat{h} = \wlp[P](\CteFun{1})$, which follow immediately from the auxiliary Lemma~\ref{thm:prog-trans-sound-strong} below by taking $h=\CteFun{1}$.

\begin{lemma}\label{thm:prog-trans-sound-strong}
  Let $P \in \cfpGCL$. Then for all expectations $f \in \Ex$ and $g, h \in \BEx$,
  it holds
  \begin{gather}
  \hat{h} \cdot \wp[\hat{P}] (f)  = \wp[P](h \cdot  f) \label{eq:1} \\ 
  \hat{h} \cdot \wlp[\hat{P}] (g)  = \wlp[P](h \cdot  g) \label{eq:2} \\
  \hat{h} =  \wlp[P](h), \label{eq:3}
  \end{gather}
  where $(\hat{P}, \hat{h}) = \Tr(P,\,h)$.
\end{lemma} 
\begin{proof}
  We prove only equations \eqref{eq:1} and \eqref{eq:3} since \eqref{eq:2} follows a reasoning similar to \eqref{eq:1}. The proof proceeds by induction on the structure of $P$. In the remainder we will refer to the inductive hypothesis about \eqref{eq:1} as to IH$_1$ and to the inductive hypothesis about \eqref{eq:3} as to IH$_2$.

\textbf{The Effectless Program \Skip.} We have $\Tr(\Skip,h) = (\Skip,
  h)$ and the statement follows immediately since 
\begin{gather*}
h \cdot \wp[\Skip](f) = h \cdot f = \wp[\Skip](h \cdot
f)
\intertext{and}
h = \wlp[\Skip](h).
\end{gather*}
\textbf{The Faulty Program \Abort.} We have $\Tr(\Abort,h) =
  (\Abort, \CteFun{1})$ and the statement follows immediately since 
\begin{gather*}
\CteFun{1} \cdot \wp[\Abort](f) =  \CteFun{1} \cdot \CteFun{0} = \wp[\Abort](h \cdot
f) 
\intertext{and} 
\CteFun{1} = \wlp[\Abort](h).
\end{gather*}
\textbf{The Assignment $\Ass x E$.}  We have $\Tr(\Ass{x}{E},h) =
  (\Ass{x}{E}, h\subst{x}{E})$ and the statement follows immediately since 
\begin{gather*}
h\subst{x}{E} \cdot \wp[\Ass{x}{E}] (f)
= 
h\subst{x}{E} \cdot f\subst{x}{E}\\
=
(h \cdot f) \subst{x}{E}
=
\wp[\Ass{x}{E}](h \cdot f)
\intertext{and}
h\subst{x}{E} = \wlp[\Ass{x}{E}](h).
\end{gather*}
\textbf{The Observation $\Observe~ G$.} We have $\Tr(\Observe \:
  G,h) \allowbreak=\allowbreak (\Skip, \ToExp{G} \cdot h)$ and the statement follows
  immediately since
\begin{gather*}
\ToExp{G} \cdot h \cdot \wp[\Skip] (f) 
= 
\ToExp{G} \cdot h \cdot f \\
= 
\wp[\Observe \: G](h \cdot f)
\intertext{and}
\ToExp{G} \cdot h = \wlp[\Observe \: G](h).
\end{gather*}
\textbf{The Concatenation $P;\, Q$.}  Let $(\hat{Q},\hat{h}_Q) = \Tr(Q,
  h)$ and $(\hat{P},\hat{h}_P) = \Tr(P, \hat{h}_Q)$. In view
  of these definitions, we obtain
\[
\Tr(P;Q, h) = (\hat{P};\hat{Q}, \hat{h}_P).
\]
Now 
\begin{align*}
\MoveEqLeft[3] 
\hat{h}_P \cdot \wp[\hat{P};\hat{Q}] (f)  & &\\
& = 
\hat{h}_P \cdot \wp[\hat{P}] \left( \wp[\hat{Q}] (f) \right)  & & \\
& =
\wp[P](\hat{h}_Q \cdot \wp[\hat{Q}] (f) )  & & 
\text{(IH$_1$ on $P$)}\\
& =
\wp[P](\wp[Q] (h \cdot f))  & & 
\text{(IH$_1$ on $Q$)}\\
& =
\wp[P;Q](h \cdot f)  & & 
\intertext{and}
\hat{h}_P 
& = 
\wlp[P](\hat{h}_Q)  & & 
\text{(IH$_2$ on $P$)}\\
& =
\wlp[P](\wlp[Q](h))  & & 
(\text{IH$_2$ on $Q$)}\\
& =
\wlp[P;Q](h).  & & 
\end{align*}

\textbf{The Conditional Choice $\Cond G P Q$.} Let $(\hat{P},\hat{h}_P) \allowbreak=\allowbreak \Tr(P, h)$
  and $(\hat{Q},\hat{h}_Q) = \Tr(Q, h)$. On view of these
  definitions, we obtain
\begin{gather*}
\Tr(\Cond{G}{P}{Q}, h) = \\
 (\Cond{G}{\hat{P}}{\hat{Q}},  \ToExp{G}\cdot \hat{h}_P +
 \ToExp{\lnot G} \cdot \hat{h}_Q).
\end{gather*}
Now 
\begin{align*}
\MoveEqLeft[1]
(\ToExp{G}\cdot \hat{h}_P +
 \ToExp{\lnot G} \cdot \hat{h}_Q) \\ 
 & \qquad \cdot \wp[\Cond{G}{\hat{P}}{\hat{Q}}] (f) & & \\
& =
(\ToExp{G}\cdot \hat{h}_P +
 \ToExp{\lnot G} \cdot \hat{h}_Q) && \\ 
& \qquad \cdot  
(\ToExp{G}\cdot \wp[\hat{P}](f) +
 \ToExp{\lnot G} \cdot \wp[\hat{Q}](f))  & & \\
& =
\ToExp{G}\cdot \hat{h}_P \cdot \wp[\hat{P}](f) +
 \ToExp{\lnot G} \cdot \hat{h}_Q \cdot \wp[\hat{Q}](f)  & & \\
& =
\ToExp{G} \cdot \wp[P](h \cdot f) +
 \ToExp{\lnot G} \cdot \wp[Q](h \cdot f)  & & 
\text{(IH$_1$)} \\
& =
\wp[\Cond{G}{P}{Q}] (h \cdot f) & &
\intertext{and}
\MoveEqLeft[1]
\ToExp{G}\cdot \hat{h}_P +
 \ToExp{\lnot G} \cdot \hat{h}_Q & & \\
& =
\ToExp{G}\cdot \wlp[P](h) +
 \ToExp{\lnot G} \cdot \wlp[Q](h)  & & 
\text{(IH$_2$)} \\ 
& =
\wlp[\Cond{G}{P}{Q}] (h) & &
\end{align*}
\textbf{The Probabilistic Choice $\PChoice P p Q$.} Let $(\hat{P},\,\hat{h}_P) =\Tr(P,\, h)$
  and $(\hat{Q},\allowbreak\,\hat{h}_Q) \allowbreak=\allowbreak \Tr(Q,\, h)$. On view of these
  definitions, we obtain 
\begin{align*}
\MoveEqLeft[2]
\Tr(\PChoice{P}{\phi}{Q}, h) = \\
 & (\PChoice{\hat{P}}{\nicefrac{\phi \cdot \hat{h}_P}{\hat{h}}}{\hat{Q}}, \phi \cdot \hat{h}_P +
(\CteFun{1}-\phi ) \cdot \hat{h}_Q) 
\end{align*}
with $\hat{h} = \phi \cdot \hat{h}_P +
(\CteFun{1}-\phi) \cdot \hat{h}_Q$.

To prove the first claim
\[
\hat{h} \cdot \wp[\PChoice{\hat{P}}{\nicefrac{\phi \cdot
    \hat{h}_P}{\hat{h}}}{\hat{Q}}](f)
=
\wp[\PChoice{P}{\phi}{Q}](h \cdot f)
\]
of the lemma we need to make a case distinction between those states that are
mapped by $\hat{h}$ to a positive number and those that are mapped to
$0$. In the first case, \ie if $\hat{h}(s) > 0$, we reason as follows:
\begin{align*}
\MoveEqLeft[1]
\hat{h} (s) \cdot 
\wp[\PChoice{\hat{P}}{\nicefrac{\phi \cdot \hat{h}_P}{\hat{h}}}{\hat{Q}}]
(f) (s) & & \\
& =
\hat{h} (s)  \cdot \left(
 \tfrac{\phi \cdot \hat{h}_P}{\hat{h}} (s) \cdot \wp[\hat{P}] (f) (s) \right. &&  \\
 &  \qquad \quad \left. + \tfrac{(\CteFun{1}-\phi) \cdot \hat{h}_Q}{\hat{h}} (s) \cdot
 \wp[\hat{Q}](f) (s) \right)  & & \\
& =
\phi (s) \cdot \hat{h}_P (s) \cdot \wp[\hat{P}] (f) (s) && \\
& \qquad  + (\CteFun{1}-\phi) (s) \cdot \hat{h}_Q (s) \cdot \wp[\hat{Q}](f) (s) & & \\
& =
\phi (s) \cdot \wp[P] (h \cdot f) (s) && \\
& \qquad + (\CteFun{1}-\phi)(s) \cdot \wp[Q](h \cdot f) (s) & & 
\text{(IH$_1$)} \\
& =
\wp[\PChoice{P}{\phi}{Q}](h \cdot f) (s) & &
\end{align*}
while in the second case, \ie if $\hat{h}(s) = 0$, the claim holds because
we will have $\wp[\PChoice{P}{\phi}{Q}](h \cdot f)(s)=0$. To see
this note that if $\hat{h}(s) = 0$ then either $\phi(s)=0 \land
\hat{h}_Q(s)=0$ or $\phi(s)=1 \land \hat{h}_P(s)=0$ holds. Now
assume we are in the first case (an analogous argument works for the other
case); using the IH$_1$ over $Q$ we obtain
\begin{gather*}
\wp[\PChoice{P}{0}{Q}](h \cdot f)(s) 
= 
\wp[Q](h \cdot f)(s) \\
=
\hat{h}_Q(s) \cdot \wp[Q](f)(s)
=
0. 
\end{gather*}

The proof of the second claim of the lemma is straightforward:
\begin{align*}
\MoveEqLeft[1]
\phi \cdot \hat{h}_P +
(\CteFun{1}-\phi ) \cdot \hat{h}_Q & & \\
& =
\phi \cdot \wlp[P] (h) +
(\CteFun{1}-\phi) \cdot \wlp[Q](h) & & 
\text{(IH$_2$)} \\
& =
\wlp[\PChoice{P}{\phi}{Q}](h). & &
\end{align*}
\textbf{The Loop $\WhileDo G Q$.} Let $\hat{h} = \gfp F$ where $F(X) = \ToExp{G} \cdot
  \Tr_P(X) + \ToExp{\lnot G} \cdot h$ and $\Tr_P(\cdot)$ is a short--hand for $\pi_2 \circ T(P, \cdot)$. Now if we let $(\hat{P}, \theta) =
  \Tr(P,\hat{h})$ by definition of $\Tr$ we obtain 
\[
\Tr(\WhileDo{G}{P}, h) = (\WhileDo{G}{\hat{P}}, \hat{h}).
\]
The first claim of the lemma says that 
\[
\hat{h} \cdot \wp[\WhileDo{G}{\hat{P}}] (f) 
= 
\wp[\WhileDo{G}{P}](h \cdot f).
\]
Now if we let $H(X) = \ToExp{G} \cdot \wp[\hat{P}] (X) + \ToExp{\lnot G} \cdot f$ and $I(X) = \ToExp{G} \cdot \wp[P] (X) + \ToExp{\lnot G} \cdot h \cdot f$, the claim can be rewritten as $\hat{h} \cdot \lfp H = \lfp I$ and a straightforward argument using the Kleene fixed point theorem (and the continuity of \wp established in Lemma~\ref{thm:wp-ext-is-cont}) shows that it is entailed by formula $\forall n \mydot \hat{h} \cdot H^n(\CteFun{0}) = I^n(\CteFun{0})$. We prove the formula by induction on $n$. The case $n=0$ is trivial. For the inductive case we reason as follows:
\begin{align*}
\MoveEqLeft[1]
\hat{h} \cdot H^{n+1}(\CteFun{0}) & & \\
& =
F(\hat{h}) \cdot H^{n+1}(\CteFun{0}) & & 
\text{(def.~$\hat{h}$)} \\
& =
(\ToExp{G} \cdot
  \Tr_P(\hat{h}) + \ToExp{\lnot G} \cdot h) \cdot H^{n+1}(\CteFun{0}) & & 
\text{(def.~$F$)} \\
& =
(\ToExp{G} \cdot
  \Tr_P(\hat{h}) + \ToExp{\lnot G} \cdot h)  & & \\
& \qquad \cdot (\ToExp{G} \cdot
  \wp[\hat{P}] (H^n(\CteFun{0})) + \ToExp{\lnot G} \cdot f) & & 
\text{(def.~$H$)} \\
& =
\ToExp{G} \cdot \Tr_P(\hat{h}) \cdot \wp[\hat{P}] (H^n(\CteFun{0})) && \\
& \qquad  + \ToExp{\lnot G} \cdot h \cdot f & &
\text{(algebra)} \\
& =
\ToExp{G} \cdot \theta \cdot \wp[\hat{P}] (H^n(\CteFun{0})) + 
\ToExp{\lnot G} \cdot h \cdot f & & 
\text{(def.~$\theta$)} \\
& =
\ToExp{G} \cdot \wp[P](\hat{h} \cdot H^n(\CteFun{0})) + 
\ToExp{\lnot G} \cdot h \cdot f & & 
\text{(IH$_1$ on P)} \\
& =
I(\hat{h} \cdot H^n(\CteFun{0})) & & 
\text{(def.~$I$)} \\
& =
I^{n+1}(\CteFun{0}) & &
\text{(IH on $n$)}
\end{align*}

We now turn to proving the second claim
\[
\hat{h}  
= 
\wlp[\WhileDo{G}{P}](h)
\]
of the lemma. By letting $J(X) = \ToExp{G} \cdot \wlp[P] (X) +
\ToExp{\lnot G} \cdot h$, the claim reduces to $\gfp F = \gfp J$, which we
prove showing that $\hat{h} = \gfp F$ is a fixed point of $J$ and $\gfp J$
is a fixed point of $F$. (These assertions basically imply that $\gfp F \geq \gfp
J$ and $\gfp J \geq \gfp F$, respectively.)
\begin{align*}
J (\hat{h}) 
& = 
\ToExp{G} \cdot \wlp[P] (\hat{h}) + \ToExp{\lnot G} \cdot
h  & & 
\text{(def.~$J$)} \\
& =
\ToExp{G} \cdot \theta + \ToExp{\lnot G} \cdot
h  & & 
\text{(IH$_2$ on $P$)} \\
& =
\ToExp{G} \cdot \Tr_P(\hat{h}) + \ToExp{\lnot G} \cdot
h  & & 
\text{(def.~$\theta$)} \\
& =
F(\hat{h})  & & 
\text{(def.~$F$)} \\
& =
\hat{h}  & & 
\text{(def.~$\hat{h}$)} \\
& & & \\
F (\gfp J) 
& = 
\ToExp{G} \cdot \Tr_P(\gfp J) + \ToExp{\lnot G} \cdot h & &
\text{(def.~$F$)} \\
& = 
\ToExp{G} \cdot \wlp[P](\gfp J) + \ToExp{\lnot G} \cdot h & &
\text{(IH$_2$ on $P$)} \\
& = 
J (\gfp J) & & 
\text{(def.~$J$)} \\
& = 
\gfp J & & 
\text{(def.~$\gfp J$)} 
\end{align*}
\end{proof}

\subsection{Proof of Theorem \ref{thm:sugar}}
\label{sec:proof-sugar}


\begin{proof}
Let us take the operational point of view. Let $\sinit$ be some initial state of $P$.
\begin{align}
{}&\qcwp[P](f)(\sinit)\\
={}&\CExpRew{\Rdtmc{\sinit}{f}{P}}{\Finally \sinklabel}{\neg\Finally\bad}\\
={}&\CExpRew{\Rdtmc{\sinit}{f}{P'}}{\Finally \sinklabel}{\neg\Finally\continue}
\label{eqn:line3}
\\
={}&\frac{
	\ExpRew {\Rdtmc{\sinit}{f}{P'}} {\Finally \sinklabel \cap \neg\Finally \continue}
	}{
	\Pr^{\Rdtmc{\sinit}{f}{P'}}(\neg\Finally \continue)
	}
\\
={}&\frac{
	\sum_{\fpath\in\Finally \sinklabel \cap \neg\Finally \continue}\Pr^{\Rdtmc{\sinit}{f}{P'}}(\fpath) \cdot f(\fpath)
	}{
	1-\Pr^{\Rdtmc{\sinit}{f}{P'}}(\Finally \continue)
	}
\\
={}&
	\sum_{i=0}^\infty \Pr^{\Rdtmc{\sinit}{f}{P'}}(\Finally \continue)^i \nonumber\\
{}&\qquad\qquad
	\cdot\sum_{\fpath\in\Finally \sinklabel \cap \neg\Finally \continue}\Pr^{\Rdtmc{\sinit}{f}{P'}}(\fpath) \cdot f(\fpath)
	\label{eqn:line6}
	\\
={}&
	\sum_{\fpath\in\Finally \sinklabel \cap \neg\Finally \continue}\sum_{i=0}^\infty 
		\left(\Pr^{\Rdtmc{\sinit}{f}{P'}}(\Finally \continue)^i \right.\nonumber\\
{}&\qquad\qquad\qquad\qquad\qquad
	\left.\cdot\Pr^{\Rdtmc{\sinit}{f}{P'}}(\fpath) \cdot f(\fpath)\right)
	\label{eqn:line7}
	\\
={}& \sum_{\fpath\in\Finally \sinklabel}\Pr^{\Rdtmc{\sinit}{f}{P''}}(\fpath) \cdot f(\fpath)\\
={}& \ExpRew{\Rdtmc{\sinit}{f}{P''}}{\Finally \sinklabel}\\
={}& \wp(P'',f)(\sinit)\enspace.
\end{align}
The equality (\ref{eqn:line3}) holds because, by construction, the probability to violate an observation in $P$ agrees with the probability to reach a state in $P'$ where \continue is \true. 
In order to obtain equation (\ref{eqn:line6}) we use the fact that for a fixed real value $r$ and probability $a$ it holds
\[\frac{r}{1-a} = \sum_{i=0}^\infty a^i r\enspace.\]
Rewriting (\ref{eqn:line6}) into (\ref{eqn:line7}) precisely captures the expected cumulative reward of all terminating paths in $P''$ which is the expression in the following line. Finally we return from the operational semantics to the denotational semantics and obtain the desired result.
\end{proof}

%% file: crowdsapp.tex
\subsection{Detailed calculations for Section~\ref{sec:crowds}}
\label{sec:calculation}
We refer to the labels \Init and \Loop introduced in the program $P$ in Section~\ref{sec:crowds}.
Further let \Body denote the program in the loop's body.
For readability we abbreviate the variable names \textit{delivered} as \delivered, \textit{counter} as \counter and \textit{intercepted} as \intercepted.
In the following we consider \delivered and \intercepted as boolean variables.
In order to determine (\ref{eqn:goal}) we first start with the numerator.
This quantity is given by
\begin{align}
{}&\wp[\Init;\Loop;\Observe( \counter\leq k)]([\neg \intercepted])\\
\label{eqn:defwp}
={}&\wp[\Init](\wp[\Loop]([\counter\leq k \wedge \neg \intercepted]))\\
\label{eqn:defwhile}
={}&\wp[\Init](\lfp\! F \mydot\left( [\neg\delivered]\cdot\wp[\Body](F) \right.\nonumber\\
&\qquad\qquad\left.+[\delivered \wedge \counter\leq k\wedge\neg\intercepted]\right)) \\
\label{eqn:kleene}
={}&\wp[\Init](\sup_n\left( [\neg\delivered]\cdot\wp[\Body](\CteFun{0}) \right.\nonumber\\
&\qquad\qquad\left.+[\delivered \wedge \counter\leq k\wedge\neg\intercepted]\right)^n)
\end{align}
where $\Phi^n$ denotes the $n$-fold application of $\Phi$.
Equation (\ref{eqn:defwp}) is given directly by the semantics of sequential composition of \cpGCL commands.
In the next line we apply the definition of loop semantics in terms of the least fixed point.
Finally, (\ref{eqn:kleene}) is given by the Kleene fixed point theorem as a solution to the fixed point equation in (\ref{eqn:defwhile}).
We can explicitly find the supremum by considering the expression for several $n$ and deducing a pattern. Let $\Phi(F) = [\neg\delivered]\cdot\wp[\Body](F)+[\delivered \wedge \counter\leq k\wedge\neg\intercepted]$. Then we have

\begingroup
\allowdisplaybreaks
\begin{align*}
\Phi(\CteFun{0})&=[\neg\delivered]\cdot\wp[\Body](\CteFun{0})+[\delivered \wedge \counter\leq k\wedge\neg\intercepted]\\
&=[\delivered \wedge \counter\leq k\wedge\neg\intercepted] \\
&\\
\Phi^2(\CteFun{0}) 
&= \Phi([\delivered \wedge \counter\leq k\wedge\neg\intercepted])\\
&= [\neg\delivered]\cdot\wp[\Body]([\delivered \wedge \counter\leq k\wedge\neg\intercepted])\\
&\qquad+[\delivered \wedge \counter\leq k\wedge\neg\intercepted]\\
&= [\neg\delivered]\cdot\left(p(1-c)\cdot[\delivered \wedge \counter+1\leq k\wedge\neg\intercepted]\right.\\
&\qquad\qquad\left.+(1-p)\cdot[\counter\leq k\wedge\neg\intercepted]\right)\\
&\qquad+[\delivered \wedge \counter\leq k\wedge\neg\intercepted]\\
&= [\neg\delivered\wedge \counter\leq k\wedge\neg\intercepted]\cdot(1-p)\\
&\qquad+[\delivered \wedge \counter\leq k\wedge\neg\intercepted]\\
&\\
\Phi^3(\CteFun{0}) 
&= \Phi([\neg\delivered\wedge \counter\leq k\wedge\neg\intercepted]\cdot(1-p)\\
&\quad+[\delivered \wedge \counter\leq k\wedge\neg\intercepted]) \\
&=\ldots \\
&= [\neg\delivered\wedge \counter\leq k\wedge\neg\intercepted]\cdot(1-p)\\
&\quad +[\neg\delivered\wedge \counter+1\leq k\wedge\neg\intercepted]\cdot(1-p)p(1-c)\\
&\quad +[\delivered \wedge \counter\leq k\wedge\neg\intercepted]
\end{align*}
\endgroup
As we continue to compute $\Phi^n(\CteFun{0})$ in each step we add a summand of the form
\[
[\neg\delivered\wedge \counter+i\leq k\wedge\neg\intercepted]\cdot(1-p)(p(1-c))^i
\]
However we see that the predicate evaluates to \false for all $i > k -\counter$.
Hence the non-zero part of the fixed point is given by
\begin{align*}
{}&[\delivered \wedge \counter\leq k\wedge\neg\intercepted]\\
{}&+\sum_{i=0}^{k-\counter}[\neg\delivered\wedge \counter+i\leq k\wedge\neg\intercepted]\cdot(1-p)(p(1-c))^i\\
={}&[\delivered \wedge \counter\leq k\wedge\neg\intercepted]\\
&+[\neg\delivered\wedge \counter\leq k\wedge\neg\intercepted]\cdot\sum_{i=0}^{k-\counter}(1-p)(p(1-c))^i\\
={}&[\delivered \wedge \counter\leq k\wedge\neg\intercepted]\\
&+[\neg\delivered\wedge \counter\leq k\wedge\neg\intercepted]\\
&\qquad\cdot(1-p)\frac{1-(p(1-c))^{k-\counter+1}}{1-p(1-c)}\enspace.
\end{align*}
where for the last equation we use a property of the finite geometric series, namely that for $r\neq 1$
\[\sum_{k=0}^{n-1} ar^k= a \, \frac{1-r^{n}}{1-r}\enspace.\]
The result coincides with the intuition that in a state where $\delivered = \false$, the probability to fail to reach the goal $\neg\intercepted \wedge \counter \leq k$ is distributed geometrically with probability $p(1-c)$. 
It is easy to verify that our educated guess is correct by checking that we indeed found a fixed point of $\Phi$:
\begin{align*}
&\Phi([\delivered \wedge \counter\leq k\wedge\neg\intercepted]\\
&+[\neg\delivered\wedge \counter\leq k\wedge\neg\intercepted]\\
&\quad\cdot(1-p)\frac{1-(p(1-c))^{k-\counter+1}}{1-p(1-c)})\\
={}&[\delivered \wedge \counter \leq k \wedge \neg \intercepted]\\
& +[\neg\delivered]\cdot\left((1-p)\cdot[\counter\leq k \wedge \neg\intercepted] \phantom{\frac{a^k}{b}}\right.\\
&\quad +p(1-c)\left([\delivered \wedge \counter+1 \leq k \wedge \neg \intercepted] \phantom{\frac{a^k}{b}}\right.\\
&\quad\qquad +[\neg \delivered \wedge \counter+1 \leq k \wedge \neg \intercepted]\\
&\quad\qquad\qquad\left.\left.\cdot(1-p)\frac{1-p(1-c)^{k-\counter}}{1-p(1-c)}\right)\right) \\
={}&[\delivered \wedge \counter \leq k \wedge \neg \intercepted]\\
& +[\neg\delivered\wedge\counter\leq k \wedge \neg\intercepted]\cdot(1-p)\\
& +[\neg\delivered\wedge\counter+1 \leq k \wedge \neg\intercepted]\\
& \qquad\cdot(1-p)(p(1-c))\frac{1-p(1-c)^{k-\counter}}{1-p(1-c)} \\
={}&[\delivered \wedge \counter \leq k \wedge \neg \intercepted]\\
& +[\neg\delivered\wedge\counter = k \wedge \neg\intercepted]\cdot(1-p)\\
& +[\neg\delivered\wedge\counter+1 \leq k \wedge \neg\intercepted]\cdot(1-p)\\
& +[\neg\delivered\wedge\counter+1 \leq k \wedge \neg\intercepted]\\
& \qquad\cdot(1-p)(p(1-c))\frac{1-p(1-c)^{k-\counter}}{1-p(1-c)} \\
={}&[\delivered \wedge \counter \leq k \wedge \neg \intercepted]\\
& +[\neg\delivered\wedge\counter = k \wedge \neg\intercepted]\cdot(1-p)\\
& +[\neg\delivered\wedge\counter+1 \leq k \wedge \neg\intercepted]\\
&\cdot(1-p)\frac{1-p(1-c)+(p(1-c))\left(1-p(1-c)^{k-\counter}\right)}{1-p(1-c)} \\
={}&[\delivered \wedge \counter \leq k \wedge \neg \intercepted]\\
& +[\neg\delivered\wedge\counter\leq k \wedge \neg\intercepted]\\
&\quad\cdot(1-p)\frac{1-p(1-c)^{k-\counter+1}}{1-p(1-c)}
\end{align*}
Moreover this fixed point is the only fixed point and therefore the least.
The justification is given by~\cite{McIver:2004} where they show that loops which terminate almost surely have only one fixed point.
We can now continue our calculation from (\ref{eqn:kleene}).
\begin{align}
={}&\wp[\Init]([\delivered \wedge \counter\leq k\wedge\neg\intercepted]\nonumber\\
&\quad+[\neg\delivered\wedge \counter\leq k\wedge\neg\intercepted]\\
&\qquad\cdot(1-p)\frac{1-(p(1-c))^{k-\counter+1}}{1-p(1-c)})\nonumber\\
={}&(1-c)(1-p)\frac{1-(p(1-c))^{k}}{1-p(1-c)}\enspace.\label{eqn:wpresult}
\end{align}
This concludes the calculation of the numerator of (\ref{eqn:goal}).
Analogously we find the denominator
\begin{align}
{}&\wlp[\Init;\Loop;\Observe(\counter\leq k)](\CteFun{1}) \nonumber\\
={}&\wlp[\Init](\wlp[\Loop]([\counter\leq k]))\nonumber\\
={}&\wlp[\Init](\gfp\! F \mydot\left( [\neg\delivered]\cdot\wlp[\Body](F) \right.\nonumber\\
&\qquad\qquad\left.+[\delivered \wedge \counter\leq k]\right)) \nonumber\\
={}&\wlp[\Init](\sup_n\left( [\neg\delivered]\cdot\wlp[\Body](\CteFun{1}) \right.\nonumber\\
&\qquad\qquad\left.+[\delivered \wedge \counter\leq k]\right)^n)\nonumber\\
={}&\wlp[\Init]([\delivered \wedge \counter \leq k] \nonumber\\
&+ [\neg \delivered \wedge \counter \leq k]\cdot(1-p^{k-counter+1}))\nonumber\\
={}&1-p^k\enspace.\label{eqn:wlpresult}
\end{align}
The only difference is that here the supremum is taken with respect to the reversed order $\geq$ in which $\CteFun{1}$ is the bottom and $\CteFun{0}$ is the top element.
However as mentioned earlier \Loop terminates with probability one and the notions of \wp and \wlp coincide.
We divide (\ref{eqn:wpresult}) by (\ref{eqn:wlpresult}) to finally arrive at
\begin{align*}
&\qcwp[\textit{P}]([\neg\textit{intercepted}])\\
={}&(1-c)(1-p)\frac{1-(p(1-c))^k}{1-p(1-c)}\cdot\frac{1}{1-p^k}\enspace.
\end{align*}
One can visualise it as a function in $k$ by fixing the parameters $c$ and $p$.
For example, Figure~\ref{fig:plot} shows the conditional probability plotted for various parameter settings.
\begin{figure}
\begin{tikzpicture}[scale=1.0]
\begin{axis}[legend style={at={(0,0)},anchor=north west,at={(axis description cs:0,-0.1)},draw=none,legend columns=2},column sep=1ex,cycle list={blue,mark=o\\ red,mark=x\\ green,mark=triangle\\ orange,mark=square\\ }]
	\foreach \c in {0.1,0.2}{
		\foreach \p in {0.6,0.8}{
    		\addplot+[domain=1:20] {(1-\c)*(1-\p)*(1-(\p*(1-\c))^x)/(1-\p*(1-\c))*(1)/(1-\p^x)};
			\addlegendentryexpanded{$c=\c\quad p=\p$}
		}
	}
\end{axis}
\end{tikzpicture}
\caption{The conditional probability that a message is intercepted as a function of $k$ for fixed $c$ and $p$.}
\label{fig:plot}
\end{figure}
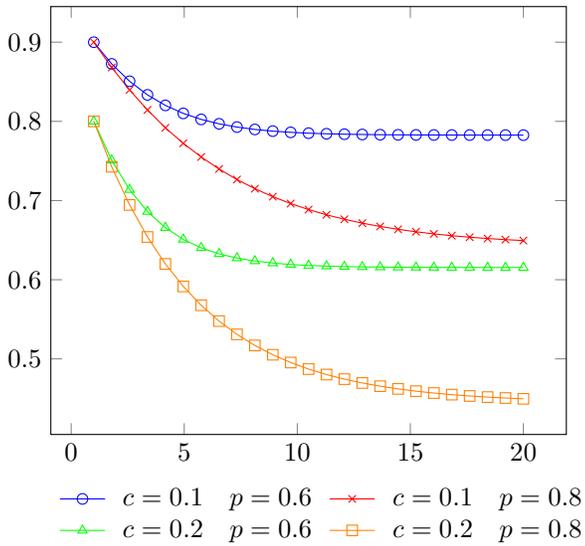